\documentclass{amsart}
\usepackage{fullpage,graphicx,subfigure,mathpazo,color}
\usepackage{amsmath,amscd,tikz,mathrsfs}
\usepackage[normalem]{ulem}
\usepackage{amsmath}
\usepackage{epstopdf}
\usepackage{tikz}
\usepackage{setspace}
\newcommand{\ii}{\mathrm{i}}
\newcommand{\ee}{\mathrm{e}}

\newcommand{\T}{\mathrm{T}}

\newtheorem{rhp}{Riemann-Hilbert Problem}
\newtheorem{theorem}{Theorem}

\newtheorem{definition}{Definition}

\newtheorem{remark}{Remark}

\usepackage{graphicx}      
\usepackage{titlesec}

\titleformat{\section}{\centering\LARGE\bfseries}{\thesection}{1em}{}
\titleformat{\subsection}{\Large\bfseries}{\thesubsection}{1em}{}
\begin{document}

\title{On the large-order asymtptics of Kuznetsov-Ma breathers}

\author{Liming Ling}
\address{School of Mathematics, South China University of Technology, Guangzhou, China 510641}
\email{linglm@scut.edu.cn}
\author{Xiaoen Zhang}
\address{School of Mathematics, South China University of Technology, Guangzhou, China 510641}
\email{zhangxiaoen@scut.edu.cn}

\begin{abstract}
We study the large-order asymptotics for the Kuznetsov-Ma breather of the nonlinear Schr\"{o}dinger equation in the far-field regime. With the aid of Darboux transformation, we first derive the corresponding Riemann-Hilbert representation for the high-order Kuznetsov-Ma breathers. Under the far-field limit, there are five asymptotical regions in the space-time plane where the breathers behave differently, the genus-two region, the algebraic-decay region, and three distinct genus-zero regions. With the aid of the Deift-Zhou nonlinear steepest decent method, we give the leading order term for each region and verify the consistency between the exact solution and the asymptotic solution numerically. Compared to the previous studies about the large-order asymptotic analysis of rogue waves and solitons,
we find a novel genus-two asymptotic region, which further enriches the research of large-order dynamics.


\end{abstract}

\date{\today}

\maketitle

\section{Introduction}
The well-known one-dimensional focusing nonlinear Schr\"{o}dinger (NLS) equation
\begin{equation}\label{eq:nls}
\ii q_t+\frac{1}{2}q_{xx}+\left|q\right|^2q=0,
\end{equation}
is a completely integrable equation and can be studied via the inverse scattering transform \cite{shabat1972exact,novikov1984theory}. For an appropriate initial datum, we can calculate the scattering data explicitly. As to the standard $N$-soliton solutions, the scattering data consist of $N$ distinct first-order poles, and the real and imaginary parts of the poles represent the velocity and amplitude of soliton respectively. If these $N$ distinct poles degenerate into $N$-th order pole, we will get the $N$-th order solitons, each of which shares the same amplitude and velocity and is separated from each other with logarithmic type. These $N$-th order soliton solutions have evident different behaviors from the $N$-soliton solutions. Besides the soliton solutions, under the non-zero background, the NLS equation \eqref{eq:nls} also has rich family solutions, such as the rational rogue wave \cite{peregrine1983water,guo2012nonlinear}, the Kuznetsov-Ma breather (KMB)\cite{kuznetsov1977solitons,ma1979perturbed}, the Akhmediev breather (AB)\cite{akhmediev1986modulation} and the Tajiri-Watanabe breather \cite{tajiri1998breather}.
In general, breathers develop due to the instability of small amplitude chaotic perturbation, and they are also related to the modulational instability and the interference effects between a bright soliton and a plane-wave background \cite{dudley2009modulation,zhao2018mechanism}. Analogously, under the non-zero background, we can also get the high-order rogue waves, KMBs, ABs, Tajiri-Watanabe breathers \cite{kedziora2012second,wang2017generation}, and their mixtures. One direct physical explanation about the high-order breathers is that they are related to the high-order modulational instability of the plane waves. In the recent works \cite{Bilman-Duke-2019,Bilman-JDE-2021,Bilman-JNS-2019,Bilman-arxiv-2021}, the authors analyzed the large-order asymptotics of solitons and rogue waves in the near-field and far-field regimes, thus it is natural to consider what the large-order asymptotics of breathers will be.

In this paper, we would like to study the large-order asymptotics of KMBs under the far-field limit, and the analysis of the ABs will be performed in the near future. In a recent literature \cite{Bilman-2019-CPAM}, Bilman and Miller put forward the robust inverse scattering method and then give the Riemann-Hilbert problem (RHP) for the high-order rogue waves. Through two different scale transformations, they studied both the near-field \cite{Bilman-Duke-2019} and far-field asymptotics \cite{Bilman-arxiv-2021} for the rogue waves via the Deift-Zhou nonlinear steepest-descent method \cite{deift1993steepest,Deift1992steepest}. Meanwhile, Bilman, Buckingham and Wang also analyzed the near-field and far-field asymptotics for the high-order solitons \cite{Bilman-JDE-2021,Bilman-JNS-2019}. The results in these articles indicate a fact, in the near-field limit, the asymptotics of high-order rogue waves degenerate into the infinite-order solitons \cite{Bilman-Duke-2019}. But their far-field large-order asymptotic behaviors are quite different due to various formulas of their corresponding Riemann-Hilbert representations in the far-field regime, which can be verified from the Ref. \cite{Bilman-arxiv-2021}. In this paper, Bilman and Miller formulated a RHP for both the solitons and rogue waves, and the exponent phase term is written as $-\ii M\vartheta(\lambda; \chi, \tau)$, where the order $M$ can vary continuously. For the soliton solutions, $M=\frac{1}{2}k, k\in \mathbb{Z}_{\geq 0}$ and for the rogue waves, $M=\frac{1}{2}k+\frac{1}{4}, k\in \mathbb{Z}_{\geq 0}$. Recently, we analyzed the large-order asymptotics of high-order two-solitons with identical real parts \cite{ling2022large}, which can also be called the high-order breathers on the vanishing background. And we found a new genus-three region compared with the high-order solitons with a single spectrum. Although the rogue wave can be regarded as a limitation of the KMB, their dynamics are completely different. Their high-order counterparts also have distinct dynamics.  So it is meaningful to analyze the large-order dynamics of KMBs.  In contrast to the large-order asymptotics of high-order rogue waves \cite{Bilman-arxiv-2021}, we will confront some challenges for the KMB. A major difficulty comes from the special form of the phase term. On the one hand, compared to the large-order rogue waves, there are two kinds of singularities in the KMB phase term, one is the spectrum $\lambda=|\alpha|\ii\,(\alpha>1),$ the other one is $\lambda=\pm\ii$. On the other hand, compared to the large-order breathers on the vanishing background \cite{ling2022large}, there appears a new factor $\frac{\ii r}{4n}\log\left(\frac{\lambda-\ii}{\lambda+\ii}\right)$ in the phase term, where $n$ is the order number of KMB. Consequently, there appears a new branch cut $[-\ii, \ii]$. During the deformation of the RHP, an additional jump condition from this new cut leads to difficulties in the analysis.

Another motivation for this work is coming from the study of breather gas and integrable turbulence \cite{akhmediev2016breather,el2020spectral}. In general, turbulence can be expressed by the nonlinear modes of integrable systems. And the state of the turbulence is determined by the majority of excitations about the solitons or the breathers.
Since the breather is related to modulation instability, it is important in the formation of the chaotic wave field. Especially, when studying the chaotic wave field, we always assume that there are an infinite number of breathers. In the literature \cite{el2020spectral}, the authors gave a description of the soliton gas and breather gas with the finite-gap theory. For large-order solitons, we have obtained the high-genus region, thus we think for the large-order breathers, we can also get the high-genus region, the leading term in this region will be expressed with the Riemann-Theta function, which may help us understand the breather gas to a certain extent.

Before analysis, we first present some preliminaries about the NLS equation.
\subsection{The Riemann-Hilbert representation of large-order Kuznetsov-Ma breathers}\label{sec:rhp}
The Lax pair for the NLS equation \eqref{eq:nls} is
\begin{equation}\label{eq:laxpair}
\begin{aligned}
\pmb{\Phi}_{x}&=\mathbf{U}(\lambda; x, t)\pmb{\Phi},\quad \mathbf{U}=-\ii\lambda\sigma_3+\mathbf{Q},\\
\pmb{\Phi}_{x}&=\mathbf{V}(\lambda; x, t)\pmb{\Phi},\quad \mathbf{V}=-\ii\lambda^2\sigma_3+\lambda \mathbf{Q}+\frac{1}{2}\ii\sigma_3\left(\mathbf{Q}_x-\mathbf{Q}^2\right),
\end{aligned}
\end{equation}
where $\lambda$ is the spectral parameter and $\mathbf{Q}$ is given by
$$\mathbf{Q}=\begin{bmatrix}0&q\\-q^*&0
\end{bmatrix}.$$
Let $\Sigma_c$ be a vertical segment connecting $-\ii$ to $\ii$ with upward orientation.
With the seed solution $q=\ee^{\ii t}$, we can get the fundamental solution matrix
\begin{equation}\label{eq:funda-sol}
\pmb{\Phi}_{\rm bg}(\lambda; x, t)=\ee^{\frac{\ii t}{2}\sigma_3}n(\lambda)\begin{bmatrix}
1&\ii \lambda-\ii\rho(\lambda)\\\ii \lambda-\ii \rho(\lambda)&1
\end{bmatrix}\ee^{-\ii\theta(\lambda; x, t)\sigma_3}:=\ee^{\frac{\ii t}{2}\sigma_3}\mathbf{E}(\lambda)\ee^{-\ii\theta(\lambda; x, t)\sigma_3},
\end{equation}
where $\rho(\lambda)$ and $n(\lambda)$ are two analytic functions for $\lambda\notin \Sigma_c$ satisfying the conditions
$\rho^2(\lambda)=1+\lambda^2, \,\,\, n^2(\lambda)=\frac{\lambda+\rho(\lambda)}{2\rho(\lambda)}$ respectively,
and $\theta(\lambda; x, t)=\rho(\lambda)\left(x+\lambda t\right)$. Thus, to obtain a holomorphic matrix solution in $\mathbb{C}$ for Lax pair \eqref{eq:laxpair} with $q=\ee^{\ii t}$,
we normalize the above solution to $\pmb{\Phi}_{\rm bg}^{\rm in}(\lambda; x, t)=\ee^{\frac{\ii t}{2}\sigma_3}\mathbf{E}(\lambda)\ee^{-\ii\theta(\lambda; x, t)\sigma_3}\mathbf{E}^{-1}(\lambda)$.
One of the simplest methods to derive the high-order KMB solutions is using the high-order Darboux transformation\cite{terng2000backlund,ling2016multi}, which is shown in the following theorem.
\begin{theorem}
Suppose there exists a smooth solution $q\in L^{\infty}(\mathbb{R}^2)\cup C^{\infty}(\mathbb{R}^2)$, the Lax solution $\pmb{\Phi}(\lambda; x, t)$ is a holomorphic function in the whole complex plane $\mathbb{C}$, then the Darboux transformation for the linear system \eqref{eq:laxpair} can be given by
\begin{equation}\label{eq:DT-breather}
\begin{aligned}
\mathbf{T}_{n}(\lambda; x, t)&=\mathbb{I}+\mathbf{Y}_n\mathbf{M}^{-1}\mathbf{D}\mathbf{Y}_{n}^{\dagger}, \qquad \mathbf{M}=\mathbf{X}^{\dagger}\mathbf{S}\mathbf{X},\\
\end{aligned}
\end{equation}
where ${\mathbf{Y}}_{n}=\left[\pmb{\phi}_{1}^{[0]},\pmb{\phi}_{1}^{[1]}, \cdots, \pmb{\phi}_{1}^{[n-1]}\right]$ and
\begin{equation*}
\begin{aligned}
{\mathbf{D}}&=\begin{bmatrix}\frac{1}{\lambda-\lambda_1^*}&0&\cdots&0\\
\frac{1}{(\lambda-\lambda_1^*)^2}&\frac{1}{\lambda-\lambda_1^*}&\cdots&0\\
\vdots&\vdots&\ddots&\vdots\\
\frac{1}{\left(\lambda-\lambda_1^*\right)^{n}}&\frac{1}{\left(\lambda-\lambda_1^* \right)^{n-1}}&\cdots&\frac{1}{\lambda-\lambda_1^*}
\end{bmatrix},\qquad
{\mathbf{X}}=\begin{bmatrix}\pmb{\phi}_{1}^{[0]}&\pmb{\phi}_{1}^{[1]}&\cdots&\pmb{\phi}_{1}^{[n-1]}\\
0&\pmb{\phi}_{1}^{[0]}&\cdots&\pmb{\phi}_{1}^{[n-2]}\\
\vdots&\vdots&\ddots&\vdots\\
0&0&\cdots&\pmb{\phi}_{1}^{[0]}
\end{bmatrix},\\
{\mathbf{S}}&=\left( \binom{i+j-2}{i-1}\frac{(-1)^{i-1}\mathbb{I}_2}{(\lambda_1^*-\lambda_1)^{i+j-1}}\right)_{1\leq i,j\leq n},
\end{aligned}
\end{equation*}
with $\pmb{\phi}_1^{[k]}=\frac{1}{k!}\left(\frac{\rm d}{{\rm d}\lambda}\right)^{k}\pmb{\phi}_1|_{\lambda=\lambda_1}$ and $$\pmb{\phi}_1=\ee^{\frac{\ii t}{2}\sigma_3}
\left\{\frac{\sin(\rho(\lambda)(x+\lambda t))}{\rho(\lambda)}\begin{bmatrix}
-\ii\lambda c_1+c_2\\
\ii\lambda c_2-c_1 \\
\end{bmatrix}+\cos(\rho(\lambda)(x+\lambda t))\begin{bmatrix}
c_1\\
c_2 \\
\end{bmatrix}\right\},$$
$c_1, c_2$ are arbitrary complex constants, $\lambda_1\in \ii\mathbb{R}$ and $|\lambda_1|>1$.  The corresponding B\"acklund transformation between $q$ and $q^{[n]}$ is represented in terms of determinant form:
\begin{equation}\label{eq:bt}
q^{[n]}=\ee^{\ii t}+2\ii \mathbf{Y}_{n,1}\mathbf{M}^{-1}\mathbf{Y}_{n,2}^{\dag}=\frac{\ee^{\ii t}\det(\mathbf{M}+2\ii\ee^{-\ii t} \mathbf{Y}_{n,1}\mathbf{Y}_{n,2}^{\dag})}{\det(\mathbf{M})}.
\end{equation}
\end{theorem}
\begin{remark}
By choosing different parameters $\lambda_1$, we can get different kinds of breathers. If $\lambda_1=\alpha \ii$ with $0<\alpha<1$,
we can get the so-called AB \cite{akhmediev1986modulation}, which is localized in the $t$ direction and periodic in the $x$ direction.
If $\lambda_1=\alpha \ii$ with $\alpha>1$, then we can get the so-called KMB \cite{kuznetsov1977solitons,ma1979perturbed}, which is localized in the $x$ direction and periodic in the $t$ direction.
If $\lambda_1=\ii$, then this special peregrine breather \cite{peregrine1983water} is also called the rogue wave solution. If $\lambda_1+\lambda_1^*\neq 0$, then we can obtain the Tajiri-Wantanbe breather \cite{tajiri1998breather}.
\end{remark}
Especially, for $n=1$, the complex constants $c_1, c_2$ can be absorbed by the phase term $\theta(\lambda; x, t)$, then the fundamental solution $\phi^{[0]}_1$ can be written as another equivalent formula,
\begin{equation}
\phi_{1}^{[0]}=\ee^{\frac{\ii t}{2}\sigma_3}\left\{(\ii\lambda_1-1)\frac{\sin(\rho(\lambda_1)(x+\lambda_1 t+c))}{\rho(\lambda_1)}\begin{bmatrix}
-1\\
1\\
\end{bmatrix}+\cos(\rho(\lambda_1)(x+\lambda_1 t+c))\begin{bmatrix}
1\\
1 \\
\end{bmatrix}\right\},
\end{equation}
where the parameters $c$ and $c_1, c_2$ satisfy the following relations,
\begin{equation}
c_1=-\frac{(\ii\lambda_1-1)\sin(\rho(\lambda_1)c)}{\rho(\lambda_1)}+\cos(\rho(\lambda_1)c), \quad c_2=\frac{(\ii\lambda_1-1)\sin(\rho(\lambda_1)c)}{\rho(\lambda_1)}+\cos(\rho(\lambda_1)c).
\end{equation}
Under the choice of the above parameters, the first-order KMB reads
\begin{equation}
q^{[1]}=\ee^{\ii t}\left(\frac{\cos\left[\sinh(2\phi)(t-t_0)-2\ii\phi\right]-\cosh(\phi)\cosh\left[2\sinh(\phi)(x-x_0)\right]}{\cosh(\phi)\cosh\left[2\sinh(\phi)(x-x_0)\right]
-\cos\left[\sinh(2\phi)(t-t_0)\right]}\right),
\end{equation}
where
$\phi={\rm arccosh}(-\ii\lambda_1), x_0=-\Re(c), t_0=-\frac{\Im(c)}{\cosh(\phi)}$. It is clear that the KMB $q^{[1]}$ is localized in the $x$ direction and periodic in the $t$ direction.
Choosing proper parameters $c_1$ and $c_2$, we exhibit the second up to the fourth order breathers by the computer graphics (Fig. \ref{fig:asym-KM}). For $|c_1|\neq |c_2|$, these solutions present asymmetric features.
As the increasing of order, the expressions for the high-order breathers are enormous,
so it is hard to analyze the dynamic behavior for these solutions directly.  An alternative way to deal with large-order solutions is utilizing the Riemann-Hilbert representation.
\begin{figure}
\centering
\includegraphics[width=0.9\textwidth]{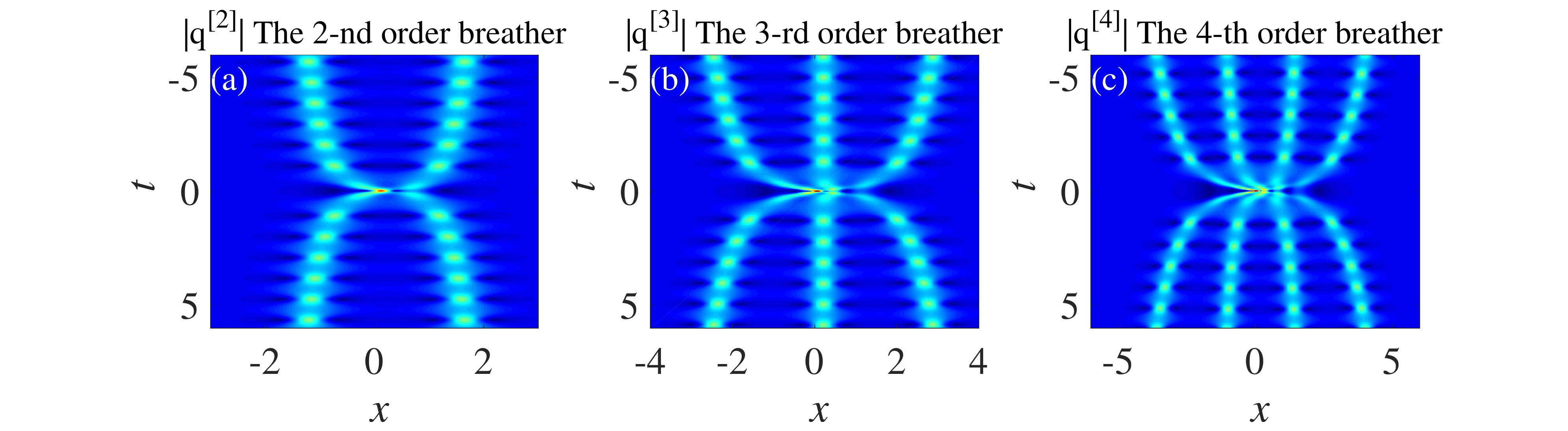}
\caption{The asymmetric soliton solutions by choosing $\lambda_1=2\ii, c_1=1, c_2=5.$}
\label{fig:asym-KM}
\end{figure}
\begin{remark}
For the fundamental solution Eq.\eqref{eq:funda-sol}, the phase term is written as $\theta(\lambda; x, t)=\rho(\lambda)\left(x+\lambda t\right).$
Actually, we can multiply a diagonal matrix independent of the variables $x, t$ from the righthand side of $\mathbf{E}(\lambda)$.
In other words, we can add a polynomial of $\lambda$ into the phase term. In that case, the breather shape will change correspondingly.
But in the following study, we still choose the phase term $\theta(\lambda; x, t)$ such that the Darboux matrix $\mathbf{T}_{n}(\lambda; x, t)$ has a good decomposition at the point $(x,t)=(0,0)$.
If $c_1$ and $c_2$ are independent of spectral parameters $\lambda_1$, it follows that, by a direct calculation, this decomposition of Darboux matrix is given as follows,
\begin{equation}
	\left(\frac{\lambda-\lambda_1}{\lambda-\lambda_1^*}\right)^{-n/2}\mathbf{T}_{n}(\lambda; 0,0)=\mathbf{Q}_{c}\left(\frac{\lambda-\lambda_1}{\lambda-\lambda_1^*}\right)^{\frac{1}{2}n\sigma_3}\mathbf{Q}_{c}^{-1}, \quad \mathbf{Q}_{c}=\frac{1}{|\mathbf{c}|}\begin{bmatrix}c_1&- c_2^*\\c_2&c_1^*
	\end{bmatrix},\,|\mathbf{c}|=\sqrt{|c_1|^2+|c_2|^2}.
\end{equation}
\end{remark}
With the aid of the above decomposition of Darboux transformation, the solution $q^{[n]}(x,t)$ at the point $(x,t)=(0,0)$ is given easily,
\begin{equation}
	q^{[n]}(0,0)=1+\frac{4\Im(\lambda_1)nc_1c_2^*}{|\mathbf{c}|^2}.
\end{equation}
If $c_1=c_2$, the origin point is also the location of the maximal value of norm $|q^{[n]}(x,t)|$, which can be proved by the mean value inequality.

	From the definition of $\mathbf{E}(\lambda)$ in Eq.\eqref{eq:funda-sol}, we know that it can be decomposed into $\mathbf{E}(\lambda)=\mathbf{Q}_{d}\left(\frac{\lambda-\ii}{\lambda+\ii}\right)^{\frac{1}{4}\sigma_3}\mathbf{Q}_{d}^{-1}$, where $\mathbf{Q}_{d}=\frac{1}{\sqrt{2}}\begin{bmatrix}1&-1\\
		1&1
	\end{bmatrix}$. We would like to establish a RHP to study the large-order asymptotics of KMBs in the far-field regime. Fortunately, we can get two types of RHPs under some constraints to the parameters $c_1, c_2$, one is $c_1=c_2$ and the other is $c_1=-c_2$. That is, $\mathbf{Q}_c=\mathbf{Q}_d$ and $\mathbf{Q}_c=\mathbf{Q}_d^{-1}$.
For the general $c_1$ and $c_2$, we just exhibit some exact solutions figures in Fig.\ref{fig:asym-KM}. The corresponding large-order asymptotics has some additional difficulties to be overcome. We have not yet thought of a suitable RHP to deal with the general case, where the difficulty is that the matrix $\mathbf{Q}_{c}$ in the Darboux matrix $\mathbf{T}_n(\lambda; 0,0)$ and $\mathbf{Q}_{d}$ appearing in $\mathbf{E}(\lambda)$ have not evident relations. Thus the RHP given in the current paper will not be available anymore. Next, we merely construct the RHP for the above-mentioned two special cases.

Based on the idea of normalization, by using the fundamental solution \eqref{eq:funda-sol} and the Darboux matrix \eqref{eq:DT-breather},
we define two sectional analytic matrices
\begin{equation}
\mathbf{M}^{[n]}(\lambda; x, t):=\left\{\begin{aligned}&r^{\frac{1}{2}\sigma_3}\mathbf{T}_{n}(\lambda; x, t)\ee^{\frac{\ii t}{2}\sigma_3}\mathbf{E}(\lambda)\ee^{-\ii\theta(\lambda; x, t)\sigma_3}\mathbf{E}^{-1}(\lambda)\mathbf{T}_{n}(\lambda; 0, 0)^{-1}
\mathbf{Q}_{d}^{r}\ee^{\ii\hat\theta(\lambda; x, t)\sigma_3}r^{-\frac{1}{2}\sigma_3},\,\,\lambda\, \text{inside }\, D_0,
\\&\left(\frac{\lambda-\lambda_1}{\lambda-\lambda_1^*}\right)^{-n/2}r^{\frac{1}{2}\sigma_3}\mathbf{T}_{n}(\lambda; x, t)\ee^{\frac{\ii t}{2}\sigma_3}\mathbf{E}(\lambda)
\ee^{\ii\left[\hat\theta(\lambda; x, t)-\theta(\lambda; x, t)\right]\sigma_3}\\ &\qquad \qquad \qquad \qquad \qquad\qquad  \times
\left(\frac{\lambda-\lambda_1}{\lambda-\lambda_1^*}\right)^{-n/2\sigma_3}\left(\frac{\lambda-\ii}{\lambda+\ii}\right)^{-1/4 r\sigma_3}r^{-\frac{1}{2}\sigma_3},
\lambda\, \text{exterior to} \,D_0,\end{aligned}\right.
\end{equation}
where $\hat{\theta}(\lambda; x, t):=\lambda\left(x+\lambda t\right), r=\pm 1.$ And $r=1$ represents the case $c_1=c_2$; $r=-1$ corresponds to the other case $c_1=-c_2$. $D_0$ is a big closed contour involving the spectra $\pm\lambda_1, \pm\ii$.
Then the newly defined matrix $\mathbf{M}^{[n]}(\lambda; x, t)$ satisfies the following RHP.
\begin{rhp}\label{rhp:reform}(KMB of order $n$-reformulation)
Let $(x, t)\in\mathbb{R}^2$ be arbitrary parameters, and $n\in\mathbb{Z}_{>0}$. Then we can find a $2\times 2$ matrix function $\mathbf{M}^{[n]}(\lambda; x, t)$ with the following properties:
\begin{itemize}
\item \textbf{Analyticity}: $\mathbf{M}^{[n]}(\lambda; x, t)$ is analytic for $\lambda\in\mathbb{C}\setminus \partial D_0$. It takes the continuous boundary values from the interior and exterior of $\partial D_0$.
\item \textbf{Jump condition}: The boundary values on the jump contour $\partial D_0$ are related by
    \begin{multline}
    \mathbf{M}^{[n]}_{+}(\lambda; x, t)=\mathbf{M}^{[n]}_{-}(\lambda; x,t)\ee^{-\ii\hat{\theta}(\lambda; x, t)\sigma_3}\left(\frac{\lambda-\lambda_1}{\lambda-\lambda_1^*}\right)^{\frac{n}{2}\sigma_3}
    \left(\frac{\lambda-\ii}{\lambda+\ii}\right)^{\frac{1}{4}r\sigma_3}\\\times
    \mathbf{Q}_{d}^{-1}\left(\frac{\lambda-\ii}{\lambda+\ii}\right)^{- \frac{1}{4}r\sigma_3} \left(\frac{\lambda-\lambda_1}{\lambda-\lambda_1^*}\right)^{-\frac{n}{2}\sigma_3}\ee^{\ii\hat{\theta}(\lambda; x, t)\sigma_3},\quad  \lambda\in\partial D_0.
    \end{multline}
\item \textbf{Normalization}: $\mathbf{M}^{[n]}(\lambda; x, t)=\mathbb{I}+\mathcal{O}(\lambda^{-1}),$ as $\lambda\to\infty.$
\end{itemize}
The potential $q^{[n]}(x, t)$ can be recovered with
\begin{equation}
	q^{[n]}(x, t)=2\ii r\lim \limits_{\lambda\to\infty}\lambda \mathbf{M}^{[n]}(\lambda; x, t)_{12}.
	\end{equation}
\end{rhp}
The existence and uniqueness of above RHP can be proved by the Zhou's vanishing lemma \cite{zhou1989riemann}.
With this RHP \ref{rhp:reform}, we prepare to study the large-order asymptotics of KMBs in the far-field regime. Before discussing it, we first introduce a scale transformation of $x$ and $t$
such that they have the same order with the factor $\left(\frac{\lambda-\lambda_1}{\lambda-\lambda_1^*}\right)^{\pm \frac{n}{2}\sigma_3}$,
\begin{equation}\label{eq:scale-trans}
x=n \chi,\quad t=n \tau,
\end{equation}
then the jump matrix in RHP \ref{rhp:reform} changes into
\begin{equation}
\mathbf{M}^{[n]}_{+}(\lambda; n\chi, n\tau)=\mathbf{M}^{[n]}_{-}(\lambda; n\chi, n\tau)\ee^{-\ii n\vartheta(\lambda; \chi, \tau)\sigma_3}\mathbf{Q}_{d}^{-1}\ee^{\ii n\vartheta(\lambda; \chi, \tau)\sigma_3},\quad  \lambda\in\partial D_0,
\end{equation}
where
\begin{equation}\label{eq:phase-term}
\vartheta(\lambda; \chi, \tau)=\lambda\chi+\lambda^2 \tau+\frac{1}{2}\ii\log\left(\frac{\lambda-\lambda_1}{\lambda-\lambda_1^*}\right)
+\frac{\ii r}{4n}\log\left(\frac{\lambda-\ii}{\lambda+\ii}\right).
\end{equation}
Compared to the large-order solitons with single spectrum, there adds a factor $\frac{\ii r}{4n}\log\left(\frac{\lambda-\ii}{\lambda+\ii}\right)$ in
the large-order KMBs. When $n\to\infty$, this factor will vanish, but for the large-order asymptotics, these two solutions have distinct behaviors.
Next, we give several decompositions to this constant matrix $\mathbf{Q}_{d}^{-1}$:
\begin{equation}\label{remark:decom}
	\begin{aligned}
		\mathbf{Q}_{d}^{-1}&=\begin{bmatrix}\frac{\sqrt{2}}{2}&0\\0&\sqrt{2}
		\end{bmatrix}\begin{bmatrix}1&0\\-\frac{1}{2}&1
		\end{bmatrix}\begin{bmatrix}1&1\\0&1
		\end{bmatrix}:=\mathbf{Q}_{L}^{[1]}\mathbf{Q}_C^{[1]}\mathbf{Q}_{R}^{[1]},&\quad \left(``{\rm DLU}"\right),\\
		\mathbf{Q}_{d}^{-1}&=\begin{bmatrix}\sqrt{2}&0\\0&\frac{\sqrt{2}}{2}
		\end{bmatrix}\begin{bmatrix}1&\frac{1}{2}\\0&1
		\end{bmatrix}\begin{bmatrix}1&0\\-1&1
		\end{bmatrix}:=\mathbf{Q}_{L}^{[2]}\mathbf{Q}_C^{[2]}\mathbf{Q}_{R}^{[2]},&\quad \left(``{\rm DUL}"\right),\\
		\mathbf{Q}_{d}^{-1}&=\mathbf{Q}_{L}^{[2]}\begin{bmatrix}1&-\frac{1}{2}\\0&1
		\end{bmatrix}\begin{bmatrix}0&1\\-1&0
		\end{bmatrix}\begin{bmatrix}1&-1\\0&1
		\end{bmatrix}:=\mathbf{Q}_{L}^{[2]}\mathbf{Q}_{L}^{[3]}\mathbf{Q}_C^{[3]}\mathbf{Q}_{R}^{[3]},&\quad \left(``\rm{DUTU}"\right),\\
		\mathbf{Q}_{d}^{-1}&=\mathbf{Q}_{L}^{[1]}\begin{bmatrix}1&0\\\frac{1}{2}&1
		\end{bmatrix}\begin{bmatrix}0&1\\-1&0
		\end{bmatrix}\begin{bmatrix}1&0\\1&1
		\end{bmatrix}:=\mathbf{Q}_{L}^{[1]}\mathbf{Q}_{L}^{[4]}\mathbf{Q}_C^{[4]}\mathbf{Q}_{R}^{[4]},&\quad \left(``\rm{DLTL}"\right),
	\end{aligned}
\end{equation}
which will be useful in the following analysis.

In this paper, we are focusing on the study of the large-order asymptotics of KMBs, without loss of generality, we choose two types of parameters $\lambda_1=\frac{3}{2}\ii, c_1=-c_2=1,$ and $ \lambda_1=2\ii, c_1=c_2=1$, and give the density plots in Fig. \ref{fig:KM}.
\begin{figure}
\centering
\includegraphics[width=1\textwidth]{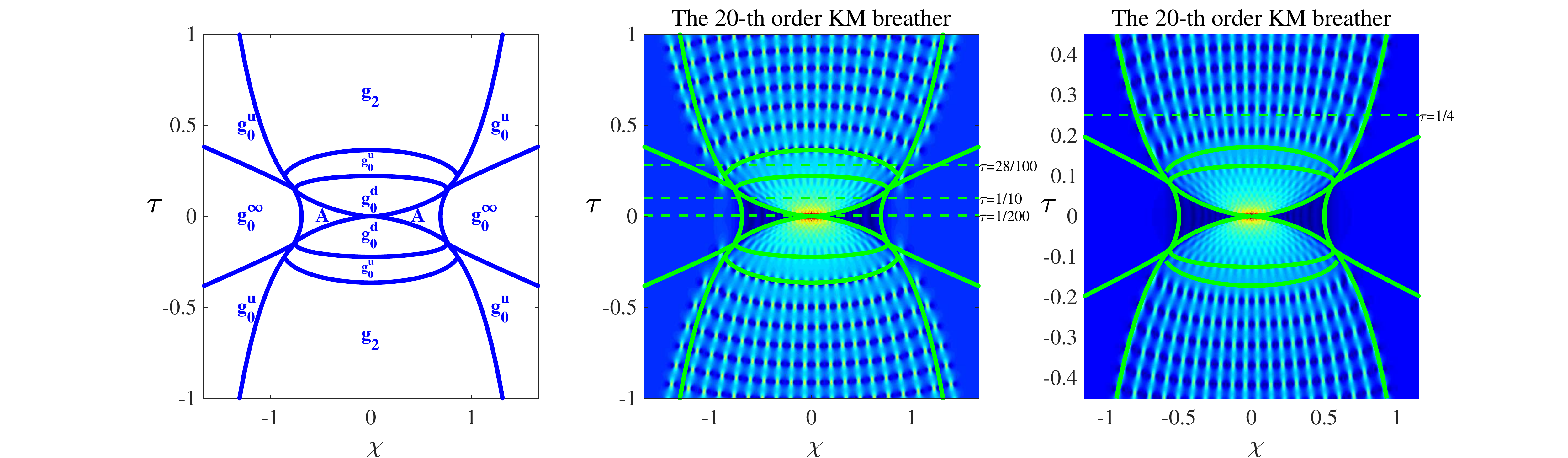}
\caption{The boundary between different regions(left), the $20$-th order KMBs by choosing $\lambda_1=\frac{3}{2}\ii, c_1{=}-c_2{=}1$(middle), $\lambda_1=2\ii, c_1{=}c_2{=}1$(right).}
\label{fig:KM}
\end{figure}
From these two figures, we can see that there are five different asymptotic regions, which are called the genus-two region ($g_2$ in Fig. \ref{fig:KM}), the genus-zero-up region ($g_0^u$ in Fig. \ref{fig:KM}),
the genus-zero-down region ($g_0^d$ in Fig. \ref{fig:KM}), the genus-zero-infinity region ($g_0^\infty$ in Fig. \ref{fig:KM})
and the algebraic-decay region (A in Fig. \ref{fig:KM}).
In the reference \cite{Bilman-JNS-2019}, the authors gave a detailed description for calculating the boundaries between different regions. Similarly, we will give a brief description of the boundaries for the high-order KMBs.
With the definition of $\vartheta(\lambda; \chi, \tau)$ in Eq.\eqref{eq:phase-term}, for convenience, we set the spectral parameter $\lambda_1=\alpha\ii(\alpha>1)$. The critical points of $\vartheta(\lambda; \chi, \tau)$ satisfy the following algebraic equation:
\begin{equation}\label{eq:cri-poi}
2n\left(\chi+2\tau\lambda\right)\left(\lambda^2+1\right)\left(\lambda^2+\alpha^2\right)-\lambda^2(r+2\alpha n)-\alpha(\alpha r+2n)=0.
\end{equation}
If the discriminant of Eq.\eqref{eq:cri-poi} about $\lambda$ is greater than $0$, the quintic polynomial Eq.\eqref{eq:cri-poi} has at least three real critical points, which corresponds to the algebraic-decay region.
Next, we give the boundary between $g_2$ and $g_0^\infty$ regions, which is given by the condition $\Im\left(\vartheta(\lambda^{\pm})\right)=0, $ where $\lambda^{\pm}$ are two critical points of $\vartheta(\lambda; \chi, \tau)$.
The rest of curves (Fig. \ref{fig:KM}) are the boundaries between $g_0^u$ and $g_0^d$ as well as $g_2$, and the boundary between $g_0^u$ and $g_0^{\infty}$, these two boundaries all depend on an algebraic curve of genus-zero.
For these three genus-zero regions, we need introduce a $g$-function \cite{deift1997new} defined as
\begin{equation}
g'(\lambda):=\frac{R(\lambda)}{2}\left(\frac{\ii r}{2n\left(\lambda-\ii\right)R(\ii)}-\frac{\ii r}{2n (\lambda+\ii)R(-\ii)}+\frac{\ii}{(\lambda-k\ii)R(k\ii)}-\frac{\ii}{(\lambda+k\ii)R(-k\ii)}+4\tau\right)-\vartheta'(\lambda; \chi, \tau),
\end{equation}
where
$$g(\lambda)\equiv g(\lambda; \chi, \tau),\quad R(\lambda)\equiv R(\lambda; \chi, \tau)=\sqrt{(\lambda-a_1(\chi, \tau))(\lambda-a_1^*(\chi, \tau))}.$$
At this time, the controlling phase term becomes $h(\lambda)\equiv h(\lambda; \chi, \tau):=g(\lambda)+\vartheta(\lambda; \chi, \tau)$ rather than $\vartheta(\lambda; \chi, \tau)$.
From $g'(\lambda)$, we know that $h'(\lambda)$ equals to
\begin{equation}
h'(\lambda)=\frac{R(\lambda)}{2}\left(\frac{\ii r}{2n\left(\lambda{-}\ii\right)R(\ii)}{-}\frac{\ii r}{2n(\lambda{+}\ii)R({-}\ii)}{+}\frac{\ii}{(\lambda{-}k\ii)R(k\ii)}{-}\frac{\ii}{(\lambda{+}k\ii)R({-}k\ii)}{+}4\tau\right).
\end{equation}
On the one hand, in the genus-zero-down region, $h'(\lambda)$ should have at least two real roots such that the singularities $\lambda=\pm \alpha\ii, \lambda=\pm \ii$ are all in a closed curve given by the $\Im(h(\lambda))$. As $\tau$ increases, these two real roots coincide into one double root,
then the variable $(\chi, \tau)$ will transfer into the genus-zero-up region. On the other hand, in the genus-zero-infinity region, $h'(\lambda)$ has the real root and in the genus-zero-up region, $h'(\lambda)$ only has complex roots, thus the boundary between genus-zero-infinity region and the genus-zero-up region satisfies the same condition. Moreover, similar to the boundary between $g_2$ and $g_0^\infty$ regions, the boundary between $g_0^u$ and $g_2$
regions satisfies the condition $\Im\left(h(\hat{\lambda}^{\pm})\right)$, where $\hat{\lambda}^{\pm}$ are the critical points of $h(\lambda)$. Then the boundaries given in Fig. \ref{fig:KM} have been given completely. In the following, we would like to derive the asymptotic expressions of high-order KMBs for the above mentioned five different regions.

\section{The large-order asymptotics in the genus-two region}\label{sec:genus-2}
Firstly, we prepare to study the asymptotics in the genus-two region. Before studying it, we introduce a $g_{2}(\lambda)$-function satisfying the following RHP.
\begin{rhp}\label{RHP:genus-2}
Let $(\chi, \tau)\in\mathbb{R}^2$, we can find a  $g_{2}(\lambda):=g_{2}(\lambda; \chi, \tau)$-function with following conditions.
\begin{itemize}
\item
{\bf Analyticity}: $g_{2}(\lambda)$ is analytic in $\mathbb{C}\setminus \Sigma_{g_{2}}^{\pm}\cup \Sigma_{g}\cup\Gamma_{g_{2}}^{\pm}$, where these arcs are to be determined,
and it takes the continuous boundary conditions from the left and right sides of each arc.
\item
{\bf Jump Condition}: The jump conditions on these arcs are related by
\begin{equation}
\begin{aligned}
&g_{2,+}(\lambda)+g_{2,-}(\lambda)+2\vartheta(\lambda; \chi, \tau)=\kappa_{2},&\qquad\lambda\in\Sigma_{g_{2}}^{\pm},\\
&g_{2,+}(\lambda)+g_{2,-}(\lambda)+\vartheta_{+}(\lambda; \chi, \tau)+\vartheta_{-}(\lambda; \chi, \tau)=l_{2},&\qquad\lambda\in\Sigma_{g},\\
&g_{2,+}(\lambda)-g_{2,-}(\lambda)=\varpi_2,&\qquad\lambda\in\Gamma_{g_{2}}^{\pm}.
\end{aligned}
\end{equation}
\item
{\bf Normalization}: As $\lambda\to\infty$, $g_2(\lambda)$ satisfies
\begin{equation}
g_{2}(\lambda)\to \mathcal{O}(\lambda^{-1}).
\end{equation}
\item
{\bf Symmetry}: $g_{2}(\lambda)$ satisfies the Schwartz symmetric condition,
\begin{equation}
g_{2}(\lambda)=g_{2}(\lambda^*)^*.
\end{equation}
\end{itemize}
\end{rhp}

From the definition of phase term $\vartheta(\lambda; \chi, \tau)$ in Eq.\eqref{eq:phase-term},
we differentiate $g_{2}(\lambda)$ with respect to $\lambda$ to remove the logarithm terms and the integral constants $\kappa_2, l_2, \varpi_2$ simultaneously, then we have
\begin{equation}
g'_{2,+}(\lambda)+g'_{2,-}(\lambda)=-2\chi-4\lambda\tau-\frac{\ii}{\lambda-\lambda_1}+\frac{\ii}{\lambda-\lambda_1^*}-\frac{\ii r}{2n(\lambda-\ii)}+\frac{\ii r}{2n(\lambda+\ii)},\quad\lambda\in\Sigma_{g_{2}}^{\pm}\cup\Sigma_{g}.
\end{equation}
To solve this scalar RHP, we introduce a square root function $R_{2}(\lambda)\equiv R_{2}(\lambda; \chi, \tau)$ with the definition
\begin{multline}\label{eq:R1}
R_{2}(\lambda):=\sqrt{(\lambda-a_2)(\lambda-a_{2}^*)(\lambda-b_{2})(\lambda-b_{2}^*)(\lambda-d_{2})(\lambda-d_{2}^*)}\\
:=\sqrt{\lambda^6-s_1\lambda^5+s_2\lambda^4-s_3\lambda^3+s_4\lambda^2-s_5\lambda+s_6},
\end{multline}
the parameters $a_2,b_2,d_2$ and $s_{i}(i=1,\cdots,6)$ have the following relationship,
\begin{equation}
\begin{aligned}
s_1&=2\left(a_{2R}+b_{2R}+d_{2R}\right), \, s_2=|a_{2}|^2+|b_{2}|^2+|d_{2}|^2+4\left(a_{2R}b_{2R}+a_{2R}d_{2R}+b_{2R}d_{2R}\right),\\
s_3&=2|a_{2}|^2\left(b_{2R}+d_{2R}\right)+2|b_{2}|^2\left(a_{2R}+d_{2R}\right)+2|d_{2}|^2\left(a_{2R}+b_{2R}\right)+8a_{2R}b_{2R}d_{2R},\\
s_4&=|a_{2}|^2\left(|d_{2}|^2+4b_{2R}d_{2R}\right)+|b_{2}|^2\left(|a_{2}|^2+4a_{2R}d_{2R}\right)+|d_{2}|^2\left(|b_{2}|^2+4a_{2R}b_{2R}\right),\\
s_5&=2a_{2R}|b_{2}|^2|d_{2}|^2+2b_{2R}|a_{2}|^2|d_{2}|^2+2d_{2R}|a_{2}|^2|b_{2}|^2,\quad s_6=|a_{2}|^2|b_{2}|^2|d_{2}|^2,
\end{aligned}
\end{equation}
where $a_{2R}, b_{2R}, d_{2R}$ are the real parts of $a_{2}, b_{2}, d_{2}$ respectively.
Divide $g_2'(\lambda)$ by the $R_{2}(\lambda)$ function, we have
\begin{equation}
\left(\frac{g'_{2}(\lambda)}{R_{2}(\lambda)}\right)_+-\left(\frac{g'_{2}(\lambda)}{R_{2}(\lambda)}\right)_-=\frac{-2\chi-4\lambda\tau-\frac{\ii }{\lambda-\lambda_1}+\frac{\ii }{\lambda-\lambda_1^*}-\frac{\ii r}{2n(\lambda-\ii)}+\frac{\ii r}{2n(\lambda+\ii)}}{R_{2,+}(\lambda)}.
\end{equation}
With the Plemelj formula and the generalized residue theorem, $g'_{2}(\lambda)$ can be expressed into an explicit formula:
\begin{equation}\label{eq:dG}
\begin{aligned}
g'_{2}(\lambda)&=R_{2}(\lambda)\left(\mathop{\rm Res}\limits_{s=\lambda}+\mathop{\rm Res}\limits_{s=\lambda_1}+\mathop{\rm Res}\limits_{s=\lambda_1^*}{+}\mathop{\rm Res}\limits_{s=\infty}\right)
\left(\frac{-\chi-2 s\tau-\frac{\ii}{2(s-\lambda_1)}+\frac{\ii}{2(s-\lambda_1^*)}-\frac{1}{4n}\frac{\ii r}{(s-\ii)}+\frac{1}{4n}\frac{\ii r}{(s+\ii)}}{R_{2}(s)\left(s-\lambda\right)}\right)\\
&+R_{2}(\lambda)\left(\mathop{\rm Res}\limits_{s=\ii}+\mathop{\rm Res}\limits_{s=-\ii}\right)\left(\frac{-\chi-2 s\tau-\frac{\ii}{2(s-\lambda_1)}+\frac{\ii}{2(s-\lambda_1^*)}-\frac{1}{4n}\frac{\ii r}{(s-\ii)}+\frac{1}{4n}\frac{\ii r}{(s+\ii)}}{R_{2}(s)\left(s-\lambda\right)}\right)\\
&=R_{2}(\lambda)\left[-\frac{\ii}{2R_{2}(\lambda_1)(\lambda_1-\lambda)}+\frac{\ii}{2R_{2}(\lambda_1^*)(\lambda_1^*-\lambda)}-\frac{\ii r}{4nR_{2}(\ii)(\ii-\lambda)}+\frac{\ii r}{4nR_{2}(-\ii)(-\ii-\lambda)}\right]\\
&-\chi-2\lambda\tau-\frac{\ii}{2}\frac{1}{\lambda-\lambda_1}+\frac{\ii}{2}\frac{1}{\lambda-\lambda_1^*}-\frac{\ii r}{4n}\frac{1}{\lambda-\ii}-\frac{\ii r}{4n}\frac{1}{\lambda+\ii}.
\end{aligned}
\end{equation}
By adding the $g_2(\lambda)$-function into the phase term $\vartheta(\lambda; \chi, \tau)$, the phase term can be modified as $h_{2}(\lambda)\equiv h_{2}(\lambda; \chi, \tau):=g_{2}(\lambda)+\vartheta(\lambda; \chi, \tau)$, thus we have
\begin{equation}\label{eq:dh}
h'_{2}(\lambda)=R_{2}(\lambda)\left[-\frac{\ii}{2R_{2}(\lambda_1)(\lambda_1-\lambda)}+\frac{\ii}{2R_{2}(\lambda_1^*)(\lambda_1^*-\lambda)}-\frac{\ii r}{4nR_{2}(\ii)(\ii-\lambda)}+\frac{\ii r}{4nR_{2}(-\ii)(-\ii-\lambda)}\right].
\end{equation}
For fixed $(\chi, \tau)$ in this region, there are six parameters $s_i(i=1,\cdots, 6)$ to be determined.
From the normalization condition of $g_{2}(\lambda)$ at $\lambda=\infty$, we get four relations about these unknown parameters,
\begin{equation}\label{eq:four-relation}
\begin{aligned}
\mathcal{O}(\lambda^2):&\frac{1}{R_{2}(\lambda_1)}-\frac{1}{R_{2}(\lambda_1^*)}+\frac{r}{2nR_{2}(\ii)}-\frac{r}{2nR_{2}(-\ii)}=0,\\
\mathcal{O}(\lambda):&\frac{\ii\lambda_1}{R_{2}(\lambda_1)}-\frac{\ii \lambda_1^*}{R_{2}(\lambda_1^*)}-\frac{r}{2nR_{2}(\ii)}-\frac{r}{2nR_{2}(-\ii)}-4\tau=0,\\
\mathcal{O}(1):&\frac{\ii\lambda_1^2}{R_{2}(\lambda_1)}-\frac{\ii (\lambda_1^*)^2}{R_{2}(\lambda_1^*)}-\frac{\ii r}{2n R_{2}(\ii)}+\frac{\ii r}{2nR_{2}(-\ii)}-2\chi-2\tau s_1=0,\\
\mathcal{O}(\lambda^{-1}):&\frac{\ii\lambda_1^3}{R_{2}(\lambda_1)}-\frac{\ii(\lambda_1^*)^3}{R_{2}(\lambda_1^*)}+\frac{r}{2nR_{2}(\ii)}+\frac{r}{2nR_{2}(-\ii)}+\frac{1}{2}\left(4s_2-3s_1^2\right)\tau-s_1\chi=0.
\end{aligned}
\end{equation}
From the first relation in Eq.\eqref{eq:four-relation}, we have
\begin{equation}\label{eq:first}
\frac{1}{R_{2}(\lambda_1)}=\frac{1}{R_{2}(\lambda_1^*)}-\frac{r}{2n R_{2}(\ii)}+\frac{r}{2nR_{2}(-\ii)}.
\end{equation}
Substituting the above relation into the second equation of Eq. \eqref{eq:four-relation}, we get
\begin{equation}\label{eq:second}
\frac{1}{R_{2}(\lambda_1^*)}=\frac{2\tau}{-\Im(\lambda_1)}-\frac{\ii\lambda_1+1}{4\Im(\lambda_1)n}\frac{r}{R_{2}(\ii)}+\frac{\ii\lambda_1-1}{4\Im(\lambda_1)n}\frac{r}{R_{2}(-\ii)}.
\end{equation}
Similarly, substitute Eq.\eqref{eq:first} and Eq.\eqref{eq:second} into the third equation and the fourth equation in Eq.\eqref{eq:four-relation}, we have
\begin{equation}\label{eq:last}
\begin{aligned}
\frac{1}{R_{2}(\ii)}&=-nr\frac{\left(8|\lambda_1|^2-8\Re(\lambda_1)s_1+3s_1^2-4s_2-16\Re(\lambda_1)\ii+4\ii s_1\right)\tau+\left(2s_1-8\Re(\lambda_1)+4\ii\right)\chi}{2(\ii-\lambda_1)(\ii-\lambda_1^*)},\\
\frac{1}{R_{2}(\lambda_1)}&=\frac{\left(4s_2-3s_1^2+4s_1\lambda_1^*-8\right)\tau+\left(4\lambda_1^*-2s_1\right)\chi}{4(1+\lambda_1^2)\Im(\lambda_1)}.
\end{aligned}
\end{equation}
Separating the real and the imaginary parts of these two equations in Eq.\eqref{eq:last}, then we get four relations about the unknown parameters $s_{i} (i=1,2,\cdots, 6)$.
Moreover, substitute Eq.\eqref{eq:last} into the Eq.\eqref{eq:dh}, we get
\begin{equation}
h'_{2}(\lambda){=}R_{2}(\lambda)\frac{8\tau\lambda^2{+}\left(4\chi{+}4s_1\tau{-}16\Re(\lambda_1)\tau\right)\lambda{+}8|\lambda_1|^2\tau{-}8\Re(\lambda_1)s_1\tau{+}3s_1^2\tau
{-}4s_2\tau{+}8\tau{-}8\Re(\lambda_1)\chi{+}2s_1\chi}{4(\lambda-\lambda_1)(\lambda-\lambda_1^*)(\lambda^2+1)}.
\end{equation}
Obviously, $h_2'(\lambda)$ has eight roots, six of which are the branch points and the remaining two are a pair of conjugate complex roots.
By integrating $g_2'(\lambda)$ function, we get the $g_{2}(\lambda)$ function, which is shown in theorem \ref{theo:g2}.

\begin{theorem}\label{theo:g2}
With the explicit formula of $g'_{2}(\lambda)$ in Eq.\eqref{eq:dG}, the $g_{2}(\lambda)$-function defined by
$$g_{2}(\lambda)=\int_{\infty}^{\lambda}g'_{2}(s)ds$$
satisfies all the jump conditions in RHP \ref{RHP:genus-2}, and the integrated constants can also be determined.
\end{theorem}
\begin{proof}
From the explicit formula $g'_{2}(\lambda)$ in Eq.\eqref{eq:dG}, we know that $g'_{2}(\lambda)$ has the same branch cuts with $R_{2}(\lambda)$, and the singularities $\lambda=\lambda_1, \lambda_1^*, \pm\ii$ can be removed,
thus it only has the jump discontinuity on the branch cuts and satisfies the jump conditions in RHP \ref{RHP:genus-2}. Choose one suitable integral path, these integrated constants can be expressed as
\begin{equation}
\begin{aligned}
\kappa_{2}&=2\int_{\infty}^{a_{2}}g'_{2}(s)ds+2\vartheta(a_{2}; \chi, \tau),\quad
\varpi_{2}=2\int_{a_{2}}^{b_{2}}h'_{2}(s)ds,\quad
l_2=\kappa_{2}+r\frac{\pi}{2n}+2\int_{b_{2}}^{d_{2}}h'_{2}(s)ds.
\end{aligned}
\end{equation}
\end{proof}
To determine the unknown parameters $s_i(i=1,2\cdots,6)$, we impose these two integrals $\int_{a_{2}}^{b_{2}}h'_{2}(s)ds$ and $\int_{b_{2}}^{d_{2}}h'_{2}(s)ds$ as real numbers. Combining the four normalization conditions in Eq.\eqref{eq:four-relation}, we numerically obtain these six unknown branch points. In Fig.\ref{fig:KM},
we give two figures by choosing different spectra $\lambda_1$ and parameters $c_1, c_2$. Both of them have five asymptotic regions.
To verify it, we will check the asymptotic solutions and the exact solutions by choosing the given $\lambda_1, c_1, c_2$ in Fig.\ref{fig:KM} and setting $(\chi, \tau)$ in the fixed regions(shown by the green dashed line in Fig.\ref{fig:KM}).
In this genus-two region, we choose $\lambda_1=2\ii, c_1=c_2=1$, which corresponds $r=1$, under this parameters setting, we give the sign chart of $\Im(h_{2}(\lambda))$ and the jump contours for the following defined $\mathbf{S}_{1}(\lambda; \chi, \tau)$ and $\mathbf{T}_{1}(\lambda; \chi, \tau)$ in Fig. \ref{fig:genus-two} by putting $\tau=\frac{1}{4}, \chi=\frac{1}{5}$.
\begin{figure}[!ht]
\centering
\includegraphics[width=0.45\textwidth]{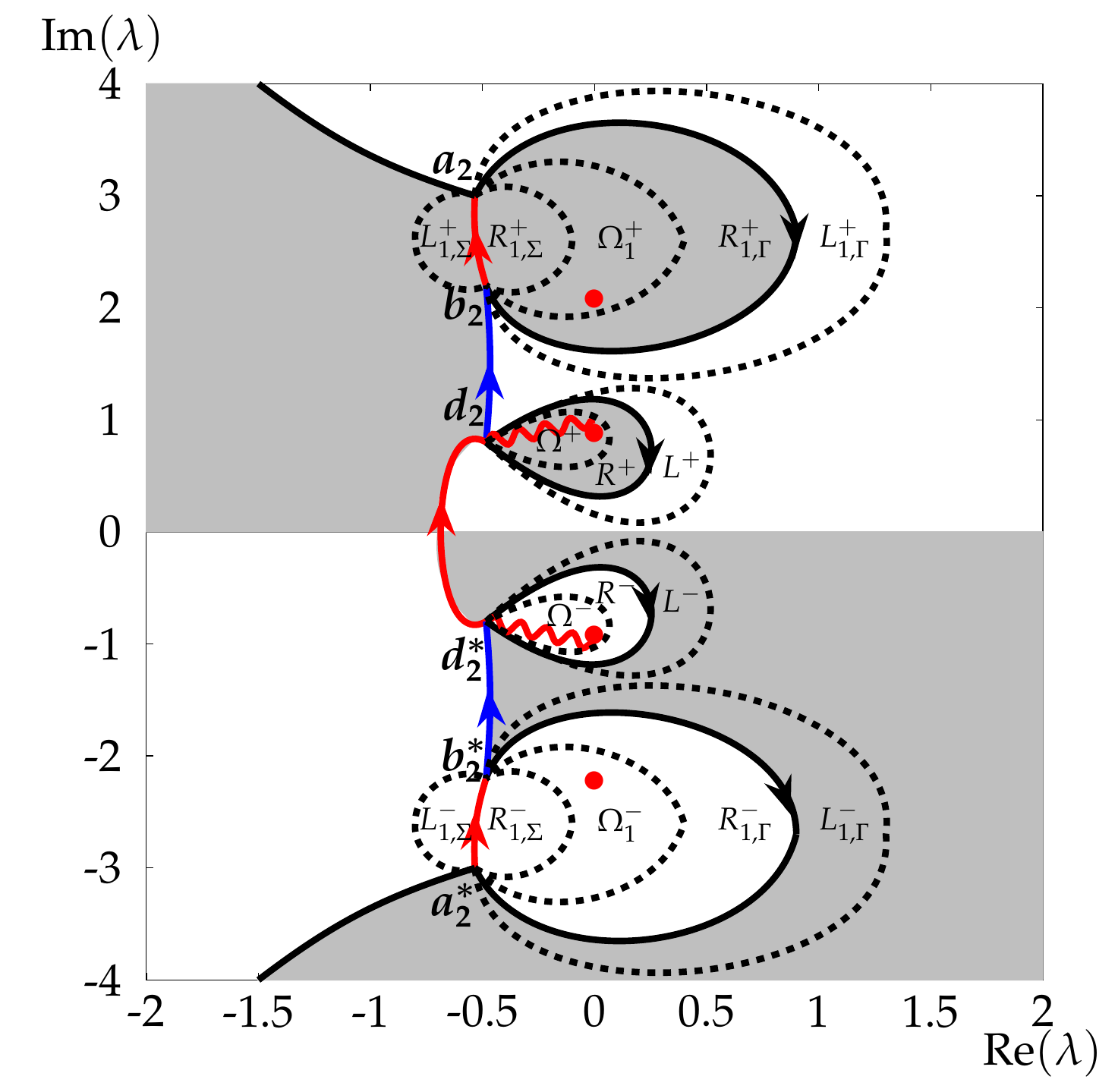}
\centering
\includegraphics[width=0.45\textwidth]{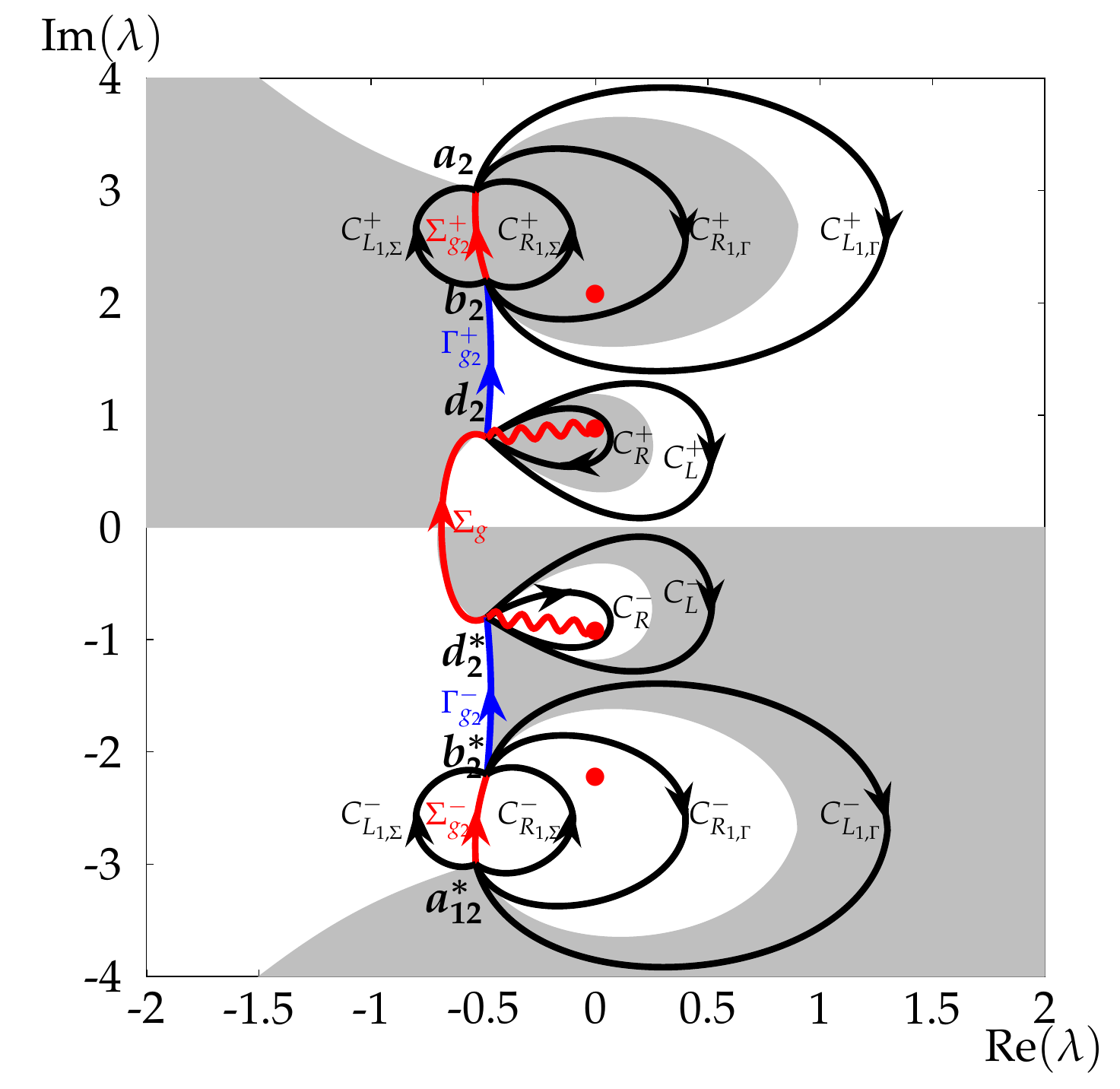}
\caption{The sign chart of ${\Im}(h_2(\lambda; \frac{1}{4}, \frac{1}{5}))$ in the genus-two region, where ${\Im}(h_2(\lambda; \frac{1}{4}, \frac{1}{5}))>0$(unshaded) and ${\Im}(h_2(\lambda; \frac{1}{4}, \frac{1}{5}))<0$(shaded). It should be noted that for this region,
$h_2'(\lambda)$ has no real roots, it is no longer possible to choose the branch cut $\Sigma_{g}$ as the curve of $\Im(h_{2}(\lambda))$. We set $\Sigma_g$ as arbitrary segments connecting the branch points of $R_{2}(\lambda)$.
The left one gives the jump contour for $\mathbf{S}_{1}(\lambda; \chi, \tau)$ and the right panel is the corresponding jump contour for $\mathbf{T}_{1}(\lambda; \chi, \tau)$.}
\label{fig:genus-two}
\end{figure}

Next, we begin to deform this RHP \ref{rhp:reform}. Set
\begin{equation}\label{S1-genus-two}
\mathbf{S}_{1}(\lambda; \chi, \tau):=\left\{\begin{aligned}&\mathbf{M}^{[n]}(\lambda; \chi, \tau)\ee^{-\ii n
\vartheta(\lambda; \chi, \tau)\sigma_3}\mathbf{Q}^{-1}_d
\ee^{\ii n\vartheta(\lambda; \chi, \tau)\sigma_3},\quad &\lambda\in D_{0}\cap\left(D_{1}^{+}\cup D_{1}^{-}\right)^{c},\\
&\mathbf{M}^{[n]}(\lambda; \chi, \tau),\quad &\text{otherwise},
\end{aligned}\right.
\end{equation}
where $D_{1}^{\pm}=R_{1,\Sigma}^{\pm}\cup\Omega_1^{\pm}\cup R_{1,\Gamma}^{\pm}\cup R^{\pm}\cup \Omega^{\pm}, $ 
then the jump of $\mathbf{S}_{1}(\lambda; \chi, \tau)$ transfers to $\partial D_{1}^{\pm}$ and $\Sigma_g$.
Since $\vartheta_{+}(\lambda; \chi, \tau)-\vartheta_{-}(\lambda; \chi, \tau)=-r\frac{\pi}{2n}$ for $\lambda\in\Sigma_g$, the jumps of $\mathbf{S}_1(\lambda; \chi, \tau)$ become
\begin{equation}\label{jump-S1}
\begin{split}
\mathbf{S}_{1,+}(\lambda; \chi, \tau)&=\mathbf{S}_{1,-}(\lambda; \chi, \tau)\ee^{-\ii n \vartheta_{-}(\lambda; \chi, \tau)\sigma_3}\begin{bmatrix}0&\ii r\\\ii r &0
	\end{bmatrix}\ee^{\ii n\vartheta_{+}(\lambda; \chi, \tau)\sigma_3},\quad \lambda\in\Sigma_{g},\\
\mathbf{S}_{1,+}(\lambda; \chi, \tau)&=\mathbf{S}_{1,-}(\lambda; \chi, \tau)\ee^{-\ii n\vartheta(\lambda; \chi, \tau)}\mathbf{Q}_{d}^{-1}\ee^{\ii n\vartheta(\lambda; \chi, \tau)},\qquad \qquad\lambda\in\partial D_{1}^{\pm}.
\end{split}
\end{equation}
With the theory of nonlinear steepest-descent method, we continue to define the sectional analytic matrices with the sign of $\Im(h_{2}(\lambda))$.
Define
\begin{equation}\label{eq:S-genus-two}
	\begin{aligned}
		\mathbf{T}_{1}(\lambda; \chi, \tau):&=\mathbf{S}_{1}(\lambda; \chi, \tau)\ee^{-\ii n\vartheta(\lambda; \chi, \tau)\sigma_3}
		\left(\mathbf{Q}_{R}^{[2]}\right)^{-1}\ee^{\ii n\vartheta(\lambda; \chi, \tau)\sigma_3}\ee^{\ii n g_{2}(\lambda)\sigma_3},\quad &\lambda\in L_{1,\Gamma}^{+}\cup L^{+},\\
		\mathbf{T}_{1}(\lambda; \chi, \tau):&=\mathbf{S}_{1}(\lambda; \chi, \tau)\mathbf{Q}_{L}^{[2]}\ee^{-\ii n\vartheta(\lambda; \chi, \tau)\sigma_3}
		\mathbf{Q}_{C}^{[2]}\ee^{\ii n\vartheta(\lambda; \chi, \tau)\sigma_3}\ee^{\ii n g_{2}(\lambda)\sigma_3},\quad &\lambda\in R_{1,\Gamma}^{+}\cup R^{+},\\
		\mathbf{T}_{1}(\lambda; \chi, \tau):&=\mathbf{S}_{1}(\lambda; \chi, \tau)\mathbf{Q}_{L}^{[2]}\ee^{\ii n g_{2}(\lambda)\sigma_3},\quad &\lambda\in \Omega_{1}^{+}\cup\Omega,\\
		\mathbf{T}_{1}(\lambda; \chi, \tau):&=\mathbf{S}_{1}(\lambda; \chi, \tau)\mathbf{Q}_{L}^{[2]}\ee^{-\ii n\vartheta(\lambda; \chi, \tau)\sigma_3}\mathbf{Q}_{L}^{[3]}\ee^{\ii n\vartheta(\lambda; \chi, \tau)\sigma_3}\ee^{\ii n g_{2}(\lambda)\sigma_3},\quad &\lambda\in R_{1,\Sigma}^{+},\\
		\mathbf{T}_{1}(\lambda; \chi, \tau):&=\mathbf{S}_{1}(\lambda; \chi, \tau)\ee^{-\ii n\vartheta(\lambda; \chi, \tau)\sigma_3}
		\left(\mathbf{Q}_{R}^{[3]}\right)^{-1}\ee^{\ii n\vartheta(\lambda; \chi, \tau)\sigma_3}\ee^{\ii n g_{2}(\lambda)\sigma_3},\quad &\lambda\in L_{1,\Sigma}^{+},\\
		\mathbf{T}_{1}(\lambda; \chi, \tau):&=\mathbf{S}_{1}(\lambda; \chi, \tau)\ee^{-\ii n\vartheta(\lambda; \chi, \tau)\sigma_3}
		\left(\mathbf{Q}_{R}^{[1]}\right)^{-1}\ee^{\ii n\vartheta(\lambda; \chi, \tau)\sigma_3}\ee^{\ii n g_{2}(\lambda)\sigma_3},\quad &\lambda\in L_{1,\Gamma}^{-}\cup L^{-},\\
		\mathbf{T}_{1}(\lambda; \chi, \tau):&=\mathbf{S}_{1}(\lambda; \chi, \tau)\mathbf{Q}_{L}^{[1]}\ee^{-\ii n\vartheta(\lambda; \chi, \tau)\sigma_3}
		\mathbf{Q}_{C}^{[1]}
		\ee^{\ii n\vartheta(\lambda; \chi, \tau)\sigma_3}\ee^{\ii n g_{2}(\lambda)\sigma_3},\quad &\lambda\in R_{1,\Gamma}^{-}\cup R^{-},\\
		\mathbf{T}_{1}(\lambda; \chi, \tau):&=\mathbf{S}_{1}(\lambda; \chi, \tau)\mathbf{Q}_{L}^{[1]}\ee^{\ii n g_{2}(\lambda)\sigma_3},\quad &\lambda\in \Omega_{1}^{-}\cup\Omega^{-},\\
		\mathbf{T}_{1}(\lambda; \chi, \tau):&=\mathbf{S}_{1}(\lambda; \chi, \tau)\mathbf{Q}_{L}^{[1]}\ee^{-\ii n\vartheta(\lambda; \chi, \tau)\sigma_3}
		\mathbf{Q}_{L}^{[4]}\ee^{\ii n\vartheta(\lambda; \chi, \tau)\sigma_3}\ee^{\ii n g_{2}(\lambda)\sigma_3},\quad &\lambda\in R_{1,\Sigma}^{-},\\
		\mathbf{T}_{1}(\lambda; \chi, \tau):&=\mathbf{S}_{1}(\lambda; \chi, \tau)\ee^{-\ii n\vartheta(\lambda; \chi, \tau)\sigma_3}
		\left(\mathbf{Q}_{R}^{[4]}\right)^{-1}\ee^{\ii n\vartheta(\lambda; \chi, \tau)\sigma_3}\ee^{\ii n g_{2}(\lambda)\sigma_3},\quad &\lambda\in L_{1,\Sigma}^{-},\\
\mathbf{T}_{1}(\lambda; \chi, \tau):&=\mathbf{S}_{1}(\lambda; \chi, \tau)\ee^{\ii n g_{2}(\lambda)\sigma_3},\quad &\text{otherwise}.
	\end{aligned}
\end{equation}
Then the jump conditions of $\mathbf{T}_{1}(\lambda; \chi, \tau)$ change into
\begin{equation}\label{eq:jump-genus-two}
	\begin{aligned}
		\mathbf{T}_{1, +}(\lambda; \chi, \tau)&=\mathbf{T}_{1, -}(\lambda; \chi, \tau)\ee^{-\ii nh_{2}(\lambda)\sigma_3}
		\mathbf{Q}_{R}^{[2]}\ee^{\ii nh_{2}(\lambda)\sigma_3},\quad &\lambda\in C_{L_{1,\Gamma}}^{+}\cup C_{L}^{+},\\
		\mathbf{T}_{1, +}(\lambda; \chi, \tau)&=\mathbf{T}_{1, -}(\lambda; \chi, \tau)\ee^{-\ii nh_{2}(\lambda)\sigma_3}\mathbf{Q}_{C}^{[2]}\ee^{\ii nh_{2}(\lambda)\sigma_3},\quad &\lambda\in C_{R_{1,\Gamma}}^{+}\cup C_{R}^{+},\\
		\mathbf{T}_{1, +}(\lambda; \chi, \tau)&=\mathbf{T}_{1, -}(\lambda; \chi, \tau)\ee^{-\ii nh_{2}(\lambda)\sigma_3}\mathbf{Q}_{L}^{[3]}\ee^{\ii nh_{2}(\lambda)\sigma_3},\quad &\lambda\in C_{R_{1,\Sigma}}^{+},\\
		\mathbf{T}_{1, +}(\lambda; \chi, \tau)&=\mathbf{T}_{1, -}(\lambda; \chi, \tau)\ee^{-\ii nh_{2}(\lambda)\sigma_3}\mathbf{Q}_{R}^{[3]}\ee^{\ii nh_{2}(\lambda)\sigma_3},\quad &\lambda\in C_{L_{1,\Sigma}}^{+},\\
		\mathbf{T}_{1, +}(\lambda; \chi, \tau)&=\mathbf{T}_{1, -}(\lambda; \chi, \tau)\ee^{-\ii nh_{2}(\lambda)\sigma_3}\mathbf{Q}_{R}^{[1]}\ee^{\ii nh_{2}(\lambda)\sigma_3},\quad &\lambda\in C_{L_{1,\Gamma}}^{-}\cup C_{L}^{-},\\
		\mathbf{T}_{1, +}(\lambda; \chi, \tau)&=\mathbf{T}_{1, -}(\lambda; \chi, \tau)\ee^{-\ii nh_{2}(\lambda)\sigma_3}\mathbf{Q}_{C}^{[1]}\ee^{\ii nh_{2}(\lambda)\sigma_3},\quad &\lambda\in C_{R_{1,\Gamma}}^{-}\cup C_{R}^{-},\\
		\mathbf{T}_{1, +}(\lambda; \chi, \tau)&=\mathbf{T}_{1, -}(\lambda; \chi, \tau)\ee^{-\ii nh_{2}(\lambda)\sigma_3}\mathbf{Q}_{L}^{[4]}\ee^{\ii nh_{2}(\lambda)\sigma_3},\quad &\lambda\in C_{R_{1,\Sigma}}^{-}\cup C_{R}^{-},\\
		\mathbf{T}_{1, +}(\lambda; \chi, \tau)&=\mathbf{T}_{1, -}(\lambda; \chi, \tau)\ee^{-\ii nh_{2}(\lambda)\sigma_3}\mathbf{Q}_{R}^{[4]}\ee^{\ii nh_{2}(\lambda)\sigma_3},\quad &\lambda\in C_{L_{1,\Sigma}}^{-},\\
		\mathbf{T}_{1, +}(\lambda; \chi, \tau)&=\mathbf{T}_{1, -}(\lambda; \chi, \tau)\begin{bmatrix}0&\ee^{-\ii n\kappa_2}\\
			-\ee^{\ii n\kappa_2}&0
		\end{bmatrix},&\lambda\in \Sigma_{g_{2}}^{\pm},\\
	\mathbf{T}_{1, +}(\lambda; \chi, \tau)&=\mathbf{T}_{1, -}(\lambda; \chi, \tau)\begin{bmatrix}0&\ee^{-\ii nl_2+\ii r\frac{\pi}{2}}\\
		-\ee^{\ii nl_2-\ii r\frac{\pi}{2}}&0
	\end{bmatrix},&\lambda\in \Sigma_{g},\\
		\mathbf{T}_{1, +}(\lambda; \chi, \tau)&=\mathbf{T}_{1, -}(\lambda; \chi, \tau)\begin{bmatrix}\ee^{\ii n\varpi_2}&0\\
			0&\ee^{-\ii n\varpi_2}
		\end{bmatrix},&\lambda\in \Gamma_{g_{2}}^{\pm}.
	\end{aligned}
\end{equation}
From the sign chart of $\Im(h_{2}(\lambda))$ in Fig. \ref{fig:genus-two} and the definition in Eq.\eqref{remark:decom}, when $n$ is large, the jump matrices will converge to the identity matrix exponentially except for the contours $\Sigma_{g_{2}}^{\pm}\cup \Sigma_{g}\cup \Gamma_{g_{2}}^{\pm}$.
Next, we will construct the parametrix to give the asymptotic analysis in the genus-two region.

\subsection{Parametrix construction for $\mathbf{T}_{1}(\lambda; \chi, \tau)$}
From the jump conditions in Eq.\eqref{eq:jump-genus-two}, we construct the outer parametrix  $\dot{\mathbf{T}}_{1}^{\rm out}(\lambda; \chi, \tau)$ satisfying the following RHP.
\begin{rhp}
	(RHP for the outer parametrix $\dot{\mathbf{T}}^{\rm out}_{1}(\lambda; \chi, \tau)$) Find a $2\times 2$ matrix $\dot{\mathbf{T}}_{1}^{\rm out}(\lambda; \chi, \tau)$ satisfying the following conditions.
	\begin{itemize}
		\item {\bf Analyticity:} $\dot{\mathbf{T}}_{1}^{\rm out}(\lambda; \chi, \tau)$ is analytic in $\lambda\in \mathbb{C}\setminus\left(\Sigma_{g_{2}}^{\pm}\cup \Sigma_{g}\cup \Gamma_{g_{2}}^{\pm}\right)$.
		\item {\bf Jump condition:} The boundary values on the contours $\left(\Sigma_{g_{2}}^{\pm}\cup \Sigma_{g}\cup \Gamma_{g_{2}}^{\pm}\right)$ are related by
$\dot{\mathbf{T}}_{1,+}^{\rm out}(\lambda; \chi, \tau)=\dot{\mathbf{T}}_{1,-}^{\rm out}(\lambda; \chi, \tau)\mathbf{V}_{\dot{\mathbf{T}}_{1}^{\rm out}}(\lambda; \chi, \tau)$, where $\mathbf{V}_{\dot{\mathbf{T}}_{1}^{\rm out}}(\lambda; \chi, \tau)$ is
		\begin{equation}
			\mathbf{V}_{\dot{\mathbf{T}}_{1}^{\rm out}}(\lambda; \chi, \tau)=\left\{\begin{aligned}&\begin{bmatrix}0&\ee^{-\ii n\kappa_2}\\
					-\ee^{\ii n\kappa_2}&0
				\end{bmatrix},&\lambda\in \Sigma_{g_{2}}^{\pm},\\
			&\begin{bmatrix}0&\ee^{-\ii nl_2+\ii r\frac{\pi}{2}}\\
				-\ee^{\ii nl_2-\ii r\frac{\pi}{2}}&0
			\end{bmatrix},&\lambda\in \Sigma_{g},\\
				&\begin{bmatrix}\ee^{\ii n\varpi_2}&0\\
					0&\ee^{-\ii n\varpi_2}
				\end{bmatrix},&\lambda\in \Gamma_{g_{2}}^{\pm}.
			\end{aligned}\right.
		\end{equation}
		\item {\bf Normalization:}  $\dot{\mathbf{T}}_{1}^{\rm out}(\lambda; \chi, \tau)\to \mathbb{I}$ as $\lambda\to\infty$.
	\end{itemize}
\end{rhp}
To solve this RHP, we introduce a scalar function $F(\lambda; \chi, \tau)$ with the following conditions,
\begin{equation}
	\begin{aligned}
		&F_{+}(\lambda; \chi, \tau)+F_{-}(\lambda; \chi, \tau)=\ii n\kappa_2,\quad&\lambda\in \Sigma_{g_{2}}^{\pm},\\
		&F_{+}(\lambda; \chi, \tau)+F_{-}(\lambda; \chi, \tau)=\ii nl_2-\ii r\frac{\pi}{2},\quad&\lambda\in \Sigma_{g},\\
		&F_{+}(\lambda; \chi, \tau)-F_{-}(\lambda; \chi, \tau)=\ii n\varpi_2, \quad&\lambda\in\Gamma_{g_{2}}^{\pm}.
	\end{aligned}
\end{equation}
With the Plemelj formula, $F(\lambda; \chi, \tau)$ can be expressed into an integral form,
\begin{equation}\label{eq:F4}
	F(\lambda;\chi,\tau)=\frac{R_{2}(\lambda)}{2\pi\ii}\Bigg[\int_{\Sigma_{g_{2}}^{\pm}}\frac{\ii n\kappa_2}{R_{2}(\xi)(\xi-\lambda)}d\xi
+\int_{\Sigma_{g}}\frac{\ii nl_2-\ii r\frac{\pi}{2}}{R_{2}(\xi)(\xi-\lambda)}d\xi+\int_{\Gamma_{g_{2}}^{\pm}}\frac{\ii n\varpi_2}{R_{2}(\xi)(\xi-\lambda)}d\xi\Bigg].
\end{equation}
When $\lambda\to\infty$, we easily get the following expansion formula,
\begin{equation}
	F(\lambda; \chi, \tau)=F_{2}\lambda^2+F_{1}\lambda+F_{0}+\mathcal{O}(\lambda^{-1}),
\end{equation}
where
\begin{equation}\label{eq:F4342}
	\begin{aligned}
		F_{2}&=-\frac{1}{2\pi\ii}\left(\int_{\Sigma_{g_{2}}^{\pm}}\frac{\ii n\kappa_2}{R_{2}(\xi)}d\xi+\int_{\Sigma_{g}}\frac{\ii nl_2-\ii r\frac{\pi}{2}}{R_{2}(\xi)}d\xi
+\int_{\Gamma_{g_{2}}^{\pm}}\frac{\ii n\varpi_2}{R_{2}(\xi)}d\xi\right),\\
		F_{1}&=-\frac{1}{2\pi\ii}\left(\int_{\Sigma_{g_{2}}^{\pm}}\frac{\ii n\kappa_2}{R_{2}(\xi)}\xi d\xi
{+}\int_{\Sigma_{g}}\frac{\ii nl_2{-}\ii r\frac{\pi}{2}}{R_{2}(\xi)}\xi d\xi{+}\int_{\Gamma_{g_{2}}^{\pm}}\frac{\ii n\varpi_2}{R_{2}(\xi)}\xi d\xi\right)-\frac{s_1}{2}F_{2},\\
		F_{0}&=-\frac{1}{2\pi\ii}\left(\int_{\Sigma_{g_{2}}^{\pm}}\frac{\ii n\kappa_2}{R_{2}(\xi)}\xi^2d\xi{+}
\int_{\Sigma_{g}}\frac{\ii nl_2{-}\ii r\frac{\pi}{2}}{R_{2}(\xi)}\xi^2d\xi{+}\int_{\Gamma_{g_{2}}^{\pm}}\frac{\ii n\varpi_2}{R_{2}(\xi)}\xi^2d\xi\right)-\frac{s_1}{2}F_{1}+\left(\frac{s_2}{2}-\frac{3s_1^2}{8}\right)F_{2}.
	\end{aligned}
\end{equation}
Based on the definition of this scalar function $F(\lambda; \chi, \tau)$, we redefine a new matrix $\mathbf{O}_1(\lambda; \chi, \tau)$,
\begin{equation}
	\mathbf{O}_1(\lambda; \chi, \tau)={\rm diag}\left(\ee^{F_{0}}, \ee^{-F_{0}}\right)\dot{\mathbf{T}}_{1}^{\rm out}(\lambda; \chi, \tau){\rm diag}\left(\ee^{-F(\lambda; \chi, \tau)}, \ee^{F(\lambda; \chi, \tau)}\right).
\end{equation}
It is clear that $\mathbf{O}_{1}(\lambda; \chi, \tau)$ satisfies a simple constant jump condition at $\lambda\in\Sigma_{g_{2}}^{\pm}\cup\Sigma_{g}$,
\begin{equation}
	\mathbf{O}_{1,+}(\lambda; \chi, \tau)=\mathbf{O}_{1,-}(\lambda; \chi, \tau)(\ii\sigma_2),\quad \lambda\in\Sigma_{g_{2}}^{\pm}\cup\Sigma_{g}.
\end{equation}
When $\lambda\to\infty$, $\mathbf{O}_{1}(\lambda; \chi, \tau)$ has the following normalization condition,
\begin{equation}\label{eq:O-genus-two}
	\mathbf{O}_{1}(\lambda; \chi, \tau){\rm{diag}}\left(\ee^{F_1\lambda+F_2\lambda^2}, \ee^{-F_1\lambda-F_2\lambda^2}\right)\to\mathbb{I} \quad \text{as}\quad \lambda\to\infty.
\end{equation}
Before solving this RHP, we give the definition of the Riemann-Theta function.
\begin{definition}\label{prop:theta}
	The $\Theta(u)$ function is defined as \cite{belokolos1994algebro}
	\begin{equation}
		\begin{aligned}
			&\Theta(u)\equiv\Theta(u; \mathbf{B}):=\sum\limits_{\mathbf{m}\in\mathbb{Z}^g}\ee^{\frac{1}{2}\langle \mathbf{m}, \mathbf{B}\mathbf{m}\rangle+\langle\mathbf{m},u\rangle},
		\end{aligned}
	\end{equation}
where $\mathbf{B}$ is a period matrix, and $\Theta(u)$ function has the following periodic properties,
\begin{equation}
	\Theta(u+2\pi\ii\mathbf{e}_j)=\Theta(u),\qquad \Theta(u+\mathbf{B}\mathbf{e}_j)=\ee^{-\frac{1}{2}B_{jj}-u_j}\Theta(u),
	\end{equation}
	where $\mathbf{e}_j$s are the unit basis vectors in $\mathbb{C}^{g}$ with the coordinates $(\mathbf{e}_{j})_{k}=\delta_{jk}$,
and $\mathbf{B}\mathbf{e}_{j}$s indicate the $j{-}th$ column of the period matrix $\mathbf{B}$.
\end{definition}

The square root function $R_{2}(\lambda)$ in Eq.\eqref{eq:R1} is related to a genus-two Riemann surface, and we give a homology basis for it in Fig.\ref{circle:genus-two}.
				\begin{figure}[ht]
						\centering
						\includegraphics[width=0.3\textwidth,angle=-90]{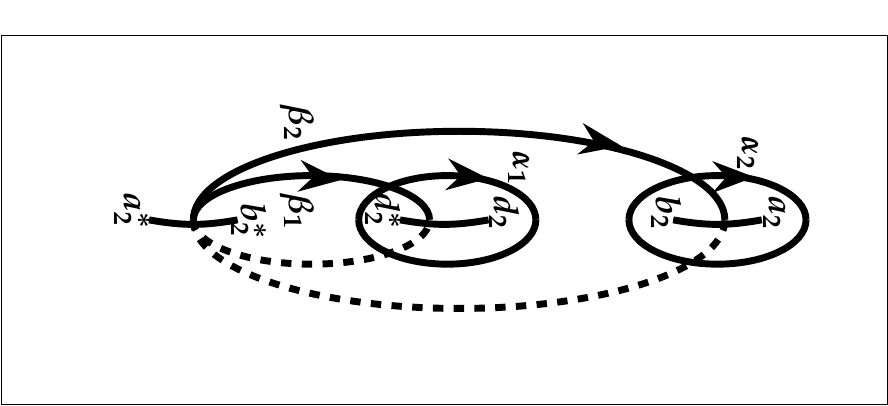}
						\caption{Homology basis for the Riemann surface of genus-two. The solid paths indicate the first sheet and the dashed lines lie in the second sheet.}
						\label{circle:genus-two}
					\end{figure}

Now, we introduce the Abel integrals for the genus-two region,
\begin{equation}
	\omega_{j}(\lambda)=\int_{a_{2}^*}^{\lambda}\psi_{j}(\xi)d\xi,\quad j=1,2,\quad \psi_{j}(\xi):=\frac{\sum\limits_{i=1}^{2}c_{ji}\xi^{2-i}}{R_{2}(\xi)}.
\end{equation}
The coefficients $c_{ji}$s can be uniquely determined by the following conditions,
\begin{equation}
	\int_{\alpha_{l}}d\omega_{j}(\mathcal{P})=2\pi\ii\delta_{jl},\quad (j,l=1,2),
\end{equation}
where $d\omega_{j}(\mathcal{P})$ is a holomorphic differential on the Riemann surface. On the basis of $\omega_{j}(\lambda)(j=1,2)$, we can calculate the period $\mathbf{B}$ matrix,
\begin{equation}
	B_{jl}=\int_{\beta_{l}}d\omega_{j}(\mathcal{P}),\quad (j,l=1,2),
\end{equation}
which is a symmetric matrix, and its real part is negative definite.
Then we can define the Abel mapping from the Riemann surface $\pmb{\chi}$ to its Jacobian variety $Jac\{\pmb{\chi}\}=\mathbb{C}^{2}/\Lambda$, $\mathbf{A}:\pmb{\chi}\to Jac\{\pmb{\chi}\}$,
\begin{equation}\label{eq:abel-map}
	A_{j}(\mathcal{P})=\int_{\mathcal{P}_{0}}^{\mathcal{P}}d\omega_j(\mathcal{Q}),\quad j=1,2,
\end{equation}
where $\Lambda$ is the lattice defined by
\begin{equation}
	\Lambda=\{2\pi\ii N+\mathbf{B}M, N, M\in\mathbb{Z}^{2}\},
\end{equation}
and the point  $\mathcal{P}_{0}$ is given from the base point $a_2^*$ satisfying the condition $\pi(\mathcal{P}_{0})=a_2^*$ and $\mathcal{Q}$ is the integration variable.
With the definition of Abel mapping $\mathbf{A}$, for the integral divisors $\mathcal{D}=\mathcal{P}_{1}+\mathcal{P}_{2}$, we have the following identity relationship
\begin{equation}
	\mathbf{A}(\mathcal{D})=\mathbf{A}(\mathcal{P}_1)+\mathbf{A}(\mathcal{P}_2).
\end{equation}
For $\lambda$ in the branch cuts or the gaps (shown in the blue and red lines in Fig.\ref{fig:genus-two}),
the Abel integrals $\mathbf{A}(\lambda)$ satisfies the following conditions,
\begin{equation}
	\begin{aligned}
		\mathbf{A}_{+}(\lambda)-\mathbf{A}_{-}(\lambda)&=0\quad \mod 2\pi\ii \mathbb{Z}^2, \quad \lambda\in \left(b_2^*, d_2^*\right)\cup\left(b_2, d_2\right),\\
		\mathbf{A}_{+}(\lambda)+\mathbf{A}_{-}(\lambda)&=\mathbf{B}\mathbf{e}_1\mod 2\pi\ii \mathbb{Z}^2, \quad \lambda\in \left(d_2^*, d_2\right),\\
		\mathbf{A}_{+}(\lambda)+\mathbf{A}_{-}(\lambda)&=\mathbf{B}\mathbf{e}_2\mod 2\pi\ii \mathbb{Z}^3, \quad \lambda\in \left(b_2, a_2\right).
	\end{aligned}
\end{equation}
Next, we will introduce another Abel integrals with the singularities at the point $P_{\infty^{+}}$,
\begin{equation}
	\begin{aligned}
		\Omega_j(\lambda)=\int_{a_2^*}^{\lambda}\Psi_j(\xi)d\xi,\quad j=1,2,\quad 	\Psi_j(\xi)=\frac{\sum\limits_{i=1}^{5}s_{ji}\xi^{5-i}}{R_{2}(\xi)},
	\end{aligned}
\end{equation}
these Abel integrals satisfy the following normalization condition,
\begin{equation}
	\begin{aligned}
		&\Omega_1(\lambda)\to \lambda+\mathcal{O}(1),\quad \Omega_2(\lambda)\to\lambda^2+\mathcal{O}(1),\quad P\to P_{\infty^+},\\
		&\int_{\alpha_{l}}d\Omega_{j}(\mathcal{P})=0,\quad j,l=1,2,
	\end{aligned}
\end{equation}
which can determine the unknown coefficients $s_{ji}$s uniquely. For these Abel integrals, the corresponding ``$\mathbf{B}$" matrix by integrating around the $\beta$ circles can be given as,
\begin{equation}\label{eq:UVW}
	\mathcal{U}_{j}=\int_{\beta_{j}}d\Omega_{1}(\mathcal{P}), \quad \mathcal{V}_{j}=\int_{\beta_{j}}d\Omega_{2}(\mathcal{P}), \quad j=1,2.
\end{equation}
Based on the normalization conditions of $\Omega_{j}(\lambda)(j=1,2)$, we get some important properties when $\lambda\to\infty$, one useful property for us is that the limits $J_{j}\,(j=1,2)$ defined as
\begin{equation}\label{eq:J-cons}
	J_{1}:=\lim\limits_{\lambda\to\infty}\int_{a_2^*}^{\lambda}d\Omega_{1}(\mathcal{P})-\lambda,\quad J_{2}:=\lim\limits_{\lambda\to\infty}\int_{a_2^*}^{\lambda}d\Omega_{2}(\mathcal{P})-\lambda^2
\end{equation}
are existent.

Before solving the RHP for the matrix $\mathbf{O}_1(\lambda; \chi, \tau)$, we now introduce an auxiliary matrix $\mathbf{P}_1(\lambda; \chi, \tau)$ defined as
\begin{equation}
	\mathbf{P}_1(\lambda; \chi, \tau):=\begin{bmatrix}\frac{\Theta\left(\mathbf{A}(\lambda)+\mathbf{d}-\pmb{\mathcal{U}}F_{1}-\pmb{\mathcal{V}}F_{2}\right)}{\Theta\left(\mathbf{A}(\lambda)+\mathbf{d}\right)}&
		 \frac{\Theta\left(\mathbf{A}(\lambda)-\mathbf{d}+\pmb{\mathcal{U}}F_{1}+\pmb{\mathcal{V}}F_{2}\right)}{\Theta\left(\mathbf{A}(\lambda)-\mathbf{d}\right)}\\
		 \frac{\Theta\left(\mathbf{A}(\lambda)-\mathbf{d}-\pmb{\mathcal{U}}F_{1}-\pmb{\mathcal{V}}F_{2}\right)}{\Theta\left(\mathbf{A}(\lambda)-\mathbf{d}\right)}&
		 \frac{\Theta\left(\mathbf{A}(\lambda)+\mathbf{d}+\pmb{\mathcal{U}}F_{1}+\pmb{\mathcal{V}}F_{2}\right)}{\Theta\left(\mathbf{A}(\lambda)+\mathbf{d}\right)}\end{bmatrix}
	\ee^{-\left(\Omega_1(\lambda)F_{1}+\Omega_2(\lambda)F_{2}\right)\sigma_3}.
\end{equation}
With the properties of Theta function and the Abel integrals, the matrix function $\mathbf{P}_{1}(\lambda; \chi, \tau)$ satisfies the following jump conditions,
\begin{equation}
	\mathbf{P}_{1,+}(\lambda; \chi, \tau)=\mathbf{P}_{1,-}(\lambda; \chi, \tau)\begin{bmatrix}0&1\\
		1&0
	\end{bmatrix},\quad\lambda\in\Sigma_{g_{2}}^{\pm}\cup\Sigma_{g}.
\end{equation}
By using $\mathbf{P}_{1}(\lambda; \chi, \tau)$, we can construct the solutions to $\mathbf{O}_{1}(\lambda; \chi, \tau)$, that is
\begin{equation}
	\mathbf{O}_{1}(\lambda; \chi, \tau):=\frac{1}{2}{\rm diag}\left(C_{1}, C_{2}\right)\begin{bmatrix}\left(\gamma(\lambda)+\frac{1}{\gamma(\lambda)}\right)\mathbf{P}_{1}(\lambda; \chi, \tau)_{11}&
\ii \left(\gamma(\lambda)-\frac{1}{\gamma(\lambda)}\right)\mathbf{P}_{1}(\lambda; \chi, \tau)_{12}\\
		-\ii \left(\gamma(\lambda)-\frac{1}{\gamma(\lambda)}\right)\mathbf{P}_{1}(\lambda; \chi, \tau)_{21}&\left(\gamma(\lambda)+\frac{1}{\gamma(\lambda)}\right)\mathbf{P}_{1}(\lambda; \chi, \tau)_{22}
	\end{bmatrix},
\end{equation}
where
$C_{1}$ and $C_{2}$ are two constants determined from the normalization condition in Eq.\eqref{eq:O-genus-two},
\begin{equation}
	\begin{aligned}
		 C_{1}&=\frac{\Theta\left(\mathbf{A}(\infty)+\mathbf{d}\right)}{\Theta\left(\mathbf{A}(\infty)+\mathbf{d}-\pmb{\mathcal{U}}F_{1}-\pmb{\mathcal{V}}F_{2}\right)}\ee^{J_{1}F_{1}+J_{2}F_{2}},\\
		 C_{2}&=\frac{\Theta\left(\mathbf{A}(\infty)+\mathbf{d}\right)}{\Theta\left(\mathbf{A}(\infty)+\mathbf{d}+\pmb{\mathcal{U}}F_{1}+\pmb{\mathcal{V}}F_{2}\right)}\ee^{-J_{1}F_{1}-J_{2}F_{2}},\\
	\end{aligned}
\end{equation}
and $\gamma(\lambda)=\left(\frac{(\lambda-a_2^*)(\lambda-b_2)(\lambda-d_2^*)}{(\lambda-a_2)(\lambda-b_2^*)(\lambda-d_2)}\right)^{\frac{1}{4}}$
satisfies $\gamma_{+}=-\ii\gamma_{-}$. Suppose $\gamma-\frac{1}{\gamma}$ has two zeros $\mathcal{P}_{1}, \mathcal{P}_{2}$
at the first sheet Riemann surface. Then the constant matrix $\mathbf{d}$ can be given by the following formula,
\begin{equation}\label{eq:d}
	\mathbf{d}=\mathbf{A}\left(\mathcal{D}\right)+\mathbf{K},
\end{equation}
where $\mathbf{K}$ is the Riemann-Theta constant vector\cite{belokolos1994algebro,kotlyarov2017planar}, defined as
\begin{equation}
	K_{j}=\frac{2\pi\ii+B_{jj}}{2}-\frac{1}{2\pi\ii}\sum\limits_{l=1,l\neq j}^{2}\int_{\alpha_{l}}\left(\int_{\mathcal{P}_{0}}^{\mathcal{Q}}\omega_{j}\right)\omega_l (j=1,2).
\end{equation}
The integral formula seems much more complicated, but in our hyperelliptic case, the entries $K_js$ are also equal to another simple formula,
\begin{equation}
	K_j=\frac{1}{2}\sum\limits_{l=1}^{2}B_{lj}+\pi\ii\left(j-2\right).
\end{equation}
Then the outer parametrix $\dot{\mathbf{T}}_{1}^{\rm out}(\lambda; \chi, \tau)$ is constructed completely. We hope that the outer parametrix can match $\mathbf{T}_{1}(\lambda; \chi, \tau)$ very well,
but unfortunately, the outer parametrix has singularities at the branch points $a_{2}, b_{2}, d_{2}, a_{2}^{*}, b_{2}^{*}, d_{2}^{*}$.
Thus we should construct the inner parametrices at the neighbourhood of these points. In Refs.\cite{Bilman-JDE-2021,Bilman-arxiv-2021} and our previous article\cite{ling2022large}, there is a detailed calculation for
constructing the inner parametrices, and the results show that these inner parametrices are related to the Airy function, and the error is $\mathcal{O}(n^{-1})$. In this paper, we omit the details and only give some notations.
Then the global parametrix for $\mathbf{T}_{1}(\lambda; \chi, \tau)$ is written as
\begin{equation}
	\dot{\mathbf{T}}_{1}(\lambda; \chi, \tau):=\left\{\begin{aligned}&\dot{\mathbf{T}}_{1}^{a_2}(\lambda; \chi, \tau),\quad\lambda\in D_{a_2}(\delta),\\&\dot{\mathbf{T}}_{1}^{a_2^*}(\lambda; \chi, \tau),\quad\lambda\in D_{a_2^*}(\delta),
\\&\dot{\mathbf{T}}_{1}^{b_2}(\lambda; \chi, \tau),\quad\lambda\in D_{b_2}(\delta),\\
		&\dot{\mathbf{T}}_{1}^{b_2^*}(\lambda; \chi, \tau),\quad\lambda\in D_{b_2^*}(\delta),\\
		&\dot{\mathbf{T}}_{1}^{d_2}(\lambda; \chi, \tau),\quad\lambda\in D_{d_2}(\delta),\\&\dot{\mathbf{T}}_{1}^{d_2^*}(\lambda; \chi, \tau),\quad\lambda\in D_{d_2^*}(\delta),\\
		&\dot{\mathbf{T}}_{1}^{\rm out}(\lambda; \chi, \tau),\quad\lambda\in\mathbb{C}\setminus\left(\overline{D_{a_2,a_2^*,b_2,b_2^*,d_2,d_2^*}(\delta)}\cup \Sigma_{g_{2}}^{\pm}\cup \Sigma_{g}\cup \Gamma_{g_{2}}^{\pm}\right).
	\end{aligned}\right.
\end{equation}
Then we will analyze the error between $\mathbf{T}_{1}(\lambda; \chi, \tau)$ and its parametrix $\dot{\mathbf{T}}_{1}(\lambda; \chi, \tau)$ in the next subsection.
\subsection{Error analysis}
To study the error, we set the error function $\mathcal{E}_1(\lambda; \chi, \tau)$ between $\mathbf{T}_{1}(\lambda; \chi, \tau)$ and $\dot{\mathbf{T}}_{1}(\lambda; \chi, \tau)$ as
\begin{equation}
	\mathcal{E}_{1}(\lambda; \chi, \tau):=\mathbf{T}_{1}(\lambda; \chi, \tau)\left(\dot{\mathbf{T}}_{1}(\lambda; \chi, \tau)\right)^{-1}.
\end{equation}
Set the jump matrix for the error function $\mathcal{E}_1(\lambda; \chi, \tau)$ as $\mathbf{V}_{\mathcal{E}_{1}}(\lambda; \chi, \tau)$.
In our previous work \cite{ling2022large}, we have given a detailed analysis for the error estimation. Following that result, we present the order of the error estimation,
\begin{equation}
	\begin{aligned}
		\|\mathbf{V}_{\mathcal{E}_{1}}(\lambda; \chi, \tau)-\mathbb{I}\|&=\mathcal{O}\left(\ee^{-\mu_1 n}\right)\,\,(\mu_1>0),\, \lambda\in C_{L_{1,\Sigma}}^{\pm}\cup C_{R_{1,\Sigma}}^{\pm}\cup C_{R_{1,\Gamma}}^{\pm}\cup C_{L_{1,\Gamma}}^{\pm}\cup C_{R}^{\pm}\cup C_{L}^{\pm},\\
		\|\mathbf{V}_{\mathcal{E}_{1}}(\lambda; \chi, \tau)-\mathbb{I}\|&=\mathcal{O}(n^{-1}),\quad \lambda\in\partial D_{a_2,a_2^*,b_2,b_2^*,d_2,d_2^*}(\delta).\\
	\end{aligned}
\end{equation}
Finally, we can recover the potential function $q^{[n]}(n\chi, n\tau)$ from $\mathbf{T}_{1}(\lambda; \chi, \tau)$
\begin{equation}\label{eq:qn-genus-2}
	\begin{aligned}
		q^{[n]}(n\chi, n\tau)&=2\ii r\lim\limits_{\lambda\to\infty}\lambda\mathbf{T}_{1}(\lambda; \chi, \tau)_{12}\\
		&=2\ii r\lim\limits_{\lambda\to\infty}\lambda\left(\mathcal{E}_{1}(\lambda; \chi, \tau)\dot{\mathbf{T}}^{\rm out}_{1}(\lambda; \chi, \tau)\right)_{12}\\
		&=2\ii r\lim\limits_{\lambda\to\infty}\lambda\left(\mathcal{E}_{1,11}(\lambda; \chi, \tau)\dot{\mathbf{T}}_{1,12}^{\rm out}(\lambda; \chi, \tau)+\mathcal{E}_{1,12}(\lambda; \chi, \tau)\dot{\mathbf{T}}_{1,22}^{\rm out}(\lambda; \chi, \tau)\right)\\
		&=2\ii r\lim\limits_{\lambda\to\infty}\lambda\dot{\mathbf{T}}_{1,12}^{\rm out}(\lambda; \chi, \tau)+\mathcal{O}(n^{-1}).
	\end{aligned}
\end{equation}
Substituting $\dot{\mathbf{T}}_1^{\rm out}(\lambda; \chi, \tau)$ into Eq.\eqref{eq:qn-genus-2}, then the asymptotic expression for the genus-two region is,
\begin{multline}\label{eq:qn-genus-3-1}
	q^{[n]}(n\chi, n\tau)=r\frac{\Theta\left(\mathbf{A}(\infty)+\mathbf{d}\right)}{\Theta\left(\mathbf{A}(\infty)+\mathbf{d}-\pmb{\mathcal{U}}F_{1}-\pmb{\mathcal{V}}F_{2}\right)}
	\frac{\Theta\left(\mathbf{A}(\infty)-\mathbf{d}+\pmb{\mathcal{U}}F_{1}+\pmb{\mathcal{V}}F_{2}\right)}{\Theta\left(\mathbf{A}(\infty)-\mathbf{d}\right)}\\
	\times \ii \left(\Im(b_2)-\Im(a_2)-\Im(d_2)\right)\ee^{2F_{1}J_{1}+2F_{2}J_{2}-2F_{0}}+\mathcal{O}(n^{-1}).
\end{multline}
With this expression, we check the asymptotic solution and the exact solution by choosing $\tau=\frac{1}{4}, c_1=c_2=1$ in Fig.\ref{fig:genus2}.
\begin{figure}[!ht]
\centering
\includegraphics[width=1\textwidth]{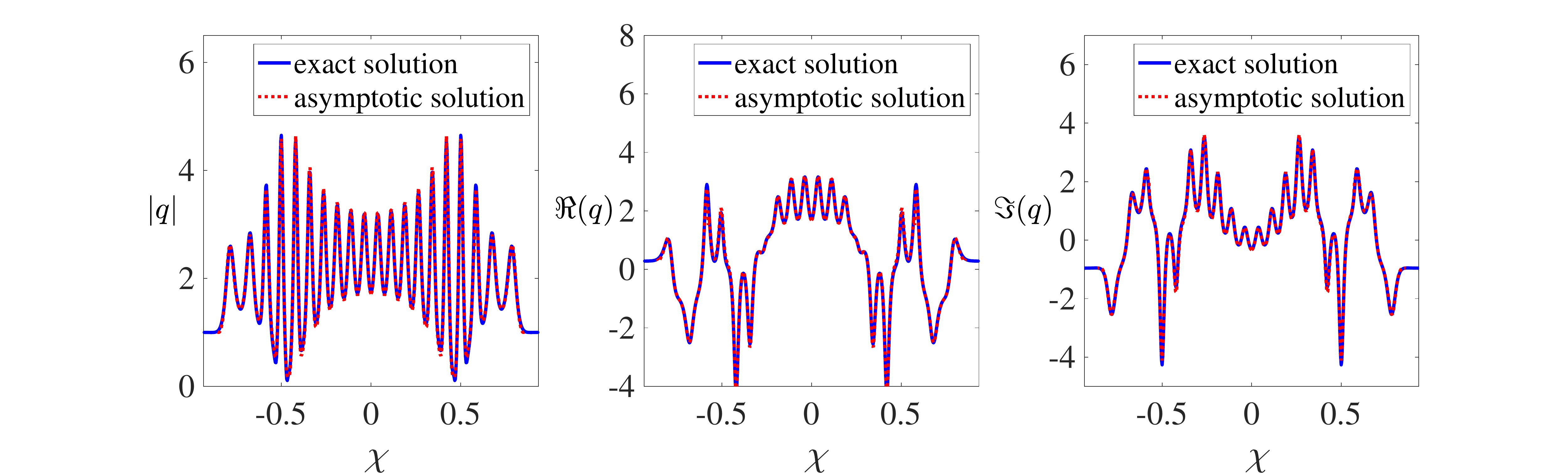}
\caption{The comparison between the exact solution ($20$-th order KMBs) and the asymptotics in genus-two region by choosing $\tau=\frac{1}{4}$ (as shown in the green dashed line in the right panel in Fig. \ref{fig:KM}), $c_1=c_2=1$. The left one is the modulus $q^{[n]}(n\chi, n\tau)$,
the middle is the real part, and the right is the imaginary part. It is seen that they are fitting very well.}
\label{fig:genus2}
\end{figure}

\section{Genus-zero-infinity and genus-zero-up regions}\label{sec:exp-dec}
In the last section, we have obtained the asymptotic expression in the genus-two region. And we check the asymptotics by choosing a group of parameters $\lambda_1=2\ii, c_1=c_2=1, \tau=\frac{1}{4}$.
For the rest of regions, we will choose another group of parameters $\lambda_1=\frac{3}{2}\ii$, $c_1=-c_2=1$, to verify their asymptotics numerically.
In this section, we first analyze the asymptotics in the genus-zero-infinity region($g_0^\infty$ in Fig.\ref{fig:KM}). For the large-order asymptotics of high-order solitons in the exponent-decay region \cite{Bilman-JDE-2021}, the leading order term exponentially decays to the zero background, in which the leading order term can be derived directly.  However, the leading order term in the genus-zero-infinity region will approach to the background wave $q^{[n]}(n\chi, n\tau)\to\ee^{\ii n\tau}$. Thus the corresponding asymptotic analysis for both two kinds of regions are different. For the high-order KMBs, there appears a new factor $\frac{\ii r}{4n}\left(\frac{\lambda-\ii}{\lambda+\ii}\right)^{1/4}$ in the phase term $\vartheta(\lambda; \chi, \tau)$, which brings a new branch cut in the vertical segment $\Sigma_c=[-\ii, \ii]$. Thus, the previous method in studying the large-order asymptotics of solitons can not apply to KMBs, and we need new skills to deal with this branch cut. A frequently-used way is to introduce a proper $g$-function. To this end, we give a RHP for the $g_{1}(\lambda; \chi, \tau)$-function in this section.
\begin{rhp}\label{RHP:EXP}
Let $(\chi, \tau)\in\mathbb{R}^2$, we can find a  $g_{1}(\lambda):=g_{1}(\lambda; \chi, \tau)$-function satisfying the following conditions:
\begin{itemize}
\item
{\bf Analyticity}: $g_{1}(\lambda)$ is analytic in $\mathbb{C}\setminus \Sigma_{g_{1}}$, where $\Sigma_{g_{1}}$ is a branch cut to be determined,
and it takes the continuous boundary conditions from the left and right sides of $\Sigma_{g_{1}}$.
\item
{\bf Jump Condition}: The boundary values on the jump contour $\Sigma_{g_{1}}$ are related by
\begin{equation}\label{eq:jump-genus-zero-1}
g_{1,+}(\lambda)+g_{1,-}(\lambda)+\vartheta_{+}(\lambda; \chi, \tau)+\vartheta_{-}(\lambda; \chi, \tau)=\kappa_{1},\qquad\lambda\in\Sigma_{g_{1}}.
\end{equation}
\item
{\bf Normalization}: As $\lambda\to\infty$, $g_{1}(\lambda)$ satisfies
\begin{equation}
g_{1}(\lambda)\to \mathcal{O}(\lambda^{-1}).
\end{equation}
\item
{\bf Symmetry}: $g_{1}(\lambda)$ satisfies the Schwartz symmetric condition:
\begin{equation}
g_{1}(\lambda)=g_{1}^*(\lambda^*).
\end{equation}
\end{itemize}
\end{rhp}
For this case, we still analyze its derivative to $\lambda$ and eliminate these two logarithmic terms and the integral constant $\kappa_1$. When $\lambda\in\Sigma_{g_{1}}$, we have
\begin{equation}
g'_{1,+}(\lambda)+g'_{1,-}(\lambda)=-2\chi-4\lambda\tau-\frac{\ii }{\lambda-\lambda_1}+\frac{\ii}{\lambda-\lambda_1^*}-\frac{1}{2n}\frac{\ii r}{\lambda-\ii}+\frac{1}{2n}\frac{\ii r}{\lambda+\ii}.
\end{equation}

Similarly, we can introduce a square root function $R_{1}(\lambda)\equiv R_{1}(\lambda; \chi, \tau)$ with the definition
\begin{equation}\label{eq:Re}
R_{1}(\lambda):=\sqrt{(\lambda-a_{1})(\lambda-a_{1}^*)},\quad a_1\equiv a_1(\chi, \tau),
\end{equation}
then $\frac{g'_{1}(\lambda)}{R_{1}(\lambda)}$ satisfies the relation
\begin{equation}
\left(\frac{g'_{1}(\lambda)}{R_{1}(\lambda)}\right)_+-\left(\frac{g'_{1}(\lambda)}{R_{1}(\lambda)}\right)_-=\frac{-2\chi-4\lambda\tau-\frac{\ii}{\lambda-\lambda_1}
+\frac{\ii}{\lambda-\lambda_1^*}-\frac{1}{2n}\frac{\ii r}{\lambda-\ii}+\frac{1}{2n}\frac{\ii r}{\lambda+\ii}}{R_{1}(\lambda)_+},
\end{equation}
which can also be solved by the Plemelj formula,
\begin{equation}
g'_{1}(\lambda)=\frac{R_{1}(\lambda)}{2\pi\ii}\int_{\Sigma_{g_1}}\frac{-2\chi-4s\tau-\frac{\ii}{s-\lambda_1}+\frac{\ii}{s-\lambda_1^*}
-\frac{1}{2n}\frac{\ii r}{s-\ii}+\frac{1}{2n}\frac{\ii r}{s+\ii}}{R_{1}(s)(s-\lambda)}ds.
\end{equation}
With the generalized residue theorem, $g'_{1}(\lambda)$ can be given as an explicit formula:
\begin{equation}
\begin{aligned}
g'_{1}(\lambda)&=R_{1}(\lambda)\left(\mathop{\rm Res}\limits_{s=\lambda}+\mathop{\rm Res}\limits_{s=\lambda_1}+\mathop{\rm Res}\limits_{s=\lambda_1^*}{+}\mathop{\rm Res}
\limits_{s=\infty}\right)\left(\frac{-\chi-2 s\tau-\frac{\ii}{2(s-\lambda_1)}+\frac{\ii}{2(s-\lambda_1^*)}-\frac{1}{4n}\frac{\ii r}{(s-\ii)}+\frac{1}{4n}\frac{\ii r}{(s+\ii)}}{R_{1}(s)\left(s-\lambda\right)}\right)\\
&+R_{1}(\lambda)\left(\mathop{\rm Res}\limits_{s=\ii}+\mathop{\rm Res}\limits_{s=-\ii}\right)\left(\frac{-\chi-2 s\tau-\frac{\ii}{2(s-\lambda_1)}+\frac{\ii}{2(s-\lambda_1^*)}-\frac{1}{4n}\frac{\ii r}{(s-\ii)}+\frac{1}{4n}\frac{\ii r}{(s+\ii)}}{R_{1}(s)\left(s-\lambda\right)}\right)\\
&=R_{1}(\lambda)\left[{-}\frac{\ii}{2R_{1}(\lambda_1)(\lambda_1-\lambda)}{+}\frac{\ii}{2R_{1}(\lambda_1^*)(\lambda_1^*-\lambda)}{-}\frac{\ii r}{4nR_{1}(\ii)(\ii-\lambda)}{+}\frac{\ii r}{4nR_{1}(-\ii)(-\ii-\lambda)}{+}2 \tau \right]\\
&-\chi-2\lambda\tau-\frac{\ii}{2}\frac{1}{\lambda-\lambda_1}+\frac{\ii}{2}\frac{1}{\lambda-\lambda_1^*}-\frac{\ii r}{4n}\frac{1}{\lambda-\ii}-\frac{\ii r}{4n}\frac{1}{\lambda+\ii}.
\end{aligned}
\end{equation}
Moreover, the phase term will be modified as $h_1(\lambda)\equiv h_1(\lambda; \chi, \tau):=g_1(\lambda)+\vartheta(\lambda; \chi, \tau)$, thus we have
\begin{equation}
h_1'(\lambda)=R_{1}(\lambda)\left[{-}\frac{\ii}{2R_{1}(\lambda_1)(\lambda_1-\lambda)}{+}\frac{\ii}{2R_{1}(\lambda_1^*)(\lambda_1^*-\lambda)}{-}\frac{\ii r}{4nR_{1}(\ii)(\ii-\lambda)}{+}\frac{\ii r}{4nR_{1}(-\ii)(-\ii-\lambda)}{+}2 \tau \right].
\end{equation}
Compared to the formula of $g_2'(\lambda)$ in Eq.\eqref{eq:dG}, $g_1'(\lambda)$ seems similar by replacing $R_2(\lambda)$ with $R_{1}(\lambda)$, which only adds an additional factor $2R_{1}(\lambda)\tau$. But the asymptotics for these two regions are different due to the difference between $R_{1}(\lambda)$ and $R_{2}(\lambda)$. In the genus-two region, the normalization condition of $g(\lambda)$ at $\lambda=\infty$ produces four conditions to the parameters $s_i(i=1,\cdots,6)$. While in the genus-zero region, we have only two parameters to de determined, the real and imaginary parts of $a_1$.
The normalization condition of $g'_{1}(\lambda)$ in the large-$\lambda$ expansion can derive two relations about these two parameters, and we do not need other conditions anymore,
\begin{equation}
\begin{aligned}
&\mathcal{O}(1): \chi+2\tau \Re(a_{1})-\frac{\ii}{2R_{1}(\lambda_1)}+\frac{\ii}{2R_{1}(\lambda_1^*)}-\frac{\ii r}{4nR_{1}(\ii)}+\frac{\ii r}{4nR_{1}(\ii)}=0,
\\&\mathcal{O}(\lambda^{-1}): \frac{\ii\lambda_1^*}{2R_1(\lambda_1^*)}-\frac{\ii\lambda_1}{2R_1(\lambda_1)}+\frac{r}{4nR_1(\ii)}+\frac{r}{4nR_1(-\ii)}+2\Re(a_1)^2\tau-\Im(a_1)^2\tau+\Re(a_1)\chi=0.
\end{aligned}
\end{equation}
Then we can calculate the unknown parameters numerically. In this region, $h'_{1}(\lambda)$ has real root. By choosing proper $(\chi, \tau)$ in this region,
we give the sign chart of $\Im(h_1(\lambda))$ in Fig. \ref{fig:contour-exp-1}.

\begin{figure}[!ht]
\centering
\includegraphics[width=0.45\textwidth]{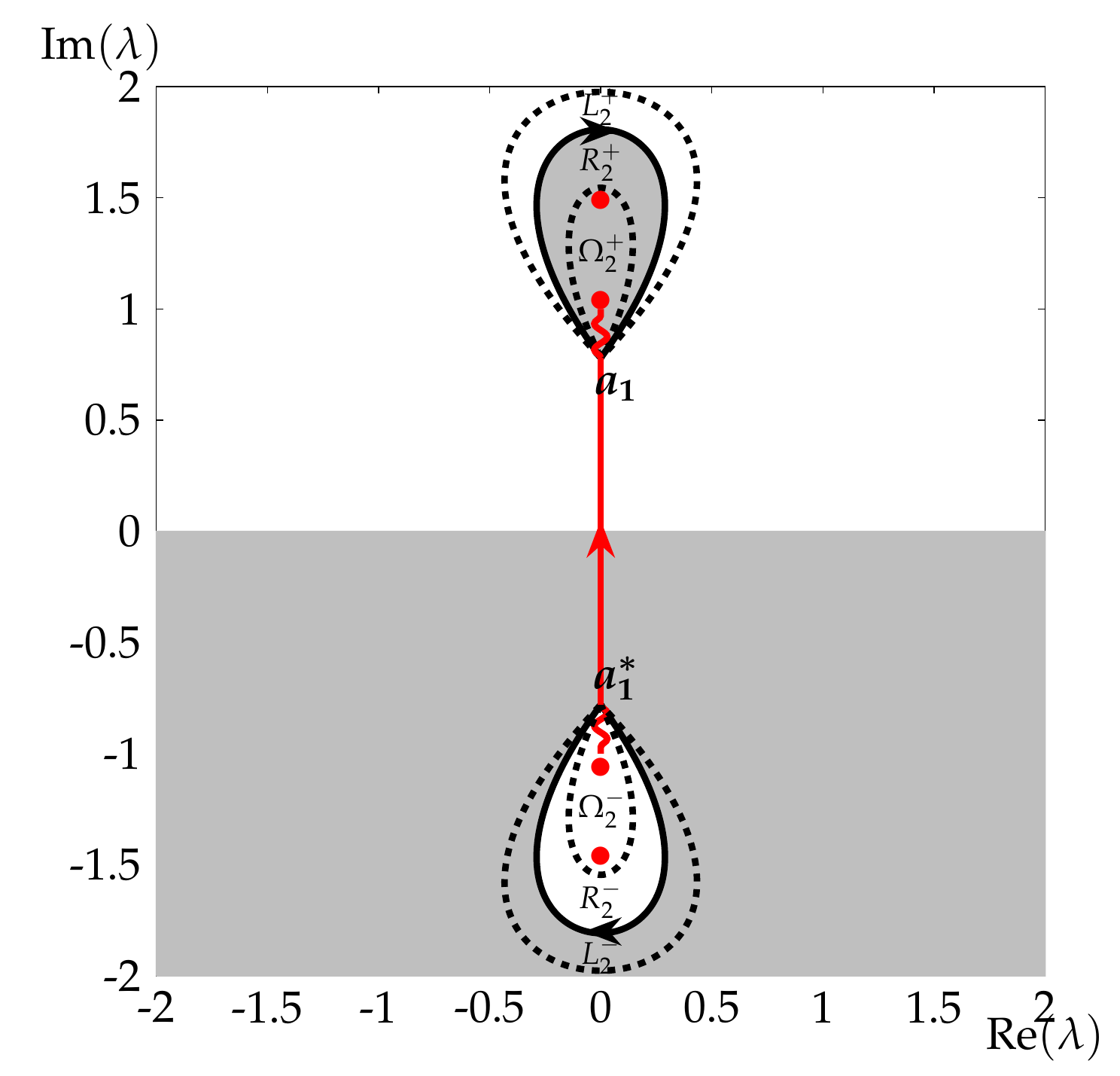}
\centering
\includegraphics[width=0.45\textwidth]{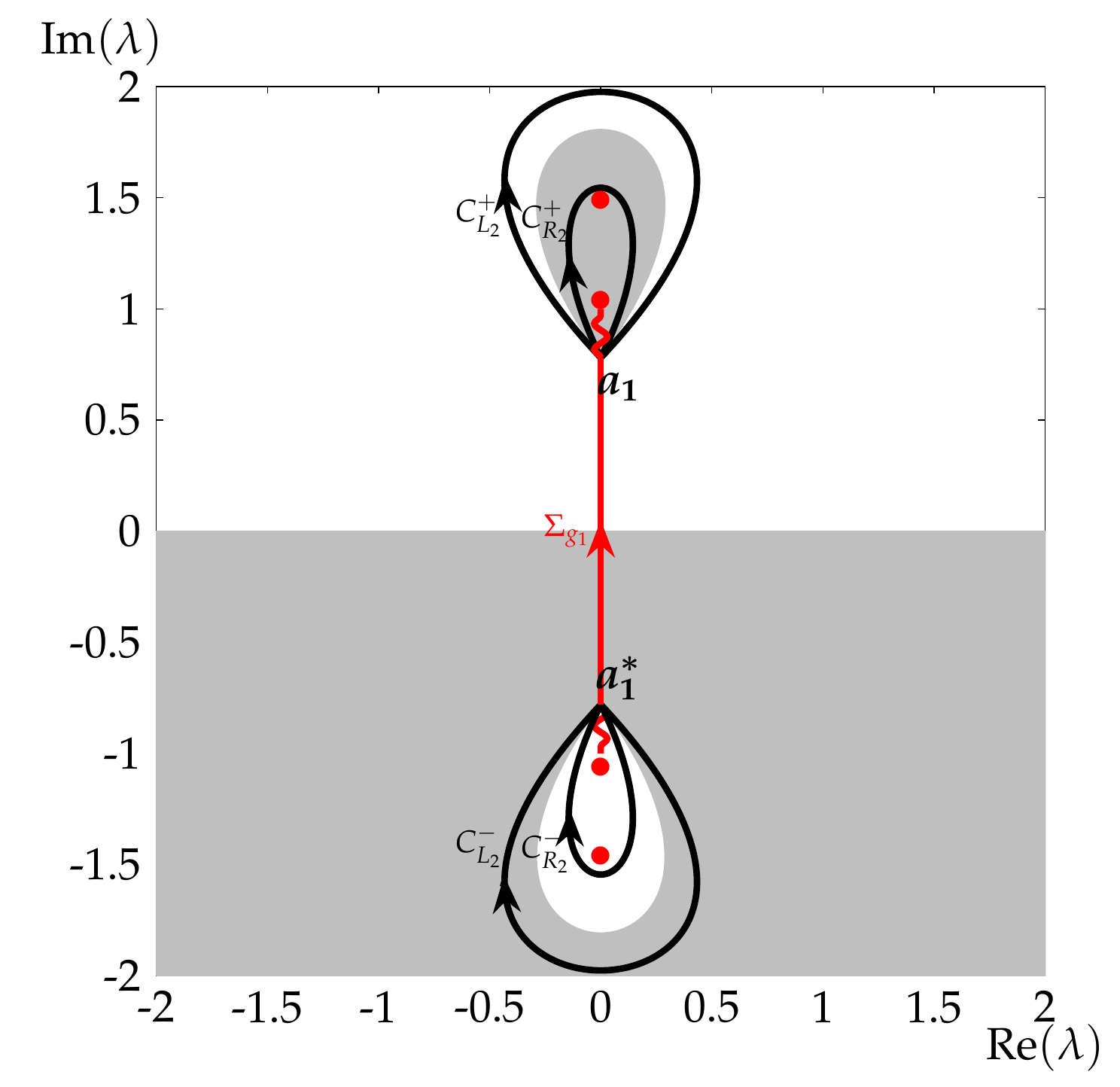}
\caption{The sign chart of ${\Im}(h_{1}(\lambda; \frac{7}{10}, \frac{1}{200}))$ in the genus-zero-infinity region, where ${\Im}(h_{1}(\lambda; \frac{7}{10}, \frac{1}{200}))>0$(unshaded) and ${\Im}(h_{1}(\lambda; \frac{7}{10}, \frac{1}{200}))<0$(shaded).
 The left one gives the original contour for $\mathbf{S}_{2}(\lambda; \chi, \tau)$, and the right panel is the corresponding jump contour for $\mathbf{T}_{2}(\lambda; \chi, \tau)$.}
\label{fig:contour-exp-1}
\end{figure}

For $(\chi, \tau)$ in the genus-zero-up region, the definition of $g_1(\lambda)$-function is similar with the genus-zero-infinity region. But in the genus-zero-up region, $h_1'(\lambda)$ has no real roots, thus the jump contour of $\Im(h_{1}(\lambda))$ has a slight difference with the genus-zero-infinity region. We omit the details for the genus-zero-up region and only give the sign chart of $\Im(h_1(\lambda))$ by choosing one fixed $(\chi, \tau)$ in Fig.\ref{fig:contour-exp-2}.
\begin{figure}[!ht]
	\centering
	\includegraphics[width=0.45\textwidth]{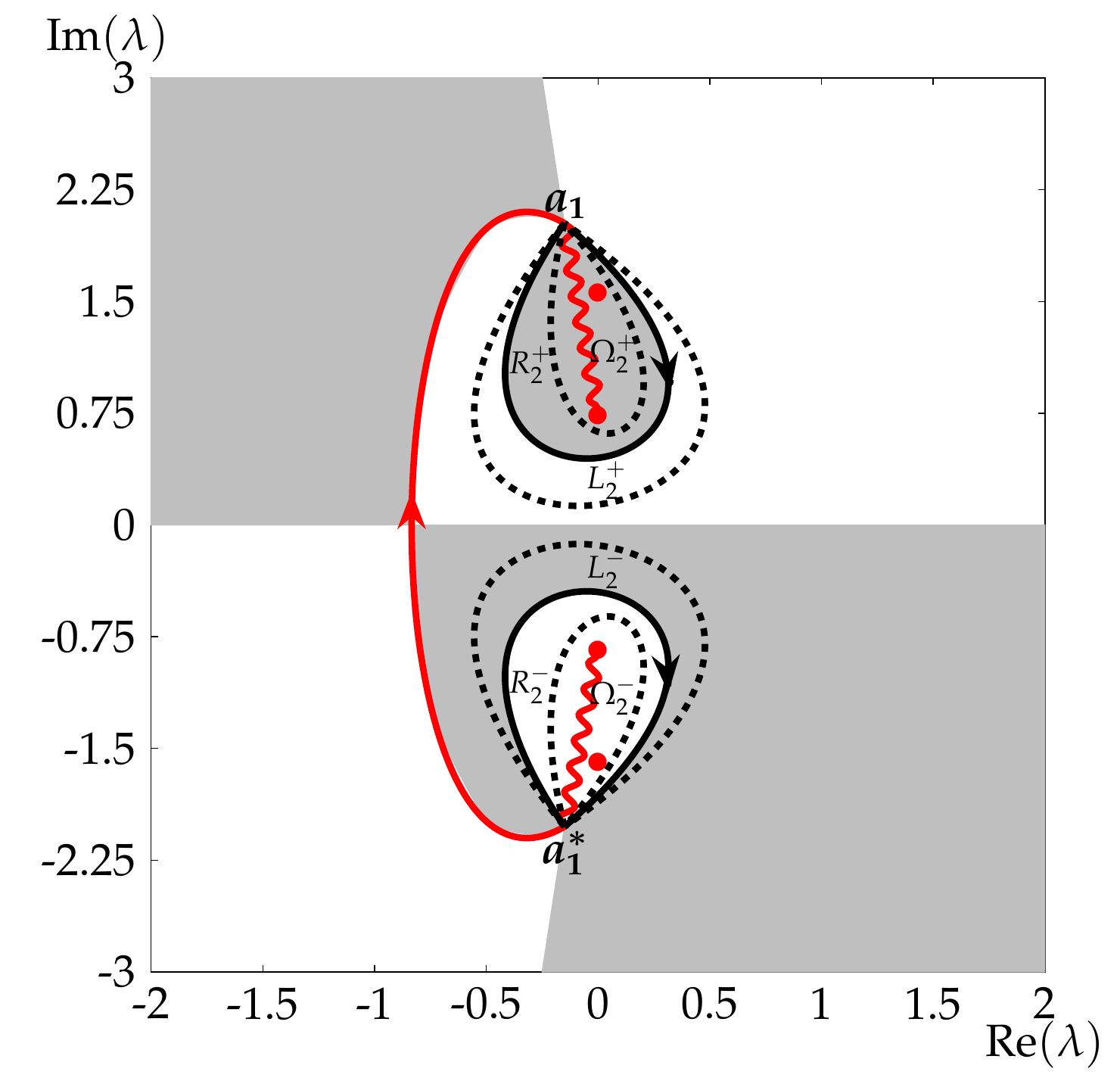}
	\centering
	\includegraphics[width=0.45\textwidth]{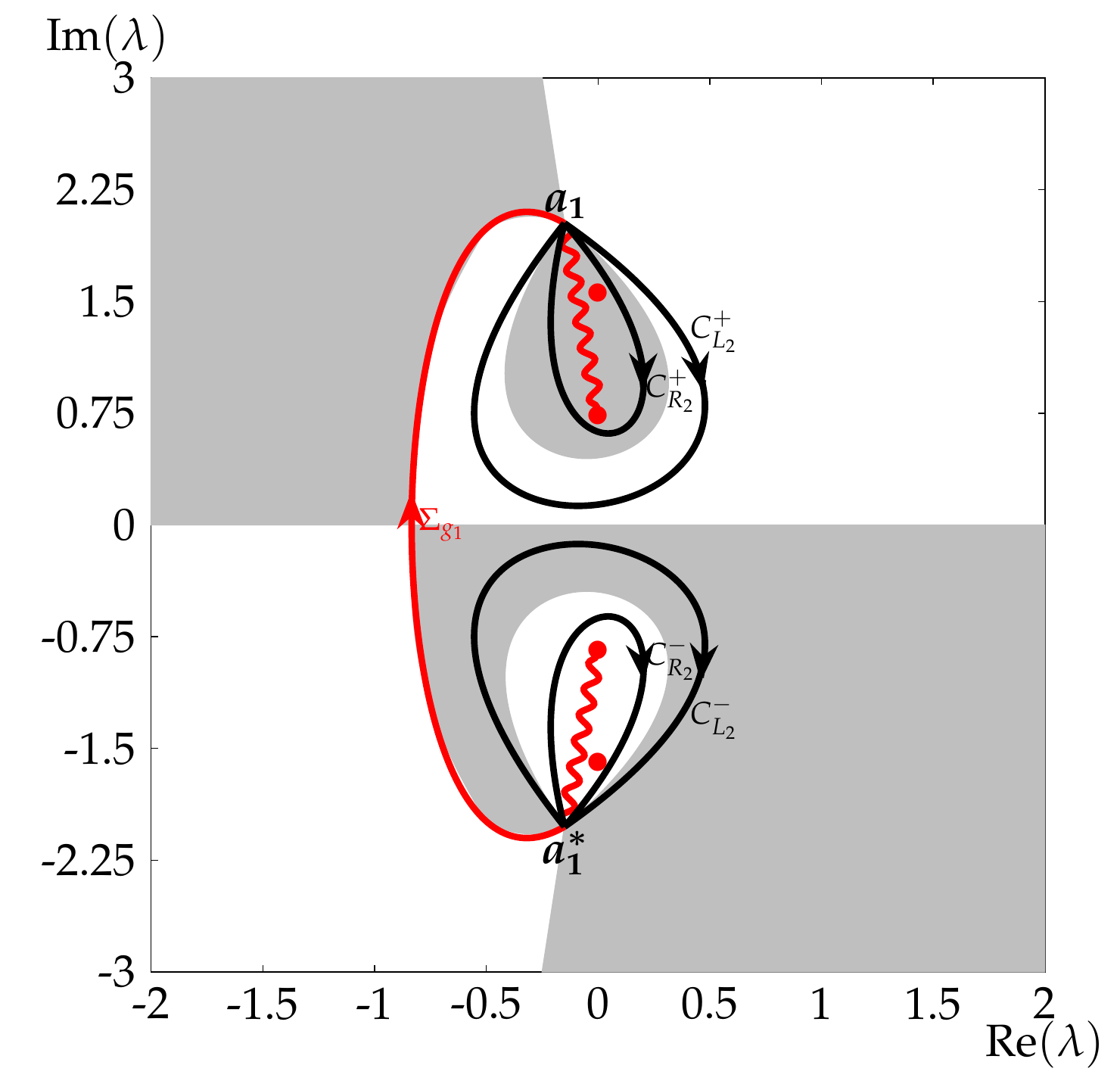}
	\caption{The sign chart of ${\Im}(h_{1}(\lambda; \frac{1}{5}, \frac{28}{100}))$ in the genus-zero-up region, where ${\Im}(h_{1}(\lambda; \frac{1}{5}, \frac{28}{100}))>0$(unshaded)
and ${\Im}(h_{1}(\lambda; \frac{1}{5}, \frac{28}{100}))<0$(shaded). The cut $\Sigma_{g_1}$ is chosen as arbitrary segments connecting the branch points of $R_1(\lambda)$.
The left one gives the jump contour for $\mathbf{S}_{2}(\lambda; \chi, \tau)$, and the right panel is the corresponding jump contour for $\mathbf{T}_{2}(\lambda; \chi, \tau)$.}
	\label{fig:contour-exp-2}
\end{figure}

With the sign of $\Im(h_{1}(\lambda))$ in Fig.\ref{fig:contour-exp-1} and Fig.\ref{fig:contour-exp-2}, we can analyze the asymptotics for the genus-zero-infinity region and genus-zero-up region together via the nonlinear steepest-descent method. Similarly, set
\begin{equation}
\mathbf{S}_{2}(\lambda; \chi, \tau):=\left\{\begin{aligned}&\mathbf{M}^{[n]}(\lambda; \chi, \tau)\ee^{-\ii n\vartheta(\lambda; \chi, \tau)\sigma_3}\mathbf{Q}^{-1}_d
\ee^{\ii n\vartheta(\lambda; \chi, \tau)\sigma_3},\quad &\lambda\in D_{0}\cap\left(D_{2}^{+}\cup D_{2}^{-}\right)^{c},\\
&\mathbf{M}^{[n]}(\lambda; \chi, \tau),\quad &\text{otherwise},
\end{aligned}\right.
\end{equation}
where $D_{2}^{\pm}=\Omega_{2}^{\pm}\cup R_{2}^{\pm}$. Then the jump curves of $\mathbf{S}_{2}(\lambda; \chi, \tau)$ transfer into $\partial D_{2}^{\pm}$ and the contour $\Sigma_{g_1}$. That is
\begin{equation}\label{jump-S2}
\begin{split}
\mathbf{S}_{2,+}(\lambda; \chi, \tau)&=\mathbf{S}_{2,-}(\lambda; \chi, \tau)\ee^{-\ii n \vartheta_{-}(\lambda; \chi, \tau)\sigma_3}\begin{bmatrix}0&\ii r\\\ii r &0
	\end{bmatrix}\ee^{\ii n\vartheta_{+}(\lambda; \chi, \tau)\sigma_3},\quad \lambda\in\Sigma_{g_1},\\
\mathbf{S}_{2,+}(\lambda; \chi, \tau)&=\mathbf{S}_{2,-}(\lambda; \chi, \tau)\ee^{-\ii n\vartheta(\lambda; \chi, \tau)}\mathbf{Q}_{d}^{-1}\ee^{\ii n\vartheta(\lambda; \chi, \tau)},\qquad \qquad\lambda\in\partial D_{2}^{\pm}.
\end{split}
\end{equation}
In the regions $\Omega_{2}^{\pm}, R_{2}^{\pm}$ and $L_{2}^{\pm}$, we define a similar matrix $\mathbf{T}_{2}(\lambda; \chi, \tau)$ as $\mathbf{T}_{1}(\lambda; \chi, \tau)$
in the regions $\Omega^{\pm}, R^{\pm}$ and $L^{\pm}$ in Eq.\eqref{eq:S-genus-two} by replacing $g_{2}(\lambda)$ with $g_1(\lambda)$ respectively. As a result, when $n$ is large, the primary jump condition of $\mathbf{T}_{2}(\lambda; \chi, \tau)$ changes into
\begin{equation}
\begin{aligned}
\mathbf{T}_{2, +}(\lambda; \chi, \tau)&=\mathbf{T}_{2,-}(\lambda; \chi, \tau)\begin{bmatrix}0&\ii r\ee^{-\ii n\kappa_{1}}\\
\ii r\ee^{\ii n\kappa_{1}}&0
\end{bmatrix},\quad&\lambda\in\Sigma_{g_{1}}.\\
\end{aligned}
\end{equation}
And other jump conditions will decay to the identity matrix exponentially.
Next, we will give the parametrix construction for $\mathbf{T}_{2}(\lambda; \chi, \tau)$ in the following subsection.
\subsection{Parametrix construction}\label{sec:para-exp}
Similar to the analysis in the genus-two region, we first give an outer parametrix $\dot{\mathbf{T}}_{2}^{\rm out}(\lambda; \chi, \tau)$ satisfying the
same jump conditions for $\lambda\in\Sigma_{g_{1}}$. By the Plemelj formula, the outer parametrix can be given as
\begin{equation}
\dot{\mathbf{T}}_{2}^{\rm out}(\lambda; \chi, \tau)=\ee^{\frac{-\ii n \kappa_{1}}{2}\sigma_3}
\mathbf{Q}_d\left(\frac{\lambda-a_{1}}{\lambda-a_{1}^*}\right)^{\frac{1}{4}r\sigma_3}\mathbf{Q}_d^{-1}\ee^{\frac{\ii n \kappa_{1}}{2}\sigma_3},\quad \lambda\in\Sigma_{g_{1}}.
\end{equation}
It can be seen that the outer parametrix has two singularities at $\lambda=a_{1}, \lambda=a_{1}^*$, thus we should consider the local analysis at these two points. Set the inner parametrices as
\begin{equation}
\begin{aligned}
\dot{\mathbf{T}}_{2}^{a_{1}}(\lambda; \chi, \tau),\quad \lambda\in D_{a_{1}}(\delta),\\
\dot{\mathbf{T}}_{2}^{a_{1}^*}(\lambda; \chi, \tau),\quad \lambda\in D_{a_{1}^*}(\delta),
\end{aligned}
\end{equation}
based on the result in \cite{Bilman-arxiv-2021}, in the neighbourhood of $\lambda=a_{1}$ and $\lambda=a_{1}^*$, the inner parametrices $\dot{\mathbf{T}}_{2}^{a_{1}}(\lambda; \chi, \tau)$
and $\dot{\mathbf{T}}_{2}^{a_{1}^*}(\lambda; \chi, \tau)$ are related to the Airy function, and we have an error estimation of
$\dot{\mathbf{T}}_{2}^{a_{1}}(\lambda; \chi, \tau), $ $\dot{\mathbf{T}}_{2}^{a_{1}^*}(\lambda; \chi, \tau)$ and $\dot{\mathbf{T}}_{2}^{\rm out}(\lambda; \chi, \tau)$
\begin{equation}
\begin{aligned}
\|\dot{\mathbf{T}}_{2}^{a_{1}}(\lambda; \chi, \tau)\left(\dot{\mathbf{T}}_{2}^{\rm out}(\lambda; \chi, \tau)\right)^{-1}\|=\mathcal{O}(n^{-1}),\\
\|\dot{\mathbf{T}}_{2}^{a_{1}^*}(\lambda; \chi, \tau)\left(\dot{\mathbf{T}}_{2}^{\rm out}(\lambda; \chi, \tau)\right)^{-1}\|=\mathcal{O}(n^{-1}).
\end{aligned}
\end{equation}
Then the global parametrix can be defined by
\begin{equation}
\dot{\mathbf{T}}_{2}(\lambda; \chi, \tau):=\left\{\begin{aligned}&\dot{\mathbf{T}}_{2}^{a_{1}}(\lambda; \chi, \tau),\quad &\lambda\in D_{a_{1}}(\delta),\\
&\dot{\mathbf{T}}_{2}^{a_{1}^*}(\lambda; \chi, \tau),\quad &\lambda\in D_{a_{1}^*}(\delta),\\
&\dot{\mathbf{T}}_{2}^{\rm out}(\lambda; \chi, \tau),\quad &\lambda\in \mathbb{C}\setminus\left(\overline{D_{a_{1}}(\delta)}\cup\overline{D_{a_{1}^*}(\delta)}\cup\Sigma_{g_1}\right).
\end{aligned}\right.
\end{equation}
Next, we can give the error analysis between $\mathbf{T}_2(\lambda; \chi, \tau)$ and its parametrix $\dot{\mathbf{T}}_{2}(\lambda; \chi, \tau)$, define
\begin{equation}
\mathcal{E}_{2}(\lambda; \chi, \tau):=\mathbf{T}_{2}(\lambda; \chi, \tau)\left(\dot{\mathbf{T}}_{2}(\lambda; \chi, \tau)\right)^{-1},
\end{equation}
then the solution $q^{[n]}(n\chi, n\tau)$ can be given in Eq.\eqref{eq:q-exp}
\begin{equation}\label{eq:q-exp}
	\begin{aligned}
		q^{[n]}(n\chi, n\tau)
		&=\Im (a_{1})\ee^{-\ii n\kappa_1}+\mathcal{O}(n^{-1}).
	\end{aligned}
\end{equation}
\begin{remark}
In the genus-zero-infinity and genus-zero-up regions, if the branch point $a_1$ is close to $\ii$, the asymptotic expression Eq.\eqref{eq:q-exp} will tend to
the background solution $q=\ee^{\ii t}$. Indeed, this result can also be obtained from the jump condition \eqref{eq:jump-genus-zero-1} in the RHP\ref{RHP:EXP}. With the Plemelj formula, we have
\begin{equation}
g_1(\lambda)=\frac{R_{1}(\lambda)}{2\pi\ii}\int_{\Sigma_{g_1}}\frac{\kappa_1-\vartheta_{+}(\xi; \chi, \tau)-\vartheta_{-}(\xi; \chi, \tau)}{R_{1}(\xi)(\xi-\lambda)}d\xi.
\end{equation}
From the normalization of $g_{1}(\lambda)$ as $\lambda\to\infty$, we get an identity,
\begin{equation}
\int_{\Sigma_{g_1}}\frac{\kappa_1}{R_{1}(\xi)}d\xi=\int_{\Sigma_{g_1}}\frac{\vartheta_+(\xi; \chi, \tau)+\vartheta_-(\xi; \chi, \tau)}{R_{1}(\xi)}d\xi.
\end{equation}
With the aid of the generalized residue theorem, the integration constant $\kappa_1$ can be represented by
\begin{equation}\label{eq:kappa1}
\kappa_1=(2\Re(a_1)^2-\Im(a_1)^2)\tau+2\Re(a_1)\chi+\ii\int_{-2\ii}^{a_1^*}\frac{1}{R_{1}(\xi)}d\xi+\ii\int_{a_1}^{2\ii}\frac{1}{R_1(\xi)}d\xi.
\end{equation}
If $a_1\to\ii,$ the last two integrals in Eq.\eqref{eq:kappa1} vanish, then $\kappa_1\to-\tau$. Plugging it into the asymptotic expression \eqref{eq:q-exp}, we get $q^{[n]}(n\chi, n\tau)\to \ee^{\ii n\tau}= \ee^{\ii t}$.
\end{remark}
By choosing proper parameters in these two regions, we compare the exact and the asymptotic solutions in Fig.\ref{fig:al} and Fig.\ref{fig:exp}.
\begin{figure}[ht]
	\centering
	\includegraphics[width=1\textwidth,height=0.25\textwidth]{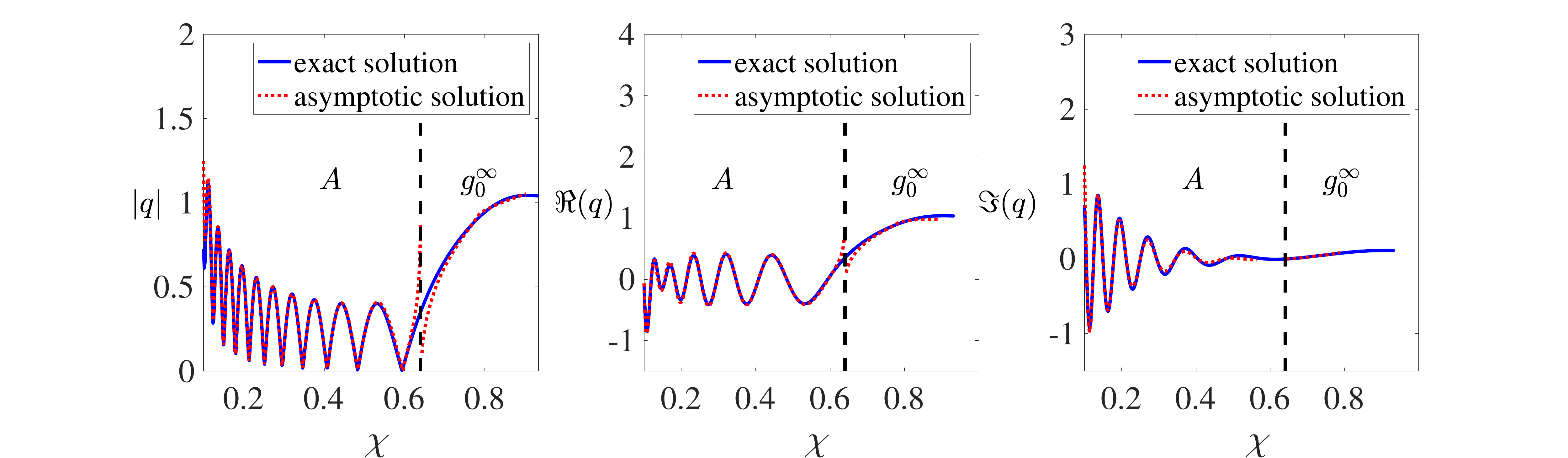}
	\caption{The comparison between the exact solution ($20$-th order KMBs) and its asymptotic solution in the algebraic-decay region(as show by the alphabet $A$) and the genus-zero-infinity region(as shown by $g_0^\infty$) by choosing $\tau=\frac{1}{200}$
 (as shown by the green dashed line in the middle panel in Fig. \ref{fig:KM}), $c_1=-c_2=1$. The left one is the modulus of $q^{[n]}(n\chi, n\tau)$, the middle and right panels show the real and imaginary parts of $q^{[n]}(n\chi, n\tau)$ respectively.}
	\label{fig:al}
\end{figure}
\begin{figure}[!ht]
	\centering
	\includegraphics[width=1\textwidth]{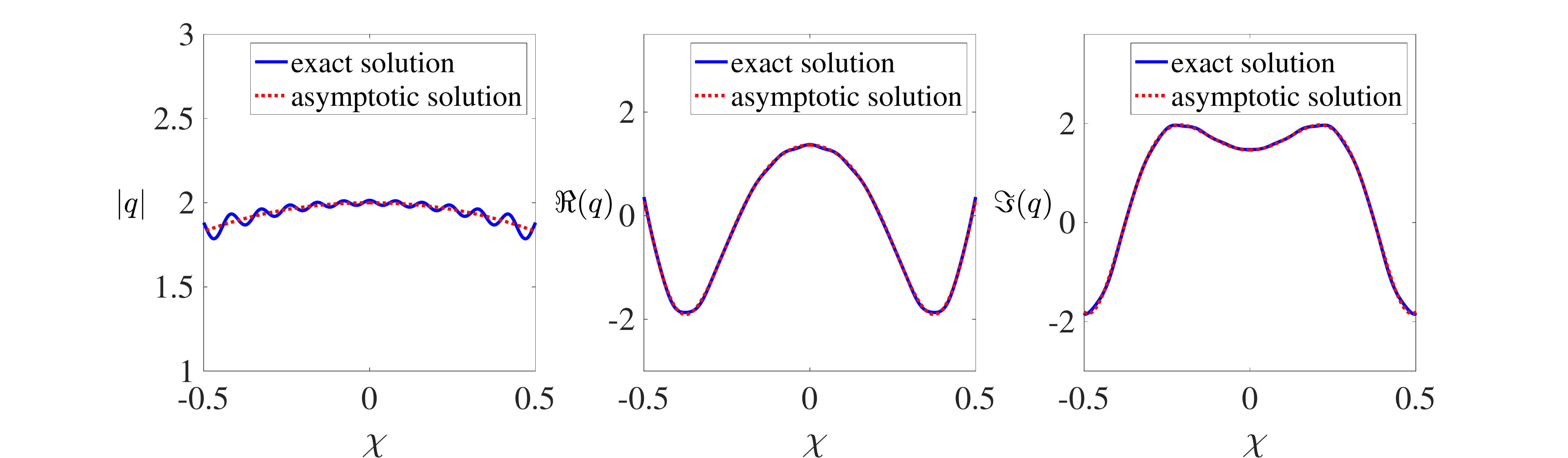}
	\caption{The comparison between the exact solution ($20$-th order KMBs) and the asymptotics in the genus-zero-up region by choosing $\tau=\frac{28}{100}$ (as shown by the green dashed line
in the middle panel in Fig. \ref{fig:KM}), $c_1=-c_2=1$. The left one is the modulus of $q^{[n]}(n\chi, n\tau)$, the middle and right panels show the real and imaginary parts of $q^{[n]}(n\chi, n\tau)$ respectively.}
	\label{fig:exp}
\end{figure}
\section{Genus-zero-down region}
In this section, we continue to study the asymptotics in the genus-zero-down region. For the asymptotics of KMBs, we present three types of genus-zero regions, the genus-zero-infinity region, the genus-zero-up region and the genus-zero-down region. In the last section, we have studied two of them.
Now we give the analysis for the genus-zero-down region.
Similar with the asymptotics in the other two genus-zero regions, we also need an auxiliary $g$-function(this $g$-function is set as $g_{0}(\lambda):=g_{0}(\lambda; \chi, \tau)$),
which has a similar formula with $g_{1}(\lambda)$ in RHP \ref{RHP:EXP}.
But the jump contour of $\Im(h_{0}(\lambda)\equiv h_0(\lambda; \chi, \tau):=g_{0}(\lambda)+\vartheta(\lambda; \chi, \tau))$ in this region is very different. In the above discussion, the branch cut $[-\ii, \ii]$ contains two parts, one coincides with the cut of $g_{1}(\lambda)$-function,
and the other one connects the branch point $a_{1}$ and $\ii$ as well as $a_{1}^{*}$ and $-\ii$. But in this region, the branch cut $[-\ii, \ii]$ is in a
closed region and the cut for the corresponding $g_{0}(\lambda)$-function is independent of the cut $[-\ii, \ii]$. Thus the RHP for $g_{0}(\lambda)$-function is the same as $g_{1}(\lambda)$-function
by replacing $\vartheta_{\pm}(\lambda; \chi, \tau)$ with $\vartheta(\lambda; \chi, \tau)$, the integration constant $\kappa_1$ is replaced with $\kappa_{0}$. Moreover, we introduce the root function $R_{0}(\lambda)\equiv R_{0}(\lambda; \chi, \tau):=\sqrt{(\lambda-a_{0})(\lambda-a_{0}^*)}$ to replace $R_{1}(\lambda)$.
By choosing one proper $\chi$ and $\tau$ in this region, we give the sign chart of $\Im(h_{0}(\lambda))$ in Fig.\ref{fig:contour-genus-0}.
\begin{figure}[!ht]
	\centering
	\includegraphics[width=0.45\textwidth]{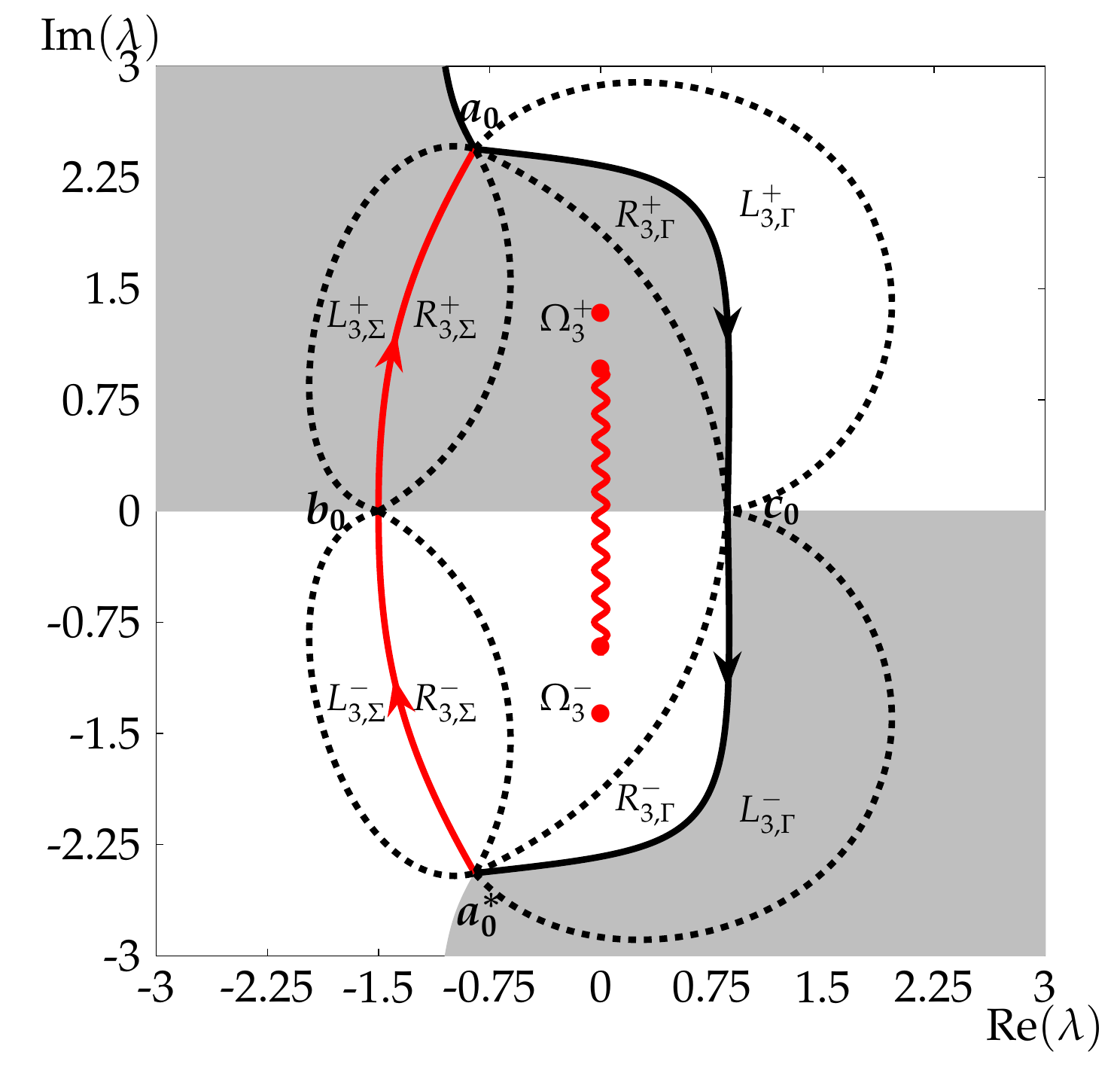}
	\centering
	\includegraphics[width=0.45\textwidth]{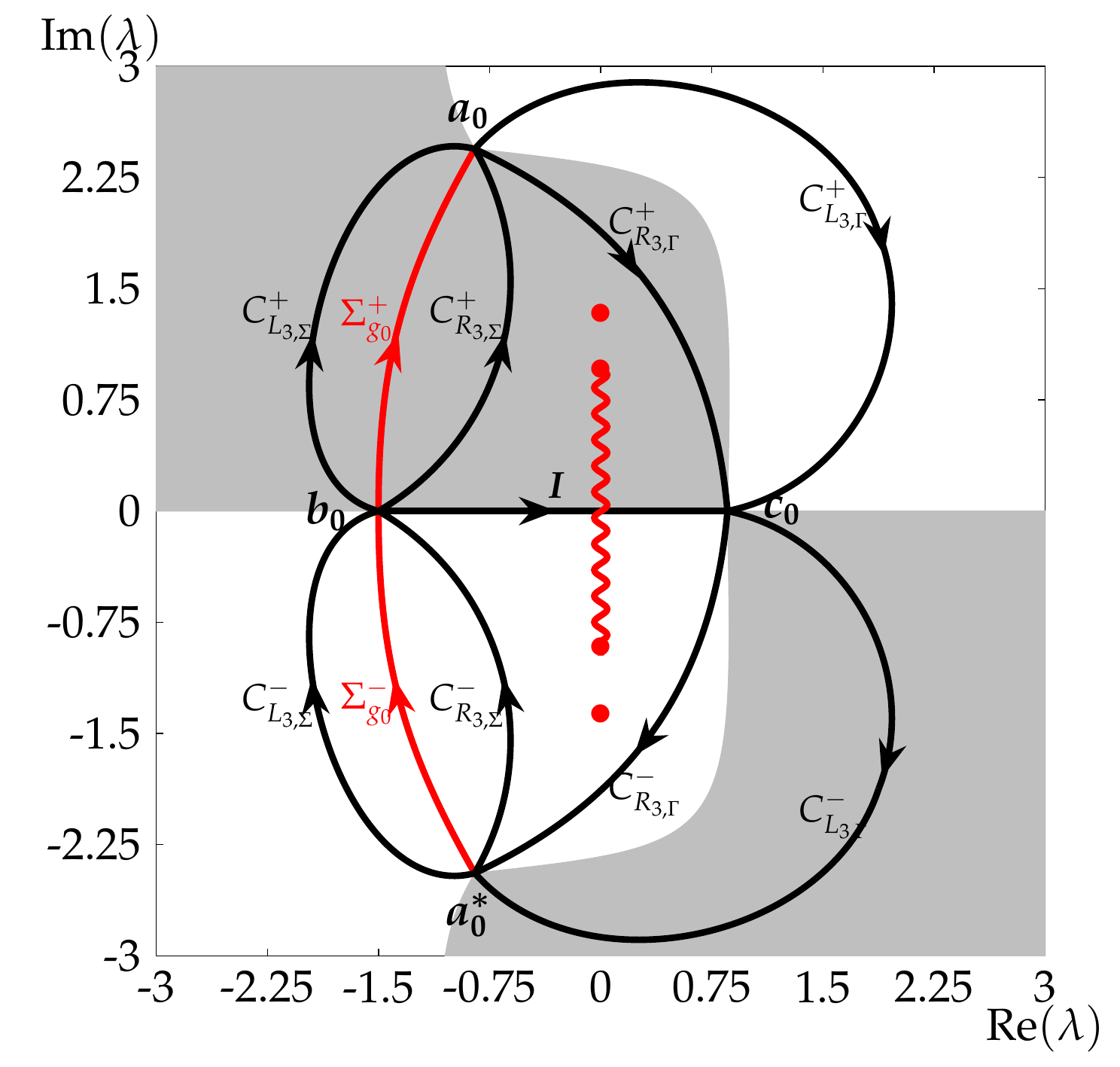}
	\caption{The sign chart of ${\Im}(h_{0}(\lambda; \frac{1}{5}, \frac{1}{10}))$ in the genus-zero-down region, where ${\Im}(h_{0}(\lambda; \frac{1}{5}, \frac{1}{10}))>0$(unshaded) and ${\Im}(h_{0}(\lambda; \frac{1}{5}, \frac{1}{10}))<0$(shaded).
  The left one gives the original jump contour for $\mathbf{S}_{3}(\lambda; \chi, \tau)$, and the right one is the corresponding jump contour after deformation.}
	\label{fig:contour-genus-0}
\end{figure}

Then we can define a similar sectional holomorphic function $\mathbf{S}_{3}(\lambda; \chi, \tau)$ as Eq.\eqref{S1-genus-two},
\begin{equation}\label{S3-genus-zero}
	\mathbf{S}_{3}(\lambda; \chi, \tau):=\left\{\begin{aligned}&\mathbf{M}^{[n]}(\lambda; \chi, \tau)\ee^{-\ii n
			\vartheta(\lambda; \chi, \tau)\sigma_3}\mathbf{Q}^{-1}_d
		\ee^{\ii n\vartheta(\lambda; \chi, \tau)\sigma_3},\quad &\lambda\in D_{0}\cap\left(D_{3}^{+}\cup D_{3}^{-}\right)^{c},\\
		&\mathbf{M}^{[n]}(\lambda; \chi, \tau),\quad &\text{otherwise},
	\end{aligned}\right.
\end{equation}
where $D_{3}^{\pm}=R_{3,\Sigma}^{\pm}\cup\Omega_3^{\pm}\cup R_{3,\Gamma}^{\pm}.$ Since $\vartheta(\lambda)$ has no cut at $\Sigma_{g_0}^{\pm}$, the jump conditions between
$\mathbf{S}_{3}(\lambda; \chi, \tau)$ and $\mathbf{S}_{2}(\lambda; \chi, \tau)$ have a little difference, in this case, we have
\begin{equation}
\mathbf{S}_{3,+}(\lambda; \chi, \tau)=\mathbf{S}_{3,-}(\lambda; \chi, \tau)\ee^{-\ii n\vartheta(\lambda; \chi, \tau)}\mathbf{Q}_{d}^{-1}\ee^{\ii n \vartheta(\lambda; \chi, \tau)},\quad\lambda\in\partial D_{3}^{\pm}.
\end{equation}
Next, in the regions $L_{3,\Sigma}^{\pm}, R_{3,\Sigma}^{\pm}, \Omega_{3}^{\pm}, R_{3,\Gamma}^{\pm}, L_{3,\Gamma}^{\pm}$,
define a similar matrix $\mathbf{T}_{3}(\lambda; \chi, \tau)$ as $\mathbf{T}_{1}(\lambda; \chi, \tau)$ in Eq.\eqref{eq:S-genus-two} in the corresponding regions $L_{1,\Sigma}^{\pm}, R_{1,\Sigma}^{\pm}, \Omega_{1}^{\pm}, R_{1,\Gamma}^{\pm}, L_{1,\Gamma}^{\pm}$
by replacing $g_{2}(\lambda)$ with $g_{0}(\lambda)$.
Then the primary jump conditions for $\mathbf{T}_{3}(\lambda; \chi, \tau)$ are,
\begin{equation}\label{eq:jump-genus-zero}
\begin{aligned}
		\mathbf{T}_{3, +}(\lambda; \chi, \tau)&=\mathbf{T}_{3, -}(\lambda; \chi, \tau)\begin{bmatrix}0&\ee^{-\ii n\kappa_0}\\
		-\ee^{\ii n\kappa_0}&0
	\end{bmatrix},&\lambda\in \Sigma_{g_0}^{\pm},\\
\mathbf{T}_{3,+}(\lambda; \chi, \tau)&=\mathbf{T}_{3,-}(\lambda; \chi, \tau)2^{\sigma_3},&\lambda\in I.
\end{aligned}
\end{equation}
Next we will give the parametrix construction for $\mathbf{T}_{3}(\lambda; \chi, \tau)$.
\subsection{Parametrix construction for $\mathbf{T}_{3}(\lambda; \chi, \tau)$}
In our previous work \cite{ling2022large}, we analyzed the large-order asymptotics of breathers for the NLS equation, which is constructed from two solitons on the vanishing background with the same velocity.
The phase terms between these two breathers are different, but for this genus-zero region, the jump contour of the $\Im(h_0(\lambda))$ is very similar because the singularities appearing in the phase terms are all in a closed contour.
As a result, after the deformation of contours, the jump conditions given in Eq. \eqref{eq:jump-genus-zero} are similar to the last jump conditions in the reference \cite{ling2022large} Eq.(77). Therefore, the parametrix construction will also be similar.
In the work \cite{ling2022large}, we presented a detailed analysis for the parametrix construction. Thus we only give a brief statement in this work,.
From the constant jump matrices when $\lambda\in \Sigma_{g_0}^{\pm}$ and $\lambda\in I$, the outer parametrix $	\dot{\mathbf{T}}_{3}^{\rm out}(\lambda; \chi, \tau)$ can be given as
\begin{equation}\label{eq:T_2out}
	\dot{\mathbf{T}}_{3}^{\rm out}(\lambda):=\mathbf{K}_{3}(\lambda)\left(\frac{\lambda-b_0}{\lambda-c_0}\right)^{\ii p\sigma_3},
\quad p=\frac{\log\left(2\right)}{2\pi},\quad \lambda\in\mathbb{C}\setminus\left(\Sigma_{g_0}^{\pm}\cup I\right),
\end{equation}
where $\mathbf{K}_{3}(\lambda)\equiv\mathbf{K}_{3}(\lambda; \chi, \tau)$ equals to
\begin{equation}
		\mathbf{K}_{3}(\lambda)=\ee^{-\frac{\ii \pi}{4}\sigma_3}\ee^{\frac{2\ii k_{3}(\infty)-\ii n\kappa_0}{2}\sigma_3}\mathbf{Q}_d\left(\frac{\lambda-a_0}{\lambda-a_0^*}\right)^{\frac{1}{4}\sigma_3}\mathbf{Q}_d^{-1}
\ee^{\frac{\ii n\kappa_0-2\ii k_{3}(\infty)}{2}\sigma_3}\ee^{\frac{\ii \pi}{4}\sigma_3}\ee^{-(\ii k_3(\lambda)-\ii k_{3}(\infty))\sigma_3},
\end{equation}
and $k_3(\lambda)$ is defined as
\begin{equation}
	k_3(\lambda)=pR_{0}(\lambda)\int_{b_0}^{c_0}\frac{1}{R_{0}(\xi)(\xi-\lambda)}d\xi+p\log\left(\frac{\lambda-b_0}{\lambda-c_0}\right),
\end{equation}
then $k_{3}(\infty)$ can be calculated directly,
\begin{equation}\label{eq:mu-genus-zero}
	k_{3}(\infty)=\lim\limits_{\lambda\to\infty}k_3(\lambda)=-p\int_{b_0}^{c_0}\frac{1}{R_{0}(\xi)}d\xi.
\end{equation}

It is easy to see that the outer parametrix $\dot{\mathbf{T}}_{3}^{\rm out}(\lambda; \chi, \tau)$ has singularities at the points $\lambda=b_0, \lambda=c_0, \lambda=a_0, \lambda=a_{0}^{*}$.
Thus we should consider the inner parametrices at the neighbourhood of these points. Similar as \cite{Bilman-arxiv-2021,ling2022large},
the inner parametrices at the neighbourhood of these points $\lambda=b_0$ and $\lambda=c_0$ can be defined as
\begin{multline}
	\dot{\mathbf{T}}_{3}^{c_0}(\lambda; \chi, \tau)=\mathbf{K}_{3}(\lambda)n^{\ii p\sigma_3/2}\ee^{-\ii nh_0(c_0)\sigma_3}\mathbf{H}_{c_0}(\lambda)\mathbf{U}_{c_0}(\zeta_{c_0})\ee^{\ii nh_0(c_0)\sigma_3},\quad \lambda\in D_{c_0}(\delta),\\
	\mathbf{H}_{c_0}(\lambda)\equiv\mathbf{H}_{c_0}(\lambda; \chi, \tau):=(\lambda-b_0)^{\ii p\sigma_3}\left(\frac{f_{c_0}(\lambda)}{\lambda-c_0}\right)^{\ii p\sigma_3},
\end{multline}
\begin{equation}
	\dot{\mathbf{T}}_{3}^{b_0}(\lambda; \chi, \tau){:=}\left\{\begin{aligned}\mathbf{K}_{3}(\lambda)\mathbf{H}_{b_0}(\lambda)&n^{\ii p\sigma_3/2}\ee^{\ii nh_{0}(b_0)\sigma_3}\ii^{-\sigma_3}
\mathbf{U}_{b_0}(\zeta_{b_0})(-\ii\sigma_2)\ii^{-\sigma_3}\ee^{\ii n h_{0}(b_0)\sigma_3},\,\lambda\in D_{b_0,-}(\delta),\\
		\mathbf{K}_{3}(\lambda)\mathbf{H}_{b_0}(\lambda)&n^{\ii p\sigma_3/2}\ee^{\ii nh_{0}(b_0)\sigma_3}\ii^{-\sigma_3}\mathbf{U}_{b_0}(\zeta_{b_0})(-\ii\sigma_2)\ii^{-\sigma_3}\ee^{\ii n h_{0}(b_0)\sigma_3}\\
&\times\ee^{\frac{-\ii n\kappa_0}{2}\sigma_3}(\ii\sigma_2)\ee^{\frac{\ii n\kappa_0}{2}\sigma_3},\,\lambda\in D_{b_0,+}(\delta),\\
		\mathbf{H}_{b_0}(\lambda)&\equiv\mathbf{H}_{b_0}(\lambda; \chi, \tau):=\left(\frac{b_0-\lambda}{f_{b_0}(\lambda)}\right)^{\ii p\sigma_3}\left(c_0-\lambda\right)^{-\ii p\sigma_3}(\ii\sigma_2),
	\end{aligned}\right.
\end{equation}
where $f_{c_0}(\lambda):=f_{c_0}(\lambda; \chi, \tau), f_{b_0}(\lambda):=f_{b_0}(\lambda; \chi, \tau)$ are two conformal mappings defined at the neighbourhood of $\lambda=b_0$ and $\lambda=c_0$ respectively,
\begin{equation}
	f_{b_0}^2(\lambda)=2\left(h_0(b_0)-h_0(\lambda)\right), \quad f_{c_0}^2(\lambda)=2\left(h_0(\lambda)-h_0(c_0)\right).
\end{equation}
For convenience, we suppose $f'_{b_0}(b_0)=-\sqrt{-h_{0}''(b_0)}<0, f'_{c_0}(c_0)=\sqrt{h_0''(c_0)}>0$. It should be noticed that $h_0(\lambda)$ is discontinuous at the point $\lambda=b_0$,
and in this case, we choose the right value of the cut $\Sigma_{g_0}^{\pm}$ in the later analysis, that is  $h_0(b_0):=h_{0,-}(b_0)$.
The variables $\zeta_{b_0}, \zeta_{c_0}$ are defined as $\zeta_{b_0}=n^{1/2}f_{b_0}(\lambda), \zeta_{c_0}=n^{1/2}f_{c_0}(\lambda)$,
and the solution to $\mathbf{U}(\zeta)$ can be given with the parabolic cylinder function. The large-$\zeta$ asymptotics is
\begin{equation}\label{eq:U-genus-zero}
	\mathbf{U}(\zeta)\zeta^{\ii p\sigma_3}=\mathbb{I}+\frac{1}{2\ii\zeta}\begin{bmatrix}0&\alpha\\
		-\beta&0
	\end{bmatrix}+\begin{bmatrix}\mathcal{O}(\zeta^{-2})&\mathcal{O}(\zeta^{-3})\\
		\mathcal{O}(\zeta^{-3})&\mathcal{O}(\zeta^{-2})
	\end{bmatrix},\quad\zeta\to\infty,
\end{equation}
where
\begin{equation}\label{eq:alphabeta}
	\alpha=2^{\frac{3}{4}}\sqrt{2\pi}\Gamma\left(\frac{\ii\ln(2)}{2\pi}\right)^{-1}\ee^{\ii\pi/4}\ee^{\ii(\ln(2))^2/(2\pi)},\quad \beta=-\alpha^*.
\end{equation}
The inner parametrices at the local points $\lambda=a_0$ and $\lambda=a_0^*$ are related to the Airy function \cite{Bilman-arxiv-2021,ling2022large}, which can be defined as
$\dot{\mathbf{T}}_{3}^{a_0}(\lambda; \chi, \tau)$ and $\dot{\mathbf{T}}_{3}^{a_0^*}(\lambda; \chi, \tau)$ respectively.
Then the global parametrix for $\mathbf{T}_{3}(\lambda; \chi, \tau)$ can be defined by
\begin{equation}
	\dot{\mathbf{T}}_{3}(\lambda; \chi, \tau):=\left\{\begin{aligned}&\dot{\mathbf{T}}_{3}^{b_0}(\lambda; \chi, \tau), &\quad\lambda\in D_{b_0}(\delta),\\
		&\dot{\mathbf{T}}_{3}^{c_0}(\lambda; \chi, \tau),&\quad\lambda\in D_{c_0}(\delta),\\
		&\dot{\mathbf{T}}_{3}^{a_0}(\lambda; \chi, \tau),&\quad\lambda\in D_{a_0}(\delta),\\
		&\dot{\mathbf{T}}_{3}^{a_0^*}(\lambda; \chi, \tau)&\quad\lambda\in D_{a_0^*}(\delta),\\
		&\dot{\mathbf{T}}_{3}^{\rm out}(\lambda; \chi, \tau),&\quad\lambda\in \mathbb{C}\setminus\left(\overline{D_{b_0}(\delta)\cup D_{c_0}(\delta)\cup D_{a_0}(\delta)\cup D_{a_0^*}(\delta)}\cup \Sigma_{g_0}^{\pm}\cup I\right).\end{aligned}\right.
\end{equation}
Next,we will analyze the error between $\mathbf{T}_{3}(\lambda; \chi, \tau)$ and $\dot{\mathbf{T}}_{3}(\lambda; \chi, \tau)$.
\subsection{Error analysis}
Set the error function between $\mathbf{T}_{3}(\lambda; \chi, \tau)$ and $\dot{\mathbf{T}}_{3}(\lambda; \chi, \tau)$ as
\begin{equation}
	\mathcal{E}_{3}(\lambda; \chi, \tau):=\mathbf{T}_{3}(\lambda; \chi, \tau)\left(\dot{\mathbf{T}}_{3}(\lambda; \chi, \tau)\right)^{-1}.
\end{equation}
For convenience, denote $\mathbf{V}_{\mathcal{E}_{3}}(\lambda; \chi, \tau)$ as the jump matrix for $\mathcal{E}_{3}(\lambda; \chi, \tau)$ and $\Sigma_{\mathcal{E}_{3}}$ as the jump contours.
From the definition of $\dot{\mathbf{T}}_{3}(\lambda; \chi, \tau)$,
the jump matrices $\mathbf{V}_{\mathcal{E}_{3}}(\lambda; \chi, \tau)$ at the boundary of $D_{b_0}(\delta), D_{c_0}(\delta), D_{a_0}(\delta), D_{a_0^*}(\delta)$ equal to
\begin{equation}
	\begin{aligned}
		\mathbf{V}_{\mathcal{E}_{3}}(\lambda; \chi, \tau)&=\dot{\mathbf{T}}_{3}^{b_{0}, c_{0} , a_0, a_0^*}(\lambda; \chi, \tau)\left(\dot{\mathbf{T}}_{3}^{\rm out}(\lambda; \chi, \tau)\right)^{-1},\quad\lambda\in\partial D_{b_{0},c_{0}, a_0, a_0^*}(\delta).
	\end{aligned}
\end{equation}
If $\lambda\in\partial D_{b_0}(\delta)$ and $\lambda\in\partial D_{c_0}(\delta)$, $\mathbf{V}_{\mathcal{E}_{3}}(\lambda; \chi, \tau) $ are written as
\begin{multline}\label{eq:E2alpha2}
	\mathbf{V}_{\mathcal{E}_{3}}(\lambda; \chi, \tau)=\pmb{\mathcal{H}}_{b_0}(\lambda; \chi, \tau)\mathbf{U}_{b_0}(\zeta_{b_0})\zeta_{b_0}^{\ii p\sigma_3}\pmb{\mathcal{H}}_{b_0}(\lambda; \chi, \tau)^{-1},\quad\lambda\in\partial D_{b_0}(\delta),\\
	\pmb{\mathcal{H}}_{b_0}(\lambda; \chi, \tau):=\mathbf{K}_{3}(\lambda)\mathbf{H}_{b_0}(\lambda)n^{\ii p\sigma_3/2}\ee^{\ii nh_{0}(b_0)\sigma_3}\ii^{-\sigma_3},
\end{multline}
and
\begin{multline}\label{eq:E2beta2}
	\mathbf{V}_{\mathcal{E}_{3}}(\lambda; \chi, \tau)=\pmb{\mathcal{H}}_{c_0}(\lambda; \chi, \tau)\mathbf{U}_{c_0}(\zeta_{c_0})\zeta_{c_0}^{\ii p\sigma_3}\pmb{\mathcal{H}}_{c_0}(\lambda; \chi, \tau)^{-1},\quad\lambda\in\partial D_{c_0}(\delta),\\
	\pmb{\mathcal{H}}_{c_0}(\lambda; \chi, \tau):=\mathbf{K}_{3}(\lambda)\mathbf{H}_{c_0}(\lambda)n^{\ii p\sigma_3/2}\ee^{-\ii nh_{0}(c_0)\sigma_3}.
\end{multline}
If $\lambda\in\partial D_{a_0}(\delta)\cup\partial D_{a_0^*}(\delta)$, $\mathbf{V}_{\mathcal{E}_{3}}(\lambda; \chi, \tau)$ can be given in a similar formula with our previous article(Eq.(101) in \cite{ling2022large}).
From the asymptotic expression of $\mathbf{U}(\zeta)$ in Eq.\eqref{eq:U-genus-zero} and the estimation for $\lambda\in\partial D_{a_0}(\delta)\cup\partial D_{a_0^*}(\delta)$ in \cite{ling2022large}, the jump matrices $\mathbf{V}_{\mathcal{E}_{3}}(\lambda; \chi, \tau)$ satisfies the following estimation,
\begin{equation}\label{eq:error-est-genus-zero}
	\begin{aligned}
		\|\mathbf{V}_{\mathcal{E}_{3}}(\lambda; \chi, \tau)-\mathbb{I}\|&=\mathcal{O}\left(\ee^{-\mu_3 n}\right)\,\,(\mu_3>0),\, \lambda\in C_{L_{3,\Sigma}}^{\pm}\cup C_{R_{3,\Sigma}}^{\pm}\cup C_{R_{3,\Gamma}}^{\pm}\cup C_{L_{3,\Gamma}}^{\pm},\\
		\|\mathbf{V}_{\mathcal{E}_{3}}(\lambda; \chi, \tau)-\mathbb{I}\|&=\mathcal{O}(n^{-1}),\quad \lambda\in\partial D_{a_0,a_0^*}(\delta),\\
		\|\mathbf{V}_{\mathcal{E}_{3}}(\lambda; \chi, \tau)-\mathbb{I}\|&=\mathcal{O}(n^{-1/2}),\quad \lambda\in\partial D_{b_{0}, c_{0}}(\delta).\\
	\end{aligned}
\end{equation}
Under this case, the solution $q^{[n]}(n\chi, n\tau)$ can be recovered by
\begin{equation}\label{eq:qn-no-1}
	\begin{aligned}
		q^{[n]}(n\chi, n\tau)&=2\ii r\lim\limits_{\lambda\to\infty}\lambda\mathbf{T}_{3}(\lambda; \chi, \tau)_{12}\\
		&=2\ii r\lim\limits_{\lambda\to\infty}\lambda\left(\mathcal{E}_{3,11}(\lambda; \chi, \tau)\dot{\mathbf{T}}_{3,12}^{\rm out}(\lambda; \chi, \tau)+\mathcal{E}_{3,12}(\lambda; \chi, \tau)\dot{\mathbf{T}}_{3,22}^{\rm out}(\lambda; \chi, \tau)\right).
	\end{aligned}
\end{equation}
Moreover, we can simplify Eq.\eqref{eq:qn-no-1} into
\begin{equation}
	q^{[n]}(n\chi, n\tau)=2\ii r\lim\limits_{\lambda\to\infty}\lambda\left(\dot{\mathbf{T}}_{3,12}^{\rm out}(\lambda; \chi, \tau)+\mathcal{E}_{3,12}(\lambda; \chi, \tau)\right).
\end{equation}

Next, we will calculate the entry of $\mathcal{E}_{3,12}(\lambda; \chi, \tau)$ and then give the leading-order term for $q^{[n]}(n\chi, n\tau)$. From the error estimations in Eq.\eqref{eq:error-est-genus-zero}, we only calculate it for $\lambda\in\partial D_{b_{0}}(\delta)$ and $\lambda\in\partial D_{c_{0}}(\delta)$. In other contours, we omit the calculations.
When $\lambda\in\partial D_{b_0}(\delta)\cup\partial D_{c_0}(\delta)$, with the Plemelj formula, the solution to $\mathcal{E}_{3}(\lambda; \chi, \tau)$ is
\begin{equation}\label{eq:E2}
	\mathcal{E}_{3}(\lambda; \chi, \tau)=\mathbb{I}+\frac{1}{2\pi\ii}\int_{\partial D_{b_0}(\delta)\cup \partial D_{c_0}(\delta)}\frac{\mathcal{E}_{3,-}(\xi; \chi, \tau)(\mathbf{V}_{\mathcal{E}_{3}}(\xi; \chi, \tau)-\mathbb{I})}{\xi-\lambda}d\xi+\mathcal{O}\left(n^{-1}\right).
\end{equation}
When $\lambda\to\infty$, the asymptotic expansion is given by,
\begin{equation}
	\mathcal{E}_{3}(\lambda; \chi, \tau)=\mathbb{I}-\frac{1}{2\pi\ii}\sum_{j=1}^{\infty}\lambda^{-j}\int_{\partial D_{b_0}(\delta)\cup \partial D_{c_0}(\delta)}\mathcal{E}_{3,-}(\xi; \chi, \tau)(\mathbf{V}_{\mathcal{E}_{3}}(\xi; \chi, \tau)-\mathbb{I})\xi^{j-1}d\xi+\mathcal{O}\left(n^{-1}\right),\quad |\lambda|\to\infty.
\end{equation}
Then we further have
\begin{multline}
	\lim \limits_{\lambda\to\infty}\lambda\mathcal{E}_{3,12}(\lambda; \chi, \tau)=-\frac{1}{2\pi\ii}\Big[\int_{\partial D_{b_0}(\delta)\cup\partial D_{c_0}(\delta)}\mathcal{E}_{3,11,-}(\xi; \chi, \tau)\mathbf{V}_{\mathcal{E}_{3},12}(\xi; \chi, \tau)d\xi\\
+\int_{\partial D_{b_0}(\delta)\cup\partial D_{c_0}(\delta)}\mathcal{E}_{3,12,-}(\xi; \chi, \tau)\left(\mathbf{V}_{\mathcal{E}_{3},22}(\xi; \chi, \tau)-1\right)d\xi\Big]+\mathcal{O}\left(n^{-1}\right).
\end{multline}
From the definition of $\mathbf{V}_{\mathcal{E}_3}(\lambda; \chi, \tau)$ in Eq.\eqref{eq:E2alpha2} and Eq.\eqref{eq:E2beta2}, we simplify the potential $q^{[n]}(n\chi, n\tau)$ as
\begin{equation}
	q^{[n]}(n\chi, n\tau)=2\ii r\lim\limits_{\lambda\to\infty}\lambda\dot{\mathbf{T}}_{3,12}^{\rm out}(\lambda; \chi, \tau)-\frac{r}{\pi}\int_{\partial D_{b_0}(\delta)\cup\partial D_{c_0}(\delta)}\mathbf{V}_{\mathcal{E}_{3},12}(\xi; \chi, \tau)d\xi+\mathcal{O}\left(n^{-1}\right).
\end{equation}
For $\lambda\in\partial D_{b_0}(\delta)\cup\partial D_{c_0}(\delta),$ $\mathbf{V}_{\mathcal{E}_3,12}(\lambda; \chi, \tau)$ equals to
\begin{multline}
	\mathbf{V}_{\mathcal{E}_3,12}(\lambda; \chi, \tau)=-\ii \frac{n^{\ii p}\left(H_{c_0, 11}(\lambda)\right)^2\left(K_{3,11}(\lambda)\right)^2\ee^{-2\ii n h_0(c_0)}\alpha}{2n^{1/2}f_{c_0}(\lambda)}\\
-\ii \frac{n^{-\ii p}\left(H_{c_0, 22}(\lambda)\right)^2\left(K_{3,12}(\lambda)\right)^2\ee^{2\ii n h_0(c_0)}\beta}{2n^{1/2}f_{c_0}(\lambda)}+\mathcal{O}(n^{-1}),\quad \lambda\in \partial D_{c_0}(\delta),
\end{multline}
and
\begin{multline}
	\mathbf{V}_{\mathcal{E}_3,12}(\lambda; \chi, \tau)=\ii \frac{n^{\ii p}\left(H_{b_0,12,-}(\lambda)\right)^{-2}\left(K_{3,12, -}(\lambda)\right)^2\ee^{2\ii n h_{0,-}(b_0)}\alpha}{2n^{1/2}f_{b_0,-}(\lambda)}\\+\ii \frac{n^{-\ii p}
\left(H_{b_0, 12,-}(\lambda)\right)^2\left(K_{3,11,-}(\lambda)\right)^2\ee^{-2\ii n h_{0,-}(b_0)}\beta}{2n^{1/2}f_{b_0,-}(\lambda)}+\mathcal{O}(n^{-1}),\quad \lambda\in \partial D_{b_0}(\delta),
\end{multline}
where $\alpha$ and $\beta$ are defined as in Eq. \eqref{eq:alphabeta}.
By the residue theorem, we have
\begin{multline}\label{eq:V-res}
	-\frac{1}{\pi}\int_{\partial D_{b_0}(\delta)\cup\partial D_{c_0}(\delta)}\mathbf{V}_{\mathcal{E}_{3},12}(\xi; \chi, \tau)d\xi\\=\frac{n^{\ii p}\left(H_{b_0,12,-}(b_0)\right)^{-2}\left(K_{3,12, -}(b_0)\right)^2
\ee^{2\ii n h_{0,-}(b_0)}\alpha}{n^{1/2}\sqrt{-h''_{0, -}(b_0)}}+ \frac{n^{-\ii p}\left(H_{b_0,12,-}(b_0)\right)^2\left(K_{3,-,11}(b_0)\right)^2\ee^{-2\ii n h_{0,-}(b_0)}\beta}{n^{1/2}\sqrt{-h''_{0,-}(b_0)}}\\
+\frac{n^{\ii p}\left(H_{c_0,11}(c_0)\right)^2\left(K_{3,11}(c_0)\right)^2\ee^{-2\ii n h_{0}(c_0)}\alpha}{n^{1/2}\sqrt{h''_{0}(c_0)}}+ \frac{n^{-\ii p}\left(H_{c_0,22}(c_0)\right)^2\left(K_{3,12}(c_0)\right)^2\ee^{2\ii n h_{0}(c_0)}\beta}{n^{1/2}\sqrt{h_{0}''(c_0)}}+\mathcal{O}(n^{-1}).
\end{multline}
Substituting the entries of $\mathbf{H}_{b_0, c_0}(\lambda)$ and  $\mathbf{K}_{3}(\lambda)$ into Eq.\eqref{eq:V-res}, we can get the asymptotic expression in the genus-zero-down region as Eq.\eqref{eq:q-gen-zero},
\begin{multline}\label{eq:q-gen-zero}
	q^{[n]}(n\chi, n\tau)=r\ee^{2\ii k_{3}(\infty)-\ii n\kappa_0}\Bigg[\frac{\sqrt{2p}}{n^{1/2}\sqrt{-h''_{0,-}(b_0)}}\left(m_{-}^{b_0}\ee^{\ii\phi_{b_0}}-m_{+}^{b_0}\ee^{-\ii\phi_{b_0}}\right)\\
	+\frac{\sqrt{2p}}{n^{1/2}\sqrt{h''_{0}(c_0)}}\left(m_{+}^{c_0}\ee^{\ii\phi_{c_0}}-m_{-}^{c_0}\ee^{-\ii\phi_{c_0}}\right)-\ii {\Im (a_0)}\Bigg]+\mathcal{O}(n^{-1}),
\end{multline}
where
\begin{equation}\label{eq:para-genus-zero}
	\begin{aligned}
		 \phi_{b_0}&{=}\frac{\pi}{4}{+}\log(2)p{-}\arg\left(\Gamma\left(\ii p\right)\right){+}2k_{3,-}(b_0){+}2nh_{0,-}(b_0){+}p\log\left({-}nh_{0,-}''(b_0)(c_0-b_0)^2\right){-}n\kappa_0,\\
		 \phi_{c_0}&{=}\frac{\pi}{4}{+}\log(2)p{-}\arg\left(\Gamma\left(\ii p\right)\right){-}2k_3(c_0){-}2nh_{0}(c_0){+}p\log\left(nh_{0}''(c_0)(c_0-b_0)^2\right){+}n\kappa_0,\\
		m_{+}^{b_0}&{=}\frac{1}{2}{+}\frac{1}{4}\left(\sqrt{\frac{b_0{-}a_0}{b_0{-}a^*_0}}{+}\left(\sqrt{\frac{b_0{-}a_0}{b_0{-}a^*_0}}\right)^{-1}\right),\quad
		m_{-}^{b_0}{=}\frac{1}{2}{-}\frac{1}{4}\left(\sqrt{\frac{b_0{-}a_0}{b_0{-}a^*_0}}{+}\left(\sqrt{\frac{b_0{-}a_0}{b_0{-}a^*_0}}\right)^{-1}\right),\\
		m_{+}^{c_0}&{=}\frac{1}{2}{+}\frac{1}{4}\left(\sqrt{\frac{c_0{-}a_0}{c_0{-}a^*_0}}{+}\left(\sqrt{\frac{c_0{-}a_0}{c_0{-}a^*_0}}\right)^{-1}\right),\quad
		m_{-}^{c_0}{=}\frac{1}{2}{-}\frac{1}{4}\left(\sqrt{\frac{c_0{-}a_0}{c_0{-}a^*_0}}{+}\left(\sqrt{\frac{c_0{-}a_0}{c_0{-}a^*_0}}\right)^{-1}\right).
	\end{aligned}
\end{equation}
Similarly, by choosing one fixed $\tau$ in this region, we numerically verify the consistency between the exact solution and the asymptotic solution in Fig.\ref{fig:genus-zero}.
\begin{figure}[ht]
	\centering
	\includegraphics[width=1\textwidth]{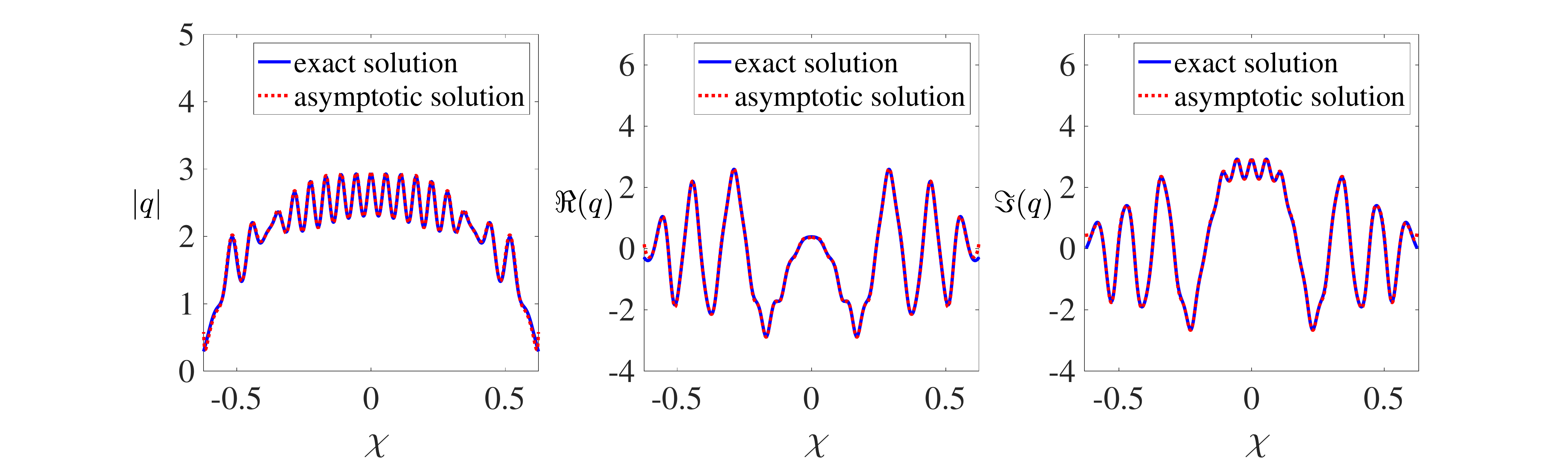}
	\caption{The comparison between the exact solution ($20$-th order KMBs) and its asymptotic solution in the genus-zero-down region by choosing $\tau=\frac{1}{10}$ (as shown by the green dashed line
in the middle panel in Fig. \ref{fig:KM}), $c_1=-c_2=1$. The left one is the modulus of $q^{[n]}(n\chi, n\tau)$, and the middle and right panels are the real and imaginary parts of $q^{[n]}(n\chi, n\tau)$ respectively.}
	\label{fig:genus-zero}
\end{figure}

In the next section, we will analyze the asymptotics in the algebraic-decay region.
\section{The algebraic-decay region}
\label{sec:al-decay}
In the above discussion, we studied the large-order asymptotics for four different regions. To give the leading-order term, we introduce four kinds of $g$-functions and modify the original phase term $\vartheta(\lambda; \chi, \tau)$ into a new one. In the algebraic-decay region, the original phase term $\vartheta(\lambda; \chi, \tau)$ has three real critical points, which can be used as the controlling phase term.
By choosing one fixed $\chi$ and $\tau$, we give the sign chart of $\Im(\vartheta(\lambda; \chi, \tau))$ in Fig.\ref{fig:al-decay}.
\begin{figure}[ht]
	\centering
	\includegraphics[width=0.45\textwidth]{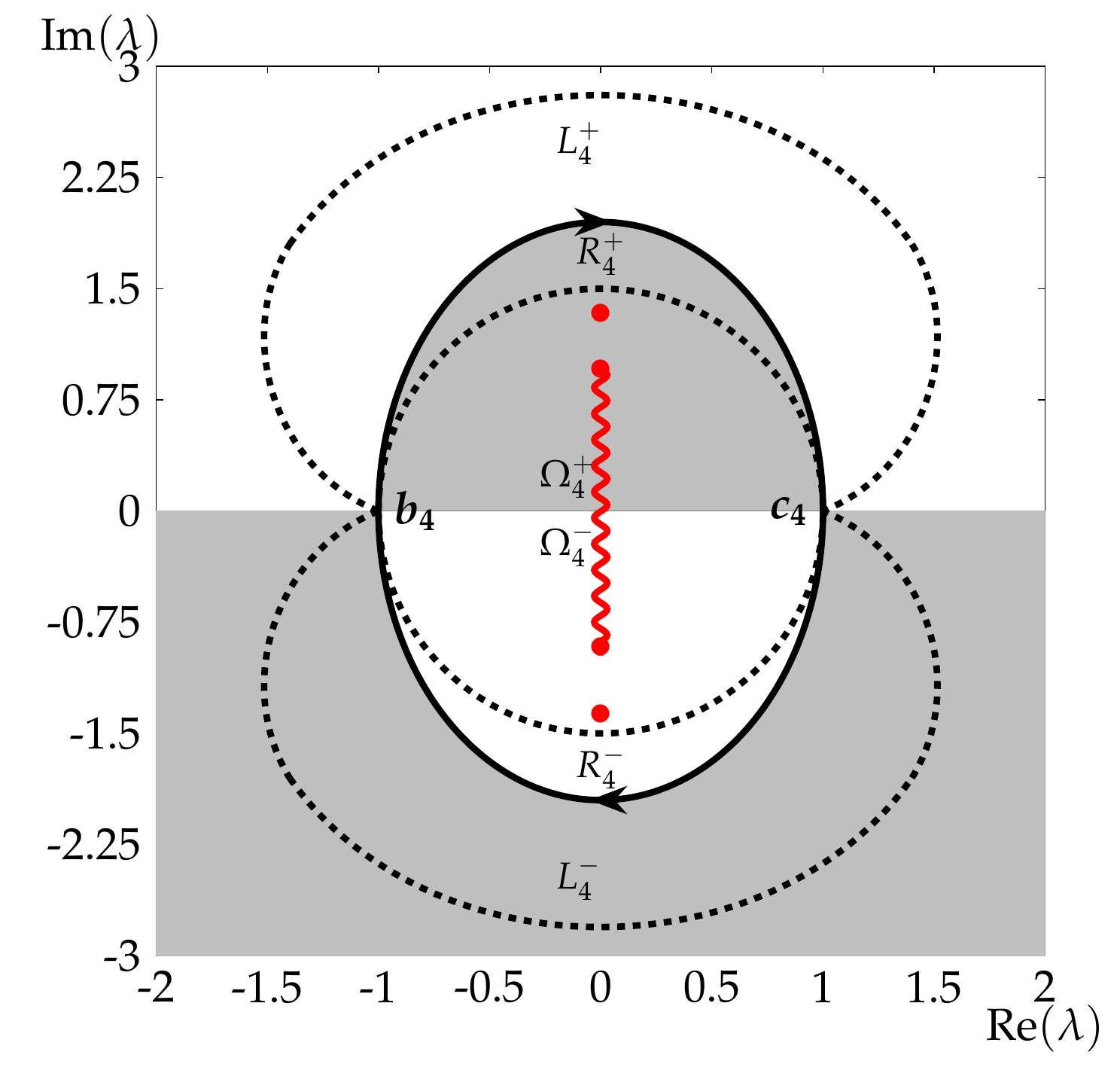}
	\centering
	\includegraphics[width=0.45\textwidth]{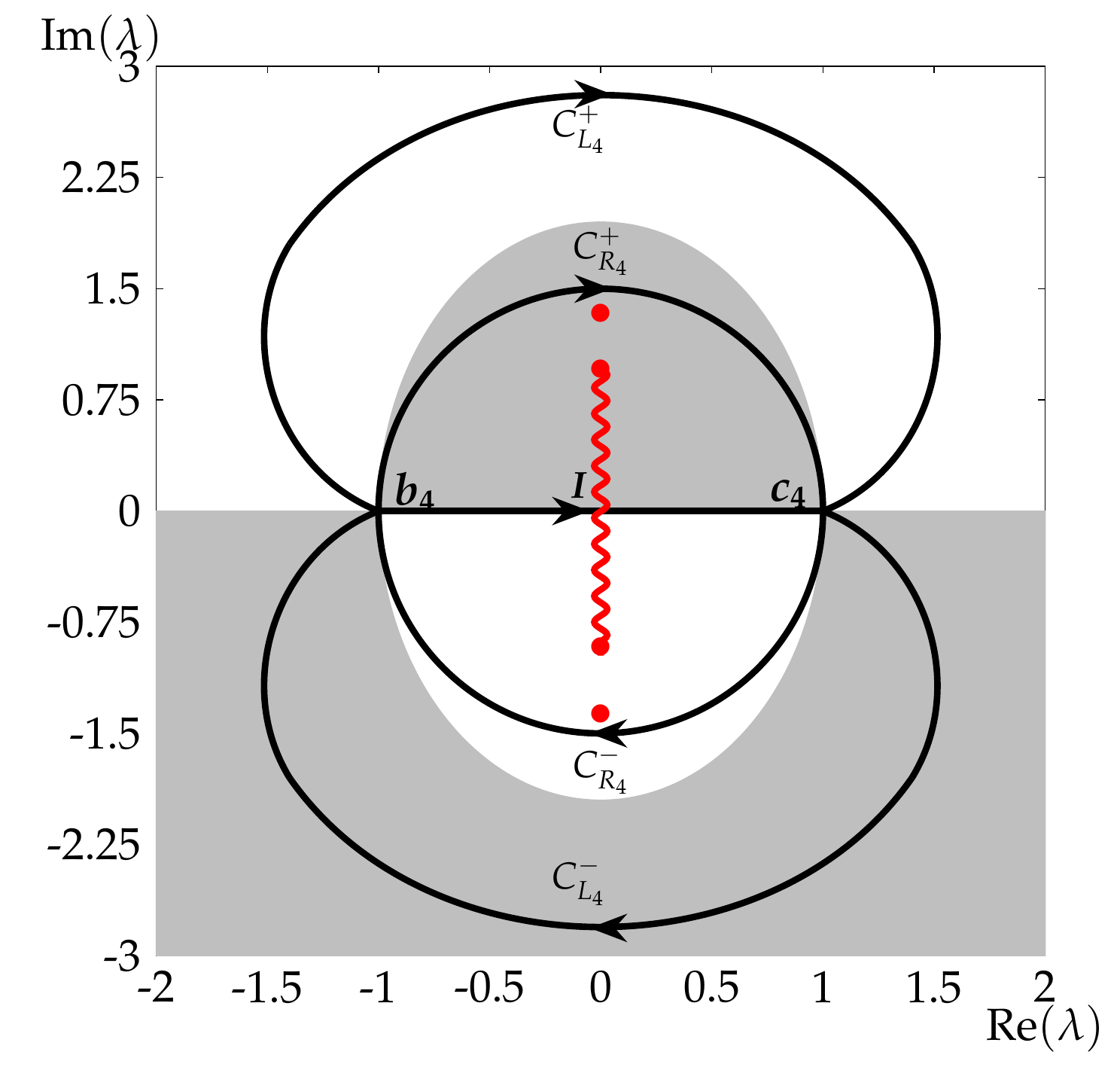}
	\caption{The sign chart of ${\Im}\left(\vartheta\left(\lambda; \frac{1}{2}, \frac{1}{200}\right)\right)$ in the algebraic-decay region, where ${\Im}\left(\vartheta\left(\lambda; \frac{1}{2}, 0\right)\right)<0$
(shaded) and ${\Im}\left(\vartheta\left(\lambda; \frac{1}{2}, \frac{1}{200}\right)\right)>0$ (unshaded).
 The left one gives the original jump contour for $\mathbf{S}_{4}(\lambda; \chi, \tau)$, and the right one is the corresponding jump contour after deformation.}
	\label{fig:al-decay}
\end{figure}

Similarly, introduce the matrix $\mathbf{S}_{4}(\lambda; \chi, \tau)$ defined by
\begin{equation}
	\mathbf{S}_{4}(\lambda; \chi, \tau):=\left\{\begin{aligned}&\mathbf{M}^{[n]}(\lambda; \chi, \tau)\ee^{-\ii n\vartheta(\lambda; \chi, \tau)\sigma_3}\mathbf{Q}_{d}^{-1}\ee^{\ii n\vartheta(\lambda; \chi, \tau)\sigma_3},\quad\lambda\in D_{0}\cap\left(D_{4}^{+}\cup D_{4}^{-}\right)^{c},\\
		&\mathbf{M}^{[n]}(\lambda; \chi, \tau),\quad\quad\quad\quad\quad\quad\quad\quad\quad\quad\quad\quad\quad\quad\quad\quad\quad {\rm otherwise},\end{aligned}\right.
\end{equation}
where $D_{4}^{\pm}=\Omega_4^{\pm}\cup R_4^{\pm}$. Next, set the sectional analytic matrices $\mathbf{T}_{4}(\lambda; \chi, \tau)$ as follows,
\begin{equation}
	\begin{aligned}
		\mathbf{T}_{4}(\lambda; \chi, \tau):&=\mathbf{S}_{4}(\lambda; \chi, \tau)\ee^{-\ii n\vartheta(\lambda; \chi, \tau)\sigma_3}
		\left(\mathbf{Q}_{R}^{[2]}\right)^{-1}\ee^{\ii n\vartheta(\lambda; \chi, \tau)\sigma_3},\quad &\lambda\in L_4^{+},\\
		\mathbf{T}_{4}(\lambda; \chi, \tau):&=\mathbf{S}_{4}(\lambda; \chi, \tau)\mathbf{Q}_{L}^{[2]}\ee^{-\ii n\vartheta(\lambda; \chi, \tau)\sigma_3}
		\mathbf{Q}_{C}^{[2]}\ee^{\ii n\vartheta(\lambda; \chi, \tau)\sigma_3},\quad &\lambda\in R_{4}^{+},\\
		\mathbf{T}_{4}(\lambda; \chi, \tau):&=\mathbf{S}_{4}(\lambda; \chi, \tau)\mathbf{Q}_{L}^{[2]},\quad &\lambda\in \Omega^{+}_{4},\\
		\mathbf{T}_{4}(\lambda; \chi, \tau):&=\mathbf{S}_{4}(\lambda; \chi, \tau)\ee^{-\ii n\vartheta(\lambda; \chi, \tau)\sigma_3}
		\left(\mathbf{Q}_{R}^{[1]}\right)^{-1}\ee^{\ii n\vartheta(\lambda; \chi, \tau)\sigma_3},\quad &\lambda\in L_{4}^{-},\\
		\mathbf{T}_{4}(\lambda; \chi, \tau):&=\mathbf{S}_{4}(\lambda; \chi, \tau)\mathbf{Q}_{L}^{[1]}\ee^{-\ii n\vartheta(\lambda; \chi, \tau)\sigma_3}
		\mathbf{Q}_{C}^{[1]}
		\ee^{\ii n\vartheta(\lambda; \chi, \tau)\sigma_3},\quad &\lambda\in R_{4}^{-},\\
		\mathbf{T}_{4}(\lambda; \chi, \tau):&=\mathbf{S}_{4}(\lambda; \chi, \tau)\mathbf{Q}_{L}^{[1]},\quad &\lambda\in \Omega_{4}^{-},\\
\mathbf{T}_{4}(\lambda; \chi, \tau):&=\mathbf{S}_{4}(\lambda; \chi, \tau),\quad &\text{otherwise}.
	\end{aligned}
\end{equation}
By a direct calculation, the jump conditions of $\mathbf{T}_{4}(\lambda; \chi, \tau)$ change into,
\begin{equation}
	\begin{aligned}
		\mathbf{T}_{4,+}(\lambda; \chi, \tau)&=\mathbf{T}_{4,-}(\lambda; \chi, \tau)\ee^{-\ii n\vartheta(\lambda; \chi, \tau)\sigma_3}\mathbf{Q}_{R}^{[2]}\ee^{\ii n\vartheta(\lambda; \chi, \tau)\sigma_3}, \quad&\lambda\in C_{L_4}^{+},\\
		\mathbf{T}_{4,+}(\lambda; \chi, \tau)&=\mathbf{T}_{4,-}(\lambda; \chi, \tau)\ee^{-\ii n\vartheta(\lambda; \chi, \tau)\sigma_3}\mathbf{Q}_{C}^{[2]}\ee^{\ii n\vartheta(\lambda; \chi, \tau)\sigma_3}, \quad&\lambda\in C_{R_4}^{+},\\
		\mathbf{T}_{4,+}(\lambda; \chi, \tau)&=\mathbf{T}_{4,-}(\lambda; \chi, \tau)\ee^{-\ii n\vartheta(\lambda; \chi, \tau)\sigma_3}\mathbf{Q}_{R}^{[1]}\ee^{\ii n\vartheta(\lambda; \chi, \tau)\sigma_3}, \quad&\lambda\in C_{L_4}^{-},\\
		\mathbf{T}_{4,+}(\lambda; \chi, \tau)&=\mathbf{T}_{4,-}(\lambda; \chi, \tau)\ee^{-\ii n\vartheta(\lambda; \chi, \tau)\sigma_3}\mathbf{Q}_{C}^{[1]}\ee^{\ii n\vartheta(\lambda; \chi, \tau)\sigma_3}, \quad&\lambda\in C_{R_4}^{-},\\
		\mathbf{T}_{4,+}(\lambda; \chi, \tau)&=\mathbf{T}_{4,-}(\lambda; \chi, \tau)2^{\sigma_{3}},\quad&\lambda\in I.
	\end{aligned}
\end{equation}
From the sign chart of $\Im(\vartheta(\lambda; \chi, \tau))$ in Fig.\ref{fig:al-decay}, when $n$ is large, the jump conditions will exponentially decay into the identity matrix except for the contour $I=\left[b_4, c_4\right]$. In the next subsection, we will construct the parametrix for $\mathbf{T}_{4}(\lambda; \chi, \tau)$.
\subsection{Parametrix construction for $\mathbf{T}_{4}(\lambda; \chi, \tau)$} From the jump matrix in the contour $I=\left[b_4, c_4\right]$, we directly give the outer parametrix $\dot{\mathbf{T}}_{4}^{\rm out}(\lambda; \chi, \tau)$ as
\begin{equation}\label{eq:T1out}
	\dot{\mathbf{T}}_{4}^{\rm out}(\lambda; \chi, \tau)=\left(\frac{\lambda-b_4}{\lambda-c_4}\right)^{\ii p\sigma_3}, \quad p=\frac{\log\left(2\right)}{2\pi}, \quad \lambda\in\mathbb{C}\setminus I.
\end{equation}
Following the calculation in the genus-zero-down region, the inner parametrices at the neighborhood of $b_{4}$ and $c_{4}$ can be constructed as
\begin{equation}
	\begin{aligned}
		\dot{\mathbf{T}}_{4}^{b_4}(\lambda; \chi, \tau):&{=}n^{{-}\ii p\sigma_3/2}\ee^{{-}\ii n\vartheta(b_4; \chi, \tau)\sigma_3}\mathbf{H}_{b_4}(\lambda; \chi, \tau)\mathbf{U}_{b_4}(\zeta_{b_4})({-}\ii\sigma_2)\ee^{\ii n\vartheta(b_4; \chi, \tau)\sigma_3},\,\lambda\in D_{b_4}(\delta),\\
		\dot{\mathbf{T}}_{4}^{c_4}(\lambda; \chi, \tau):&=n^{\ii p\sigma_3/2}\ee^{-\ii n\vartheta(c_4; \chi, \tau)\sigma_3}\mathbf{H}_{c_4}(\lambda; \chi, \tau)\mathbf{U}_{c_4}(\zeta_{c_4})\ee^{\ii n\vartheta(c_4; \chi, \tau)\sigma_3},\, \lambda\in D_{c_4}(\delta),
	\end{aligned}
\end{equation}
where
\begin{equation}
	\mathbf{H}_{b_4}(\lambda; \chi, \tau):=\left(\frac{b_4-\lambda}{f_{b_4}(\lambda; \chi, \tau)}\right)^{\ii p\sigma_3}(c_4-\lambda)^{-\ii p\sigma_3}\left(\ii\sigma_2\right),\quad \mathbf{H}_{c_4}(\lambda; \chi, \tau):=(\lambda-b_4)^{\ii p\sigma_3}
\left(\frac{f_{c_4}(\lambda; \chi, \tau)}{\lambda-c_4}\right)^{\ii p\sigma_3},
\end{equation}
and $f_{b_4}(\lambda; \chi, \tau)$ and $f_{c_4}(\lambda; \chi, \tau)$ are two conformal mappings defined as
\begin{equation}\label{eq:con-map}
	f_{b_4}(\lambda; \chi, \tau)^2=2\left[\vartheta(b_4, \chi, \tau)-\vartheta(\lambda; \chi, \tau)\right],\quad f_{c_4}(\lambda; \chi, \tau)^2=2\left[\vartheta(\lambda; \chi, \tau)-\vartheta(c_4; \chi, \tau)\right].
\end{equation}
In this case, we still choose the root such that $f'_{b_4}(b_4; \chi, \tau)=-\sqrt{-\vartheta''(b_4; \chi, \tau)}<0, f'_{c_4}(c_4; \chi, \tau)=\sqrt{\vartheta''(c_4; \chi, \tau)}>0.$  $\zeta_{b_4}$ and $\zeta_{c_4}$ are two variables defined by
$\zeta_{b_4}=n^{1/2}f_{b_4}(\lambda; \chi, \tau), \zeta_{c_4}=n^{1/2}f_{c_4}(\lambda; \chi, \tau)$.

Here $\mathbf{U}(\zeta)$ is the same as the definition in the genus-zero-down region, and its asymptotics also satisfies Eq.\eqref{eq:U-genus-zero}.

Then the global parametrix of $\mathbf{T}_{4}(\lambda; \chi, \tau)$ is
\begin{equation}
	\dot{\mathbf{T}}_{4}(\lambda; \chi, \tau):=\left\{\begin{aligned}&\dot{\mathbf{T}}_{4}^{b_4}(\lambda; \chi, \tau),&\quad\lambda\in D_{b_4}(\delta),\\
		&\dot{\mathbf{T}}_{4}^{c_4}(\lambda; \chi, \tau),&\quad\lambda\in D_{c_4}(\delta),\\
		&\dot{\mathbf{T}}_{4}^{\rm out}(\lambda; \chi, \tau),&\quad\lambda\in \mathbb{C}\setminus\left(I\cup\overline{D_{b_4}(\delta)}\cup \overline{D_{c_4}(\delta)}\right).\\
	\end{aligned}\right.
\end{equation}
Next we will analyze the error between $\mathbf{T}_{4}(\lambda; \chi, \tau)$ and the parametrix $\dot{\mathbf{T}}_{4}(\lambda; \chi, \tau)$.
\subsection{Error analysis}
Without loss of generality, set the error function between $\mathbf{T}_{4}(\lambda; \chi, \tau)$ and $\dot{\mathbf{T}}_{4}(\lambda; \chi, \tau)$ as $\mathcal{E}_4(\lambda; \chi, \tau)$, that is
\begin{equation}
	\mathcal{E}_4(\lambda; \chi, \tau):=\mathbf{T}_{4}(\lambda; \chi, \tau)\left(\dot{\mathbf{T}}_{4}(\lambda; \chi, \tau)\right)^{-1},
\end{equation}
the jump matrix of $\mathcal{E}_4(\lambda; \chi, \tau)$ can be set as $\mathbf{V}_{\mathcal{E}_4}(\lambda; \chi, \tau)$, the corresponding contours are set as $\Sigma_{\mathcal{E}_4}$.
From the definition of $\dot{\mathbf{T}}_{4}(\lambda; \chi, \tau)$, we have the following estimation for the jump matrix $\mathbf{V}_{\mathcal{E}_4}(\lambda; \chi, \tau)$,
\begin{equation}
\begin{aligned}
	\|\mathbf{V}_{\mathcal{E}_4}(\lambda; \chi, \tau)-\mathbb{I}\|&=\mathcal{O}(\ee^{-\mu_4n})(\mu_4>0),\quad \lambda\in \left(C_{L_4}^{\pm}\cup C_{R_4}^{\pm}\right)\cap \Sigma_{\mathcal{E}_4},\\
\|\mathbf{V}_{\mathcal{E}_4}(\lambda; \chi, \tau)-\mathbb{I}\|&=\mathcal{O}(n^{-1/2}),\quad \lambda\in \partial D_{b_4}(\delta)\cup \partial D_{c_4}(\delta)  .
\end{aligned}
\end{equation}

Similar to the genus-zero-down region, the potential $q^{[n]}(n\chi, n\tau)$ can be recovered from $\mathbf{T}_{4}(\lambda; \chi, \tau)$ by the following formula,
\begin{equation}\label{eq:qn-al-1}
	\begin{aligned}
		q^{[n]}(n\chi, n\tau)&=2\ii r\lim\limits_{\lambda\to\infty}\lambda \mathbf{T}_{4}(\lambda; \chi, \tau)_{12}
		\\&=2\ii r\lim\limits_{\lambda\to\infty}\lambda\left(\mathcal{E}_{4,11}(\lambda; \chi, \tau)\dot{\mathbf{T}}_{4,12}^{\rm out}(\lambda; \chi, \tau)+\mathcal{E}_{4,12}(\lambda; \chi, \tau)\dot{\mathbf{T}}_{4,22}^{\rm out}(\lambda; \chi, \tau)\right).
	\end{aligned}
\end{equation}
Moreover, we have
\begin{equation}
	q^{[n]}(n\chi, n\tau)=-r\frac{1}{\pi}\int_{\partial D_{b_4}(\delta)\cup\partial D_{c_4}(\delta)}\mathbf{V}_{\mathcal{E}_4,12}(\xi; \chi, \tau)d\xi+\mathcal{O}(n^{-3/2}).
\end{equation}
Through a similar calculation with the genus-zero-down region, the asymptotic expression in the algebraic-decay region is given by,
\begin{multline}\label{eq:al}
	q^{[n]}(n\chi, n\tau){=}r\frac{\sqrt{2p}}{n^{1/2}}\left[\frac{\ee^{{-}2\ii n\vartheta(b_4; \chi, \tau)}\left({-}\vartheta''(b_4; \chi, \tau)\right)^{{-}\ii p}}{\sqrt{{-}\vartheta''(b_4; \chi, \tau)}}\ee^{\ii\varphi(\chi, \tau)}{+}\frac{\ee^{{-}2\ii n\vartheta(c_4; \chi, \tau)}
\vartheta''(c_4; \chi, \tau)^{\ii p}}{\sqrt{\vartheta''(c_4; \chi, \tau)}}\ee^{{-}\ii\varphi(\chi, \tau)}\right]\\+\mathcal{O}(n^{-3/2}),
\end{multline}
where
$$\varphi(\chi, \tau)=-p\ln(n)-2p\ln\left(c_4-b_4\right)-p\log(2)-\frac{1}{4}\pi+\arg\left(\Gamma\left(\ii p\right)\right).$$
In this case, we choose $\tau=\frac{1}{200}$ and give the comparison between the exact solution and the asymptotic
solution with the genus-zero-infinity region together, which is shown
in Fig. \ref{fig:al}. It is seen that they are fitting very well.

\section{Conclusions and Discussions}
In this paper, we analyze the large-order asymptotics for KMBs of the NLS equation under two constraints to the vector constant $\mathbf{c}=[c_1, c_2]^{\T}$, one of which is $c_1=c_2$ and the other case is $c_1=-c_2$. In the far-field regime, the $(\chi, \tau)$ space-time plane can be partitioned into five distinct regions. Compared to the large-order asymptotics of solitons, the phase term appearing in the RHP for the KMBs involves an additional factor $\frac{\ii r}{4n}\log\left(\frac{\lambda-\ii}{\lambda+\ii}\right)$, which produces a new cut on $[-\ii, \ii]$ and brings new difficulties to study the asymptotics. Due to this new term, a genus-two region appears under the large-order asymptotics, which was not reported in the previous studies of high-order solitons and rogue waves \cite{Bilman-JDE-2021,Bilman-arxiv-2021,ling2022large}.

Up to now, through the known results in the literature \cite{Bilman-JDE-2021,Bilman-arxiv-2021,ling2022large}, we can obtain uniform insights for large-order localized waves in both the zero and non-zero backgrounds. In the far-field regime, under the zero background, the single high-order solitons have four distinct asymptotic regions and the high-order breathers have five asymptotic regions, where the new additional region is the genus-three. Similarly, under the non-zero background, compared to the high-order rogue waves there appears a new genus-two region for the large-order KMBs. In our previous studies \cite{ling2022large}, we conjectured that, under the zero background, if there are $l$ spectral parameters with the same real part, we will get a genus-$2l-1$ region under the large-order limit. We guess that for the high-order KMBs, if we construct the general high-order KMBs with two distinct spectra, namely, the phase term $\vartheta(\lambda; \chi, \tau)$ will be modified as follows,
\begin{equation}\label{eq:hatvar}
\hat{\vartheta}(\lambda; \chi, \tau)=\lambda\chi+\lambda^2\tau+\frac{\ii}{2}\log\left(\frac{\lambda-\lambda_1}{\lambda-\lambda_1^*}\right)
+\frac{\ii}{2}\log\left(\frac{\lambda-\lambda_2}{\lambda-\lambda_2^*}\right)
+\frac{\ii}{4n}\log\left(\frac{\lambda-\ii}{\lambda+\ii}\right),
\end{equation}
where $\lambda_1=\alpha_1\ii, \lambda_2=\alpha_2\ii, \alpha_1, \alpha_2>1, \alpha_1\neq \alpha_2$, and then we can get six asymptotic regions which involve a genus-four region. 

In future work, we would like to give a detailed asymptotic analysis of the corresponding RHP with the modified phase term Eq.\eqref{eq:hatvar}. Moreover, we will generalize the analysis to the general $l$ spectral parameters for the high-order KMBs. As we conjectured, a genus-$2l$ region maybe appear in the center part of asymptotic regions.

\section*{Acknowledgements}

Liming Ling is supported by the National Natural Science Foundation of China (Grant No. 12122105); Xiaoen Zhang is supported by the National Natural Science Foundation of China (Grant No.12101246),
the China Postdoctoral Science Foundation (Grant No. 2020M682692), the Guangzhou Science and Technology Program of China(Grant No. 202102020783).
\bibliographystyle{siam}
\bibliography{reference}

\end{document}